\newcommand{\pointsize}{11pt}

\documentclass[oneside, \pointsize]{amsbook}

\usepackage[
includehead,
includefoot,
left = 1in, 
top = 1in, 
right = 1in,
bottom = 1in
]{geometry}
\usepackage{fancyhdr}
\usepackage{setspace}
\usepackage{calc}
\usepackage[nocompress]{cite} %%% optional
\usepackage[pdfborder={0 0 0}, pdfpagemode=UseNone, pdfstartview=FitH]{hyperref} %%% optional

\fancyheadoffset[R]{0.5in} 

\fancypagestyle{prelim}{%    
	\renewcommand{\headrulewidth}{0pt} 
	\fancyhf{}           
	\pagenumbering{roman}    
	\cfoot{-\thepage-}       
}

\fancypagestyle{maintext}{%
	\renewcommand{\headrulewidth}{0.4pt}
	\pagenumbering{arabic}
	\fancyhf{}
	\fancyhead[L]{\rightmark}
	\cfoot{\thepage}%%%
}

\fancypagestyle{abstract1}{%
	\renewcommand{\headrulewidth}{0.4pt}
	\fancyheadoffset[R]{0in}
	\pagenumbering{arabic}
	\fancyhf{}
}

\fancypagestyle{abstract2}{%
	\renewcommand{\headrulewidth}{0.4pt}
	\fancyheadoffset[R]{0in}
	\pagenumbering{arabic}
	\fancyhf{}
	\rhead{-\thepage-}
}

\numberwithin{figure}{chapter} 
\numberwithin{table}{chapter}
\numberwithin{equation}{chapter}
\numberwithin{section}{chapter}

\usepackage{graphicx}
\usepackage{caption} %%% optional but probably necessary
\usepackage{subcaption} %%% optional but probably necessary

\usepackage{color}
\usepackage{amssymb}
\usepackage{pgf,tikz}\usepackage{mathrsfs}\usetikzlibrary{arrows}

\usepackage{mathptmx}

\newtheorem{thm}{Theorem}[section]
\newtheorem{cor}[thm]{Corollary}
\newtheorem{lemma}[thm]{Lemma}
\newtheorem{prop}[thm]{Proposition}

\newtheorem{defn}[thm]{Definition}
\newtheorem{assumption}[thm]{Assumption}

\newtheorem{rem}[thm]{Remark}

\newtheorem{ex}{Example}

\newcommand{\ra}{\rightarrow}

\newcommand{\abs}[1]{\left\lvert #1 \right\rvert}
\newcommand{\ket}[1]{\lvert #1 \rangle}
\newcommand{\bra}[1]{\langle #1 \rvert}
\newcommand{\wslim}{{\rm w}^*\mbox{-}\lim}
\newcommand{\wlim}{{\rm w}\mbox{-}\lim}
\newcommand{\conv}[1]{{\rm Conv}\left( #1 \right)}

\newcommand{\RR}{\mathbb R}

\newcommand{\CC}{\mathbb C}
\newcommand{\ZZ}{\mathbb Z}

\newcommand{\NN}{\mathbb N}

\newcommand{\calA}{\mathcal A}
\newcommand{\calB}{\mathcal B}
\newcommand{\calG}{\mathcal G}

\newcommand{\calH}{\mathcal H}
\DeclareMathOperator{\Tr}{Tr}

\begin{document}
   \frontmatter

   \pagestyle{prelim}
   
   % Redefine plain page style so that the first pages of chapters
   % have desired page style.
   %
   \fancypagestyle{plain}{%
      \fancyhf{}
      \cfoot{-\thepage-}
   }%
\begin{center}
	\null\vfill
	{%
		{\huge  Topologically Ordered States \\ in Infinite Quantum Spin Systems}
	}%
	\\
	\bigskip
	By \\
	\bigskip
	Matthew M. Cha \\
	\bigskip
	B.S. (University of California, San Diego) 2011 \\
	\bigskip
	DISSERTATION \\
	\bigskip
	Submitted in partial satisfaction of the requirements for the
	degree of \\
	\bigskip
	DOCTOR OF PHILOSOPHY \\
	\bigskip
	in \\
	\bigskip
	MATHEMATICS \\
	\bigskip
	in the \\
	\bigskip
	OFFICE OF GRADUATE STUDIES \\
	\bigskip        
	of the \\
	\bigskip
	UNIVERSITY OF CALIFORNIA \\
	\bigskip
	DAVIS \\
	\bigskip
	Approved: \\
	\bigskip
	\bigskip
	\makebox[3in]{\hrulefill} \\
	Bruno Nachtergaele (Chair) \\
	\bigskip
	\bigskip
	\makebox[3in]{\hrulefill} \\
	Greg Kuperberg\\
	\bigskip
	\bigskip
	\makebox[3in]{\hrulefill} \\
	Craig Tracy \\
	\bigskip
	Committee in Charge \\
	\bigskip
	2017 \\
	\vfill
\end{center}

   \newpage
	
	 %%% (optional) copyright page <== this page is not numbered!
	 \thispagestyle{empty}
	 \begin{titlepage}
	 	\begin{center}{ \Large
	 		To my parents,\\
	 		\vspace{1ex}
	 		 Nkias Pov Tsab and Mim Faaj, \\
	 		 \vspace{1ex}
	 		for your unconditional love.}
	 	\end{center}
	 \vspace*{45em}
	 \begin{center}
		 \copyright \ Matthew M.\ Cha, 2017.  All rights reserved.  
	 \end{center}
	 \end{titlepage}
	 \newpage
	 \stepcounter{page}
	
	 %%% (optional) dedication page
	 \thispagestyle{empty}

	 \newpage
   
   % Begin Double Spacing
   %
   \doublespacing
   
   \tableofcontents
   \newpage
   
   {\singlespacing
   	\begin{flushright}
   		Matthew M. Cha \\
   		June 2017 \\
   		Mathematics \\
   	\end{flushright}
   }
   
   \bigskip
   
   \begin{center}
   	Topologically Ordered States in Infinite Quantum Spin Systems \\
   \end{center}
   
   \section*{Abstract}
   
   This dissertation discusses some properties of topologically ordered states as they appear in the setting of infinite quantum spin systems. 
   We will focus attention on the quasi-particle charges that may arise as elementary excitations in these models.  
   In planar systems, one indication of topological order is that the quasi-particle have braided statistics.
   We call these \emph{anyons}, that is, a process of braiding one quasi-particle around another may result in a factor of \emph{any} phase or even a unitary operation on the initial quantum state.
   
   The exposition naturally splits into two parts; preliminaries and  results.
   The preliminary part consists of Chapters 2 and 3 while the results are contained in Chapters 4 and 5.
   In Chapter 2, we begin by giving a brief review of infinite quantum spin systems as $C^*$-dynamical systems.
   The dynamics are determined by an interaction map and its corresponding local Hamiltonians.
   A key technical tool for studying the dynamics is the Lieb-Robinson bound \cite{LiebR, NachtergaeleSLR, HastingsKLR}.
   This gives an estimate for the speed at which the support of a local observable may grow up to exponentially small errors.
   The Lieb-Robinson bound may be thought of playing a role analogous to the speed of light in relativistic quantum  theories, and 
   has been foundational to many modern results in quantum spin systems, such as the automorphic equivalence in gapped phases \cite{BachmannMNS} and stability of the spectral gap in frustration-free Hamiltonians\cite{BravyiHM,MichalakisZ,NachSYffs}.
   
   The primary example in our analysis of topological order is a planar quantum spin system introduced by Kitaev \cite{KitaevQD}.
   In Chapter 3, we define the Kitaev quantum double models on the bond set of the planar square lattice 
   and compute the ground state degeneracy in the finite volume.
   Elementary excitation arise by application of ribbon operators to a ground state.  
   The mutual statistics of these excitations are braided and completely described by the representation theory of the quantum double for a finite group $G$, $\operatorname{Rep}(\mathcal{D}(G))$.
   Although the quantum double models are exactly solvable in the finite volume, 
   there are relatively few results regarding the thermodynamic limit \cite{AlickiFH, Naaijkens11, FiedlerN}.

   In the second part, we discuss how elementary excitations appear in infinite quantum spin systems.
   In Chapter 4, we study the set of infinite volume ground states for Kitaev's abelian quantum double models
   and summarize the results of \cite{ChaNN}.
   We show that the single excitation states as constructed in \cite{Naaijkens11} 
   are infinite volume ground states, that is, local perturbations cannot remove the charge.
   The single excitations states, which are inequivalent for distinct charges, give a complete characterization of the sector theory for the set of ground states.  Furthermore, any pure ground state is equivalent to some single excitation ground state.
   
   In the infinite system, quasi-particle excitations are thought to be classified by 
   certain representations of the algebra of observables.
   Equivalence classes of representations form different charged superselection sectors of the system.
   In Chapter 5, we introduce a new superselection criterion selecting almost localized and transportable $*$-endomorphisms with respect to a vacuum state
   and summarize the results in preparation of \cite{ChaNN2}.
   We show that if the vacuum state satisfies certain locality conditions then the superselection structure will be a braided tensor $C^*$-category.
   Further, this superselection structure is stable up to deformations by a quasi-local dynamics.
   This result is then applied to show that the anyon structure of the abelian quantum double models is stable under local perturbations.

   \newpage
   
   \section*{Acknowledgments}

  I'd like to start by thanking my wife, Sandra Thao.
  Your dedication and unwavering support throughout this uncertain time of graduate school has been my backbone.
  I love you.

  Thanks to my advisor Bruno Nachtergaele for sharing with me your enthusiasm and endless wealth of knowledge of quantum spin systems.
  It was by your encouragement and careful guidance that I have somehow managed to finish my graduate studies.
  Thanks to my unofficial supervisor Pieter Naaijkens.  
  The result of this thesis would not have been possible without your original insights for studying superselection sectors of infinite quantum spin systems.
  Thanks to Craig Tracy and Greg Kuperberg for serving on my dissertation committee and the helpful comments on my thesis.
  Thanks to the UC Davis Mathematics staff, especially Tina and Sarah, for your part in making my time at Davis worry-free.
  Thanks to Yasuyuki Kawahigashi for supporting my time at the Universiy of Tokyo.
  A special thanks to all current and past members of our Friends of Anyons research group, especially
  Sven Bachmann, Jogia B.,  Michael Bishop,  Alvin Moon, and Amanda Young.

  To my family, Jenn, Emma, Hannah, April, Steph, Joe, Elliot, Tim, Mom, Dad and Sandra's family, thanks for believing in me and the constant support!
  Thanks to Detroit Life Church for letting me be a part of your family here in Sac. 
  To my grandma Tais, I will always remember your smile.
  
  This work was supported in part by the National Science Foundation (NSF) under Grant DMS-1515850,
  the NSF East Asian and Pacific Science Initiative (EAPSI) in collaboration with the JSPS Summer Program Grant OISE-1515557,
  and the Graduate Assistance in Areas of National Need (GAANN) Fellowship.

   \mainmatter
   
   \pagestyle{maintext}
   
   % Redefine plain page style so that the first pages of 
   % chapters have desired page style.
   %
   \fancypagestyle{plain}{%
      \renewcommand{\headrulewidth}{0pt}
      \fancyhf{}
      \cfoot{\thepage}%%%
   }%
   
   \chapter{Introduction}
   \label{ch:intro}

	A family of quantum states with qualitatively similar properties are generally said to be in the same quantum phase.
	Quantum phases with local order parameters have been successfully analyzed in the Landau theory of symmetry breaking and phase transitions, see \cite{Toledano} for a review of Landau theory.
	In the 1980s, physicists discovered phases of matter with no local order parameters called \emph{topological phases}.
	For their foundational work on topological order, David J. Thouless, F. Duncan M. Haldane, and J. Michael Kosterlitz received the 2016 Nobel Prize in Physics.
	
	Some defining features of topological order include: there are no local order parameters, the ground state degeneracy depends on the topology of the underlying space, the area law for entanglement entropy is satisfied with a  correction term  called the topological entanglement entropy,
	and the elementary excitations have braid statistics.
	Although, topological order has been extensively studied in both concrete models and theoretically,
	a mathematically rigorous formulation for topological order in quantum many-body systems has yet to be developed. 
	The aim of this thesis is to provide a mathematical framework for the analysis of 
	certain properties  of topological order in the setting of infinite quantum spin systems.
	
	Elementary quasi-particle excitations with braid statistics are called \emph{anyons}~\cite{Wilczek}.  
	A operation braiding one anyon around another may result in a factor of any phase or even a unitary operation to the initial state.
	Perhaps the most well known model for anyons are as the quasi-particle excitations of fractional quantum 
	hall systems~\cite{ArovasSW, MooreR}. 
	Braid statistics have been studied in the context of local quantum physics~\cite{FredRS,FrohlichG}
	and gauge theories~\cite{BaisDW}. The fusion rules and braiding of anyons is encoded
	algebraically as a unitary modular tensor category~\cite{BakalovK}.
	In particular, the case of the representation theory of the quantum double for a finite group 
	has been well studied~\cite{BaisDW,RochePD,KitaevQD,SzlachV}.
	
	Kitaev \cite{KitaevQD} introduced a family of quantum double models  demonstrating
	the existence of quantum spin models with short-range interactions that have ground states exhibiting 
	topological order and anyonic excitations. 
	Kitaev's models are very special in the sense that 
	they are exactly solvable in the finite volume, 
	have frustration-free ground states, and that the interaction terms in the Hamiltonian are all mutually commuting.
	From the properties above, it may seem as if the models may not be very interesting.
	However, a remarkable feature is that the ground space degeneracy 
	depends on the genus of the surface on which the model is defined. 
	Trivially, these models have a non-vanishing spectral gap above the ground state in the thermodynamic limit.
	In fact, the Hamiltonian of the frustration-free ground state has an entirely discrete spectrum.
	The spectral gap is an important feature in the classification of topologically ordered ground states into a  
	topological phase, and has important implications, such as exponential decay 
	of correlations~\cite{HastingsKLR,NachtergaeleSLR} and entanglement area laws for spin chains~\cite{HastingsAL}.
	It was therefore important to show that this gap does not close for sufficiently small uniform perturbations of these 
	models~\cite{BravyiHM}.
	
	The past two decades have witnessed a rising interest in topologically ordered states, mainly due to the 
	realization that their properties could be useful for fault tolerant quantum computation~\cite{Freedman,KitaevQD}.  
	A topological quantum computation is executed in a topologically ordered state in  three steps.
	First, a quantum state is prepared by creating  anyon particle-antiparticle pairs from the vacuum.
	This can be thought of as the state initialization of a quantum algorithm.
	Next, the anyons are braided along set paths resulting in a  unitary transformation of the initial state.
	Finally, neighboring anyons are fused and any annihilation which takes place is recorded.
	The output of the computation is the final state and the record of annihilation.
	If $G = S_5$ the anyon theory of Kitaev's quantum double model can efficiently simulate any quantum circuit,
	that is, topological quantum computation is universal \cite{KitaevQD}.
	A foundational physical property of topologically ordered states  is that the anyon structure is stable against local perturbations
	and, therefore, topological quantum computation provides robustness to errors on the physical level.
	
	The stability of the anyon structure is closely related to the classification of topological phases of matter. 
	One approach to classifying a phase is to construct a complete set of invariants. 
	By definition, an invariant is a quantity that is constant within a phase. 
	Consequently, if an invariant is computed for two systems and is found to take different values, the systems must be in different phases. 
	From the physical point-of-view, the invariance of the structure of anyon quasi-particles is usually taken as fact.
	However, there are few mathematically rigorous results in this direction \cite{Haah}.
	In the literature, a topological phase is often defined as an open region
	in a space of Hamiltonians where there is a non-vanishing gap above the ground state \cite{ChenGW}. 
	Therefore, the construction of invariants can be expected to rely on the existence of a spectral gap.
	
	The main results of this dissertation are as follows.
	We give a complete characterization of the set of infinite volume ground states for Kitaev's abelian quantum double models. 
	This is joint work of the author in with Bruno Nachtergaele and Pieter Naaijkens \cite{ChaNN}.
	Next, we study a new superselection criterion for infinite quantum spin systems.  
	In the case where the ground state vacuum satisfies certain locality conditions, we show the statistics can be computed exactly and the superselection structure is stable against quasi-local deformations.
	The author plans to continuing working on and publish this result in collaboration with Bruno Nachtergaele and Pieter Naaijkens \cite{ChaNN2}.

	\section{Summary of Main Results}
	
	\subsection{Ground state for Kitaev's abelian quantum double models}
	
	Ground states of quantum lattice models are a well-studied subject. Knowing the set of ground states is essential
	for understanding the properties of quantum many-body systems at sufficiently low temperatures. 
	In the mathematical analysis of certain statistical mechanics phenomena,  such as equilibrium states, phase transitions, 
	superselection sectors and phase classification, it is often necessary or convenient to take the infinite volume limit 
	(or thermodynamic limit)~\cite{BachmannO,BratteliR2,Naaijkens11}. General existence 
	and decomposition properties of the set of infinite volume ground states have been mastered for some time~\cite{BratteliKR,BratteliR2}. 
	The problem of finding the complete the set of  ground states for a given model and proving that it is the complete set, however, has been solved only in a few cases.
	
	We study quantum double models for abelian groups, in their implementation as quantum spin
	Hamiltonians with short-range interactions as defined by Kitaev~\cite{KitaevQD}.
	The simplest example is the toric code model, which corresponds to the choice $G = \mathbb{Z}_2$.
	The abelian quantum double model is particularly interesting because it has all of the characteristic features 
	of topologically ordered systems, while at the same time being simple enough to be tackled directly.
	The main features of the model are: it is exactly solvable in the sense that the Hamiltonian can be explicitly 
	diagonalized; the dimension of the space of ground states of the models defined
	on a compact orientable surface is a topological invariant and corresponds to the number of flat $G$-connections 
	on the lattice (up to conjugation); there is a spectral gap above the ground state;
	the elementary excitations correspond to quasi-particles with braid statistics.
	
	Although the quantum double models are exactly solvable in finite volume,
	much less is known about the thermodynamic limit. 
	The first results in this direction are due to Alicki, Fannes and Horodecki~\cite{AlickiFH}.
	They showed that in the case $G =\ZZ_2$,
	there is a unique frustration-free ground state, which coincides with the translation invariant ground state.
	This uniqueness property is not general \cite{GottsteinW}, but is related to topological order in the ground state.
	The difficulty of solving the full ground state problem can be understood as follows.
	If $\delta$ is the derivation generating the dynamics, one has to find \emph{all} states $\omega$ on the quasi-local algebra $\calA$ of observables that satisfy $\omega(A^*\delta(A)) \geq 0$ for all $A$ in the domain of $\delta$.
	It is possible to construct ground states as weak$^*$ limits of finite volume ground states,
	but even though the boundary goes to infinity in a sense, 
	the resulting state strongly depends on the boundary conditions chosen in the finite volume.
	
	The main result is a complete classification of the set of infinite volume ground states
	for Kitaev's quantum double model for finite abelian groups.
	We find that the set of ground states can be decomposed into $\abs{G}^2$ sectors.
	There is a one-to-one correspondence between the ground state sectors and
	the superselection sectors defined in \cite{FiedlerN}.
	In particular, each sector corresponds to a different anyon type.
	The strategy of the proof is to reduce the infinite volume calculation to 
	a finite volume calculation.
	In particular, we find a boundary term for every finite box 
	such that the restriction of any infinite volume ground state to the box 
	is a ground state of the finite volume Hamiltonian plus the boundary term.
	This strategy is motivated by the fact that infinite volume ground states minimize energy in a local region
	among all states that are equivalent in the complement of that region  \cite{BratteliKR},
	and resembles the classical Dobrushin-Lanford-Ruelle theory of boundary conditions
	for the restriction of infinite volume equilibrium state \cite{FannesW}.
	Although the results of this chapter specialize to the case of abelian groups,
	we believe that similar results hold in the case of non-abelian groups.
	The main technical challenge is that the quantum double $\mathcal{D}(G)$ has higher dimensional irreducible representations.
	In physical terms, this manifests itself in the presence of \emph{non-abelian} anyons.
	
	It is often surprisingly difficult to classify the full set of ground states in the thermodynamic limit.
	To our knowledge, the complete ground state problem has only been solved for the one-dimensional 
	$XY$-model by Araki and Matsui \cite{ArakiM}, for the one-dimensional $XXZ$-models by Matsui \cite{MatsuiXXZ} 
	and Koma and Nachtergaele \cite{KomaN}, and for the finite-range spin chains with a unique frustration free 
	matrix product ground state by Ogata~\cite{Ogata3}.
	To make progress, one typically has to pair the ground state problem with 
	model specific notions; in the $XY$-model it was the Jordan-Wigner transformation to fermions
	and in the $XXZ$-model and the frustration-free spin chains it was a connection to zero-energy states \cite{FNWFCS,GottsteinW}.
	As far as we are aware, our result is the first solution to the ground state problem for a quantum model in two dimensions.

	\subsection{Stability of charges}

	The quasi-local algebra of observables for an infinite quantum spin system $\calA$ has many inequivalent representations.
	For example, consider the ground states of the  Ising spin chain $\omega^{+}$ and $\omega^{-}$ described, 
	respectively, by an infinite tensor product of all plus one eigenvector of $\sigma^3$ and 
	an infinite tensor product of all minus one eigenvectors of $\sigma^3$.
	There are no local operators that will map one state to the other.
	One can show these states are inequivalent.
	
	Most representations do not have any physical relevance (for example, because the energy is unbounded),
	so it is important to restrict the class of representations of interest.
	For example, a theory may have different, inequivalent particle types, like the excitations in the quantum double.
	Another example would be of states with different values of the electric charge.
	Here we will use the term ``charge'' in a generalized sense, as a label of the different particle types.
	Once we can identify different classes of representations with charges, it is reasonable to impose additional constraints.
	In particular, we can impose certain locality conditions, and demand that we are able to move the localization regions around.
	A superselection criterion is a rule that tells us precisely which representations we select.
	A superselection sector is an equivalence class of representations that are all unitarily equivalent and satisfy a superselection criterion.
	
	The Doplicher-Haag-Roberts (DHR) analysis in algebraic quantum field theory showed that 
	starting from a vacuum state and a physically motivated superselection criterion, 
	one could recover a family of superselection sectors corresponding to the global gauge group \cite{DHR1,DHR2}.
	This allows one to recover all physically relevant properties of the charges, such as their particle statistics.
	A similar analysis has been done for the quantum double models, 
	producing the single excitation ground states as the 
	irreducible objects in each superselection sector \cite{FiedlerN,Naaijkens11}.
	The role of the vacuum is played by the translation invariant frustration-free ground state.
	
	We introduce a new superselection criterion for charges in infinite quantum spin systems
	selecting representations that arise from the vacuum by composition with a $*$-endomorphisms that are  \emph{almost localized} and transportable.
	The almost localized property will be defined with respect to a  localization region, which in our case will be an infinite cone, 
	and a rapidly decaying function.
	The main  result is: if the ground state vacuum of the theory has certain locality conditions
	then the superselection structure is that of a braided tensor $C^*$-category, and further,  
	is stable against deformations by a quasi-local dynamics.
	The locality conditions we assume on the vacuum state are Haag duality  and approximate split property for cones \cite{NaaijkensKL}.
	The main technical tools used in the proofs are the notions of asymptotic abelianness and bi-asymptopias introduced in \cite{BuchholzAA}
	and the Lieb-Robinson bounds for quasi-local dynamics \cite{NachOS}.
	We apply our results to prove the stability of anyons in the Kitaev's abelian quantum double models.

   \part{Preliminaries}

   \chapter[% 
      Infinite quantum spin systems
   ]{% 
      Infinite quantum spin systems
   }%
   \label{ch:qss}
   In this chapter we introduce the mathematical framework for infinite quantum spin systems.
   The standard textbook references are Bratteli and Robinson \cite{BratteliR1, BratteliR2} and Simon \cite{SimonLattice}.
   For more recent references, we recommend Nachtergaele and Sims \cite{NachSQSS} and Naaijkens \cite{NaaijkensQSS}.
   
   Let $(\Gamma,d)$ be a metric graph.  
   We will mainly consider the case $\Gamma = \ZZ^\nu$, the square lattice in $\RR^\nu$ with metric $d(x,y) = \abs{x-y}$, 
   although the techniques we develop are general.
   To each $x\in \Gamma$ assign a finite dimensional Hilbert space $ \calH_x$.
   Let $\mathcal{P}_0(\Gamma)$ denote the set of finite subsets of $\Gamma$.
   For $ \Lambda\in \mathcal{P}_0(\Gamma)$  we
   define the Hilbert space of states for the composite system as the tensor product space $ \calH_\Lambda := \bigotimes_{x\in \Lambda} \calH_x$
   with the corresponding algebra of observables $\calA_\Lambda = \calB(\calH_\Lambda)$.
   If $ \Lambda_1 \subset \Lambda_2$ there is a natural inclusion of 
   $ \calA_{\Lambda_1} \hookrightarrow \calA_{\Lambda_2}$ via the map $ A \mapsto A\otimes I_{\Lambda_2\backslash \Lambda_1}$.
   This gives a local net of $C^*$-algebras and allows us to define the algebra of 
   \emph{local observables} as 
   \[ \calA_{loc} = \bigcup_{\Lambda \subset \Gamma} \calA_\Lambda \]
   and the $C^*$-algebra of {quasi-local observables} as the norm completion
   \[
   \calA = \overline{\calA_{loc}}^{\| \cdot \|}.
   \]
   An observable $A$ is said to be \emph{localized} in a set $X$ if $ A \in \calA_X$. 
   The \emph{support} of $A$ is defined as the smallest set such that $A \in \calA_X$.
   
   Let $X \subset \Gamma$ be a potentially infinite set.
   We define the quasi-local algebra of observables in $X$ as the $C^*$-subalgebra 
   \[
   \calA_X \equiv \overline{ \bigcup_{\Lambda\in \mathcal{P}_0(X)} \calA_{ \Lambda}}^{\| \cdot \|} \subset \calA.
   \]
   Notice that when $X \in \mathcal{P}_0(\Gamma)$ then we recover $\calA_X = \calB(\calH_X)$.
   
   An \emph{interaction} is a map $\Phi: \mathcal{P}_0(\Gamma) \ra \calA_{loc}$ such that
   $\Phi(X) \in \calA_X $ and $ \Phi(X)^* = \Phi(X)$.
   The local Hamiltonians corresponding to $\Phi$ are 
   
   \begin{equation*} 
   H_\Lambda = \sum_{X\subset \Lambda} \Phi(X) 
   \end{equation*}
   and the Heisenberg dynamics are given by the one-parameter group of automorphisms on $\calA_\Lambda$,
   \begin{equation*}
   \tau^\Lambda_t (A) = e^{i t H_\Lambda} A e^{-i t H_\Lambda}.
   \end{equation*}
   
   The term quantum spin system describes a physical system where $\calH_x$ 
   holds a representation of some spin-$j$ particle fixed at $x$
   and $\Phi$ describes the interactions.
   For instance if we consider a chain of spin-$\frac{1}{2}$ particles,
   the ferromagnetic XXZ Heisenberg model given by
   \begin{equation*}
   H_{[a,b]} = -\sum_{i=a}^{b-1} \sigma_x^1 \sigma_{x+1}^1 + \sigma_{x}^2 \sigma_{x+1}^2 + \Delta\sigma_x^3\sigma_{x+1}^3,
   \end{equation*}
   models magnetism in a idealized low-temperature condensed matter system on the line.
   The interaction terms $\vec{\sigma} = ( \sigma^1, \sigma^2, \sigma^3)$ are the 
   Pauli matrices representing the components of spin
   obeying the angular momentum commutation relation:
   $[ \sigma^i, \sigma^j] = 2 i \epsilon_{ijk} \sigma^k.$
   
   \section{Dynamics}\label{sec:dynamics}
   In this section, we establish sufficient decay conditions on an interaction  such that 
   there exists a corresponding quasi-local infinite volume dynamics. 
   Decay conditions on $\Phi$ are given with respect to locality and regularity of the metric graph $(\Gamma,d)$.
   \begin{defn}\label{def:Ffunc}
   	A function $F: \RR^{\geq 0} \ra \RR^{> 0}$ is called an $\mathcal{F}$-function for $\Gamma$ if it is monotone decreasing and satisfies:
   	\begin{equation}\label{eqn:unifint}
   	\|F \|_0 = \sup_{x \in \Gamma} \sum_{y \in \Gamma} F(d(x,y)) < \infty  \quad \quad \text{ (uniform integrability)}
   	\end{equation}
   	\begin{equation*}
   	C_F  = \sup_{x,y\in \Gamma} \sum_{z\in \Gamma} \frac{ F(d(x,z)) F(d(z,y))}{F(d(x,y))} < \infty \qquad \text{ (convolution inequality)}.
   	\end{equation*}
   \end{defn}
   For $\Gamma = \ZZ^\nu$ the function 
   $F(r) = \frac{1}{(1+r)^{\nu+\epsilon}}$ is an $\mathcal{F}-$function for all $\epsilon>0$.
   
   Let $ b>0$, if $F$ is an $\mathcal{F}$-function then
   \begin{equation*}
   F_{bg}(r)\equiv e^{-bg(r)} F(r)
   \end{equation*}
   is an $\mathcal{F}$-function if
   $g:\RR^{\geq 0} \ra \RR^{\geq 0}$ is uniformly continuous, non-decreasing and sub-additive, that is, $g(x+y) \leq g(x) + g(y)$.
   These properties are satisfied by the following functions:
   \begin{align}
   g(r) &= r^\alpha  \qquad \quad \mbox{ for } 0<\alpha\leq 1, \\
   g(r) &= \left\{ \begin{array}{ll}
   \frac{r}{\ln^p(r)} & \mbox{ if }  x> e^p\\
   \left( \frac{e}{p}\right)^p & \mbox{ if } x\leq e^p
   \end{array} \right.\label{eqn:subexpg}
   \end{align}
   for $p \in \NN$.
   
   When the context is clear, we will abuse notation and denote 
   \begin{equation*}
   F_{b}(r) \equiv e^{-b r }F(r).
   \end{equation*} 
   
   \begin{defn}
   	An interaction $\Phi$ satisfies a finite $F$-norm if 
   	\begin{equation*}
   	\| \Phi\|_F \equiv \sup_{x,y \in \Gamma} \frac{1}{F(d(x,y))} \sum_{\substack{ X \subset \Gamma: \\x,y \in X}} \| \Phi(X) \| < \infty.
   	\end{equation*}
   \end{defn}
   
   One measure of locality in quantum spin system is by commutator bounds.  
   Let $X,Y \in \mathcal{P}_0(\Gamma)$.
   If $d(X,Y) >0$ then $[A,B] = 0$ for all $A\in\calA_X$ and $ B \in\calA_Y$.  
   To measure the diffusion or spreading of a dynamics $\tau_t^\Lambda$, 
   Lieb and Robinson \cite{LiebR} considered bounding the commutator $[ \tau_t^\Lambda(A), B]$.
   Since the Heisenberg dynamics are non-relativistic, 
   it is generally expected that the commutator norm is non-zero at any finite time $t>0$.
   In the relativistic case,  the commutator is non-zero only if $A$ and $B$ are not spacelike localized.

   \begin{thm}\cite{NachOS}
   	Let $\Phi$ be an interaction with a finite $F$-norm and $X, Y \subset \Lambda \in \mathcal{P}_0(\Gamma)$.
   	Then, for any $A \in \calA_X$ and $B \in \calA_Y$  with $d(X,Y) >0$, 
   	\begin{equation}\label{eqn:ExpLRb}
   	\| [ \tau^\Lambda_t (A) ,B ] \| \leq \frac{2 \|A \| \| B\|}{C_F} (e^{ v_\Phi \abs{t}} - 1) \sum_{x\in X}\sum_{y\in Y}F(d(x,y))
   	\end{equation}
   	for any $t \in \RR$, where $v_\Phi = 2 \|\Phi\|_{F} C_{F}$. 
   \end{thm}
   
   If there is a $b>0$ such that $\Phi$ satisfies an finite $F_{bg}$-norm then \eqref{eqn:ExpLRb} becomes 
   \begin{equation}\label{eqn:fvLRb}
   \| [ \tau^\Lambda_t (A) ,B ] \| \leq \frac{2 \|A \| \| B\|}{C_{F_{bg}}} \|F\|_0\min\{\abs{X}, \abs{Y}\} e^{ v_{bg} \abs{t} - b g(d(X,Y))},
   \end{equation}
   where $ v_{bg} = 2 \|\Phi\|_{F_{bg}} C_{F_{bg}}$.
   Bounds of the form \eqref{eqn:ExpLRb} and \eqref{eqn:fvLRb} are generally referred to as \emph{Lieb-Robinson bounds}.
   
   The existence of an infinite volume dynamics $\{\tau_t\}_{t \in \RR}$ follows from the Lieb-Robinson bounds. 
   We say that a sequence $\Lambda_n \in \mathcal{P}_0(\Gamma)$ is increasing and exhausting in $\Gamma$
   if $\Lambda_n \subset \Lambda_{n+1}$ and $\Gamma = \bigcup_{n} \Lambda_n$.
   
   \begin{thm}\label{thm:infvoldyn}\cite{NachOS}
   	Let $\Phi$ be an interaction with a finite $F$-norm.
   	Then, for any increasing and exhausting sequence $\Lambda_n$ the norm limit 
   	\begin{equation*}
   	\tau_t(A) \equiv \lim_{n \ra \infty} \tau_t^{\Lambda_n}(A) 
   	\end{equation*}
   	exists for all $t\in \RR$ and $A \in \calA$.
   	The limiting dynamics $\tau_t$ defines a strongly continuous, one-parameter group of automorphisms on $\calA$.
   	The convergence is uniform for $t$ in compact sets and 
   	is independent of the sequence $\Lambda_n$.
   \end{thm}
   
   Let  $\Phi$ have a finite $F$-norm and $A \in \calA_{ \{x\}}$.
   For an increasing and exhausting sequence $\Lambda_n$ in $\Gamma$ such that $ x \in \Lambda_1$ we have 
   that for $m<n$ 
   \begin{align}\label{eqn:dercauchy}
   \| [H_{\Lambda_n} - H_{\Lambda_m}, A] \| &\leq 2\|A\| \sum_{y\in \Lambda_n\setminus\Lambda_m} \sum_{\substack{X \subset \Gamma \\ x,y\in X}} \|\Phi(X)\| \\
   &\leq 2 \|A\| \| \Phi\|_F \sum_{y \in \Lambda_n\setminus\Lambda_m} F(d(x,y)),
   \end{align}
   where the second inequality comes from the assumption that $\Phi$ has a finite $F$-norm.
   The right hand side can be made arbitrarily small due to the uniform integrability condition \eqref{eqn:unifint}.
   By a similar argument,  the norm limit exists for all $A\in \calA_{loc}$
   \begin{equation*}
   \delta(A) \equiv\lim_{n \ra \infty} [ H_{\Lambda_n}, A] \quad \text{ for all } \quad A\in\calA_{loc}.
   \end{equation*}
   Therefore, $\delta$ is a norm densely defined derivation and is norm-closable 
   with $\calA_{loc}$ as a core (\cite{BratteliR2}, Proposition 6.2.3).
   It can be checked that
   
   \begin{align*}
   \frac{d}{dt} \tau_t(A)  = i \tau_t(\delta(A) ) \quad \text{ for all} \quad A\in\calA_{loc}.
   \end{align*}
   We call $\delta$ the generator of the dynamics $\tau_t$ and  write $\tau_t = e^{it \delta}$.

   The Lieb-Robinson bound \eqref{eqn:ExpLRb} was stated for a local dynamics $\tau_t^\Lambda$, $\Lambda \in \mathcal{P}_0(\Gamma)$.
   However, \eqref{eqn:ExpLRb} is indeed independent of the finite volume $\Lambda$ and therefore can be extended to the infinite system.
   In particular, \eqref{eqn:ExpLRb} holds when both $A$ and $B$ have infinite support. 
   
   \begin{prop}
   	Let $\Phi$ be an interaction with a finite $F$-norm.
   	Then, for any $A \in \calA_{X}$ and $B \in \calA_{Y}$ with $d(X,Y) >0$,
   	\begin{equation}\label{eqn:LRbound}
   	\| [ \tau_t (A) ,B ] \| \leq \frac{2 \|A \| \| B\|}{C_F} (e^{ v_\Phi \abs{t}} - 1) \sum_{x\in X}\sum_{y\in Y}F(d(x,y)).
   	\end{equation}
   \end{prop}
   
   \begin{proof}
   	Suppose $X,Y \in \mathcal{P}_0(\Gamma)$ are finite subsets with $d(X,Y) >0$
   	and let $A \in \calA_X$ and $ B \in \calA_Y$.
   	Let $\Lambda_n$ be an increasing and exhausting sequence in $\Gamma$.
   	Then, $X, Y \subset \Lambda_n$ for $n$ large enough.
   	Therefore, \eqref{eqn:ExpLRb} implies 
   	\begin{align*}
   	\| [\tau_t(A) , B] \| & \leq 2 \| \tau_t (A) - \tau_t^{\Lambda_n}\| \|B\| + \| [\tau_t^{\Lambda_n}(A), B]\|\\
   	&  \leq 2 \| \tau_t (A) - \tau_t^{\Lambda_n}\| \|B\|+ \frac{2 \|A \| \| B\|}{C_F} (e^{ v_\Phi \abs{t}} - 1) \sum_{x\in X}\sum_{y\in Y}F(d(x,y))
   	\end{align*}	
   	Taking the limit as $ n\ra \infty$ gives the result for local observables.
   	
   	Now suppose $X,Y \subset \Gamma$ be not necessarily finite subsets with $d(X,Y) >0$
   	and let $A \in \calA_X$ and $B \in\calA_Y$.
   	First, let us assume $Y$ is finite.
   	Let $A_n\in\calA_X\cap \calA_{loc}$ be a sequence of local observables such that $\| A - A_n \| \ra 0$.
   	Let $X_n$ denote the support of $A_n$.
   	Then, 
   	\begin{equation*}
   	\sum_{x\in X_n}\sum_{y\in Y}F(d(x,y)) \leq  \sum_{x\in X}\sum_{y\in Y}F(d(x,y)).
   	\end{equation*}
   	It follows that 
   	\begin{align*}
   	\| [\tau_t(A) , B] \| & \leq 2 \| \tau_t(A -A_n)\| \| B\| + \| [\tau_t(A_n), B]\| \\
   	&\leq 2 \| A -A_n\| \| B\| +\frac{2 \|A_n \| \| B\|}{C_F} (e^{ v_\Phi \abs{t}} - 1) \sum_{x\in X_n}\sum_{y\in Y}F(d(x,y))\\
   	&\leq 2 \| A -A_n\| \| B\| +\frac{2 \|A_n \| \| B\|}{C_F} (e^{ v_\Phi \abs{t}} - 1) \sum_{x\in X}\sum_{y\in Y}F(d(x,y))\\
   	&\ra \frac{2 \|A \| \| B\|}{C_F} (e^{ v_\Phi \abs{t}} - 1) \sum_{x\in X}\sum_{y\in Y}F(d(x,y)).
   	\end{align*}
   	A similar argument shows the result in the case where both $X$ and $Y$ are infinite subsets.
   \end{proof}

   A complementary viewpoint of quasi-locality in the dynamics $\tau_t$ is 
   to bound the growth of support for an observable evolved in time.
   The following lemma gives a precise relation for commutator bounds and local approximations,
   which allows us to use the Lieb-Robinson bounds to exactly this.
   
   \begin{lemma}\cite{NachSW}
   	Let $\calH_1$ and $\calH_2$ be Hilbert spaces.
   	Then, there is a completely positive linear map $\mathbb{E}: \calB(\calH_1 \otimes \calH_2) \ra \calB(\calH_1)$ with the following properties:
   	\begin{enumerate}
   		\item For all $A \in \calB(\calH_1), \mathbb(E)( A \otimes I) = A$,
   		\item If $A \in \calB(\calH_1 \otimes \calH_2)$ satisfies the commutator bound 
   		\begin{equation*}
   		\| [A,I \otimes B] \| < \epsilon \|A \| \|B\| \quad \text{ for all } \quad B \in \calB(\calH_2), 
   		\end{equation*}
   		then 
   		\begin{equation*}
   		\|\mathbb{E}(A) \otimes I - A \| < \epsilon \| A\|,
   		\end{equation*} 
   		\item For all $C,D \in \calB(H_1)$ and $ A \in \calB(H_1 \otimes \calH_2)$
   		\begin{equation*}
   		\mathbb{E}( C A B) = C \mathbb{E}(A)B.
   		\end{equation*}
   	\end{enumerate}
   \end{lemma}
   
   Let $\epsilon >0$ be given 
   and $X  \in \mathcal{P}_0(\Gamma)$. 
   Suppose there is a $b>0$ such that $\Phi$ has a finite $F_{b}$-norm and let $\tau_t$ be the corresponding infinite dynamics.
   Let $X_t(l) = \{ x \in \Gamma: d(x , X)< v_{\Phi} \abs{t} + l\}$ denote the set $X$ fattened by $v_\Phi \abs{t} + l$.
   Choose $l>0$ such that
   \[ \frac{2 }{C_{F_{b}}} \|F\|_0 \abs{X} e^{ - l} <\epsilon. \]
   Then by \eqref{eqn:fvLRb}, it follows that for all $A \in \calA_X$ and $B \in \calA_{ X_t(l)^c}$
   \begin{align*}
   \| [\tau_t(A), B] \| \leq \epsilon \|A\| \|B\|.
   \end{align*}
   Therefore, there is an observable $A' \in \calA_{X_t(l)}$ such that 
   \begin{equation*}
   \| A - A'\| < \epsilon \|A\|.
   \end{equation*}

   In practice, $A'$ is calculated using the partial trace.
   Let $X \subset \Lambda \in \mathcal{P}_0(\Gamma)$ 
   and define the partial trace of $A$ in $X$ with respect to $\Lambda$ by
   \begin{equation*}
   \langle A \rangle_{X^c} = \int_{ \mathcal{U}_\Lambda(X^c) } U^* A U \mu(dU),
   \end{equation*}
   where $\mathcal{U}_\Lambda(X^c)$ is the unitary group in $\calA_{X^c\cap\Lambda}$ 
   and $ \mu$ is the associated normalized Haar measure.
   Then, for a large enough $\Lambda$, $A$ is well approximated by
   \begin{equation*}
   A' =  \langle A \rangle_{X_t(l)^c}.
   \end{equation*}

   \section{Ground states}
   A \emph{state} on $\calA$ is a linear functional $\omega: \calA \ra \CC$ such that 
   $\omega(A) \geq 0$ if $A\geq 0$  and $ \omega(I) = 1$.  
   The set of all states is denoted $ \calA_{+,1}^*$
   and is a compact convex set, its extremal points are called pure states.
   Since $\calA$ is unital, the set of states is compact in the weak$^*$ topology 
   by the Banach-Alaoglu theorem.
   
   Let $\tau_t = e^{it \delta}$ be a strongly continuous one-parameter group of automorphisms on $\calA$
   and $\delta$ be its generator with dense domain $dom(\delta)$.
   
   \begin{defn}\label{defn:gs}
   	A state $\omega$ is a $\tau$-ground state if
   	\begin{equation}\label{eqn:gs}
   	\omega( A^* \delta(A)) \geq 0 \quad \quad \text{ for all } \quad A \in dom(\delta).
   	\end{equation} 
   \end{defn}
   
   Recall, if $\tau_t$ is a quasi-local dynamics corresponding to an interaction $\Phi$ with finite $F$-norm
   then $\calA_{loc}$ is a core for $\delta$, thus we can replace $dom(\delta)$ with $\calA_{loc}$ in the definition above.
   When $\omega$ is a $\tau$-ground state for an infinite quantum spin system, we will typically call it an \emph{infinite volume ground state}.
   
   This definition can be interpreted as an infinite volume variational principle expressing
   that local perturbations do not decrease the energy of a ground state.
   At finite temperature $T = 1/\beta$, equilibrium states are defined by the KMS-condition \cite{HaagHW, BratteliR2}
   \begin{equation*}
   \omega( A \tau_{i \beta}(B)) = \omega(BA)
   \end{equation*}
   for all $A, B$ in a norm dense, $\tau$-invariant $*$-subalgebra.
   The ground state condition can be obtained as the zero-temperature limit 
   $T\ra 0$ of the KMS-condition (see Theorem 5.3.15, \cite{BratteliR2}).
   
   Let $K_\tau = \{ \omega \in \calA_{+,1}^* \mid \forall A\in\calA_{loc}; \omega(A^*\delta(A))\geq 0 \}$ 
   be the set of ground states.
   Then, $K_\tau$ is compact and closed in the weak$^*$ topology and is a face in $\calA_{+,1}^*$ 
   (Theorem 5.3.37, \cite{BratteliR2}).
   In a quantum spin system, the space of ground states is not empty. 
   Indeed, one may construct an infinite volume ground state by taking a weak$^*$ limit of finite volume ground states as follows.
   Let $\Lambda_n$ be an increasing and exhausting sequence of finite subsets in $\Gamma$
   and define $\partial \Lambda_n = \Lambda_n \setminus \Lambda_{n-1}$.
   Let $\Omega_n \in \calH_{\Lambda_n}$ be any normalized ground state for $\widetilde{H}_{\Lambda_n} = H_{\Lambda_n}+W_{\partial\Lambda_n}$
   where $W_{\partial \Lambda_n} \in \calA_{ \partial\Lambda_n}$ is any self-adjoint boundary operator.
   From the variational principle, the ground state energy $E_0 \equiv \langle \Omega_n, \widetilde{H}_{\Lambda_n} \Omega_n \rangle$ satisfies
   \begin{equation*}
   E_0 = \inf_{\psi \in \calH_{\Lambda_n}} \frac{\langle \psi, \widetilde{H}_{\Lambda_n} \psi\rangle}{\langle \psi, \psi\rangle}. 
   \end{equation*}
   Therefore, for all $A \in \calA_{ \Lambda_n}$
   \begin{align*}
   \langle \Omega_n, A^* [\widetilde{H}_{\Lambda_n}, A] \Omega_n \rangle & = \langle \Omega_n, \left(  A^* \widetilde{H}_{\Lambda_n} A - A^*A \widetilde{H}_{\Lambda_n} \right) \Omega_n\rangle
   = \langle \Omega_n, A^* \left( \widetilde{H}_{\Lambda_n} -E_0   \right) A\Omega_n\rangle
   \geq 0
   \end{align*}
   
   Let  $\omega_n$  be an extension of the state  $\langle \Omega_n,\ \cdot \ \Omega_n\rangle $ on $\calA_{\Lambda_n}$ to the quasi-local algebra $\mathcal{A}$.
   For instance, fix a state $ \omega$ on $\calA$ and define $\omega_n \equiv \langle \Omega_n,\  \cdot\  \Omega_n\rangle  \otimes \omega|_{\calA_{ [a,b]}\Lambda_n^c}$.
   By Banach-Alaoglu, $\calA_{+,1}^*$ is weakly compact and thus there is a converging subsequence we also denote $\omega_n$.
   Let $A \in \calA_{loc}$ be supported on $X$.
   There exists $N>0$ such that if $n>N$ then $X \subset \Lambda_n$.
   It follows that 
   \begin{align*}
   \omega(A^* \delta(A)) &= \lim_{n \ra \infty} \langle \Omega_n, A^* [\widetilde{H}_{\Lambda_n},A]\Omega_n\rangle \geq 0.
   \end{align*}
   Therefore,  $ \omega$ is an infinite volume ground state.
   The existence of KMS states for quantum spin systems follow from similar arguments, see~\cite{PowersS}.
   
   As the above construction shows, the limiting ground state $\omega$ may be highly sensitive to the choice of boundary conditions.
   However, all boundary conditions are allowed in the sense that for an arbitrary boundary term $W_{\partial \Lambda_n}$  we have that for all $ A\in \calA_{loc}$ that
   \begin{equation}\label{eqn:derivationlimit}
   \delta(A) = \lim_{n \ra \infty}  [ H_{\Lambda_n}, A] = \lim_{n \ra \infty} [H_{\Lambda_n} + W_{\partial \Lambda_n}, A].
   \end{equation}
   That is, the infinite dynamics $\tau_t$ is independent of boundary conditions.
   In the following result, we will see an alternative classification for infinite volume ground states clarifying 
   the statement that infinite volume ground states are indeed bulk ground states.
   \begin{thm}\cite{BratteliKR}
   	Let $\Phi$ be an interaction satisfying a finite $F$-norm.
   	For each $\Lambda \in \mathcal{P}_0(\Gamma)$, let $\widetilde{H}_\Lambda \in \calA$ be a self-adjoint observable such that 
   	$ \delta(A) = i [\widetilde{H}_\Lambda, A]$ for all $ A\in \calA_{ \Lambda}$.
   	A state $\omega$ is an infinite volume ground state iff $\omega$ satisfies
   	\begin{equation*}
   	\omega(\widetilde{H}_\Lambda) = \inf \left\{  \omega'(\widetilde{H}_\Lambda): \omega' \in C_\Lambda^\omega \right\}
   	\end{equation*}
   	where 
   	$	C^\omega_\Lambda := \left\{ \omega' \in \calA^*_{1,+} : \omega'|_{\calA_{\Lambda^c}} = \omega|_{\calA_{\Lambda^c}} \right\}.$
   \end{thm}
   
   For each state $\omega$, the \emph{GNS-construction} gives a $*$-representation
   $\pi:\calA \ra \calB(\calH)$ and a cyclic vector $\Omega$ satisfying 
   $\overline{\{ \pi(A)\Omega:A \in\calA_{loc}\}^{\|\cdot\|} } = \calH$
   and 
   \begin{equation*}
   \omega(A) = \langle \Omega, \pi(A) \Omega\rangle.
   \end{equation*} 
   The representation $\pi$ is determined uniquely up to unitary equivalence.
   Thus for a time-invariant state, i.e.  $\omega(\tau_t (A)) = \omega(A)$ for all $t \in \RR$,
   there is a one parameter family of 
   unitary equivalent GNS representations defined by $ (\pi_\omega\circ \tau_t, \calH_\omega, \Omega_\omega)$.
   Let $\{U_t\}_{t\in\RR}$ be the strongly continuous one-parameter group of unitary operators implementing the equivalence.
   By Stone's theorem, we recover a self-adjoint generator $H_\omega$, called the \emph{GNS-Hamiltonian},
   with a core given by $\{\pi(A)\Omega: A\in \calA_{loc}\}$.
   Since $U_t \Omega = \Omega$ for all $t$ we have that $H_\omega \Omega =  0$.
   Differentiating $\pi(\tau_t(A)) = U_t^*\pi(A) U_t$ at $t = 0$ we obtain
   \begin{equation}\label{eqn:GNSHamCom}
   \pi(\delta(A)) = [H_\omega, \pi(A)] \quad \text{for all} \quad A \in \calA_{loc}.
   \end{equation}
   
   \begin{lemma}
   	$\omega$ is a $\tau$-ground state if and only if $\omega$ is $\tau_t$ invariant and $H_\omega\geq 0$.
   \end{lemma}
   \begin{proof}
   	Suppose $\omega$ is a ground state, that is, $ \omega(A^* \delta(A))\geq 0$.
   	It follows that 
   	\[ - i \omega(A^* \delta(A)) = \overline{- i \omega(A^* \delta(A))} = i \omega(\delta(A^*)A) \]
   	so that 
   	\[\omega(A^*\delta(A) +\delta(A^*)A) = 0. \]
   	First, consider the case where $A$ is a positive operator and let $A = B^2$ 
   	for $B$ a self-adjoint operator.
   	Since $\delta$ is a derivation we have that 
   	$\omega(\delta(A)) = \omega(\delta(B^2)) = \omega(\delta(B)B + B\delta(B)) = 0$.
   	Therefore, $\frac{d}{dt} \omega(\tau_t(A))|_{t=0} = \omega(\delta(A)) = 0$
   	so that $ \omega(\tau_t(A)) = \omega(A)$.
   	But, every operator can be written as a linear combination of four positive operators,
   	so by linearity of $\omega$ it follows that $\omega$ is $\tau_t$ invariant.
   	For $A \in A_{loc}$ and $\psi = \pi(A) \Omega$, it follows from \eqref{eqn:GNSHamCom} that
   	\begin{align*}
   	\langle \psi, H_\omega \psi \rangle& = \left\langle \pi(A) \Omega, H_\omega \pi(A) \Omega\right\rangle  = \left\langle \Omega, \pi(A^*) [H_\omega,\pi(A) ] \Omega \right\rangle.\\
   	& = \omega(A^* \delta(A)) \geq 0.
   	\end{align*}
   	Therefore, $H_\omega \geq 0$.
   	
   	Now suppose $\omega$ is $\tau_t$ invariant and $H_\omega\geq 0$.
   	Again from \eqref{eqn:GNSHamCom} it follows that 
   	$ 0 \leq \langle \psi, H_\omega \psi\rangle = \omega(A^* \delta(A)),$
   	where $ \psi = \pi(A) \Omega$ for any $A \in \calA_{loc}$.
   \end{proof}
   
   The GNS Hamiltonian will generally be an unbounded self-adjoint operator.
   Its relation to the finite volume Hamiltonians is through the derivation of the dynamics, \eqref{eqn:derivationlimit} and \eqref{eqn:GNSHamCom}.
   In certain cases, the local Hamiltonians will converge in some sense to the GNS-Hamiltonian.
   As a result, certain spectral properties of $H_\omega$ can also be deduced from the finite volume spectral properties of $H_\Lambda$.

   \begin{defn}
   	For $A \in \calA$ and $\lambda \in \CC \setminus spec(A)$ let $R_\lambda(A) = (A - \lambda I)^{-1}$ be the resolvent.
   	A sequence of self-adjoint (not necessarily bounded) operators $\{A_n\}_{n=1}^\infty$  on $\calH$
   	is said to converge to $A$ in the \emph{strong resolvent sense} if $R_\lambda(A_n) \ra R_\lambda(A)$ 
   	in the strong operator topology for all $\lambda \in \CC$ with $ \operatorname{Im}(\lambda) \neq 0$.
   \end{defn}
   
   \begin{lemma}\label{lem:SRC}
   	Let $\omega$ be a ground state and $\Lambda_n \in \mathcal{P}_0(\Gamma)$ be an increasing and exhausting sequence.
   	If $\pi(H_{\Lambda_n}) \Omega \ra 0$ then,
   	$\pi(H_{\Lambda_n}) \ra H_\omega$  in the strong resolvent sense.
   	It follows that $spec(H_\omega) \subset \bigcup_{n} spec(\pi(H_{\Lambda_n}))$.
   \end{lemma}
   
   \begin{proof}
   	Let $A \in A_{loc}$.  From \eqref{eqn:GNSHamCom} it follows that 
   	\begin{align*}
   	H_\omega \pi(A) \Omega & = \pi(\delta(A)) \Omega
   	= \lim_{n\ra \infty}\pi([H_{\Lambda_n}, A]) \Omega 
   	= \lim_{n \ra \infty}\pi(H_{\Lambda_n}) \pi(A) \Omega.
   	\end{align*}
   	Note that $\pi(\calA_{loc})\Omega_\omega$ is a common core for $\pi(H_{\Lambda_n})$ 
   	and $H_\omega$.
   	Thus, $\pi( H_{\Lambda_n}) \ra H_\omega$ in the strong resolvent sense (see Theorem VIII.25,  \cite{ReedSimon}),
   	which implies (see Theorem VIII.24, \cite{ReedSimon})
   	\begin{equation*}
   	spec(H_\omega) \subset \bigcup_{n\geq 0} spec(\pi_\omega(H_\Lambda).
   	\end{equation*} 
   \end{proof}
   
   In the recent years, there has been much interest in gapped ground states and 
   in particular, gapped ground state phases \cite{ChenGW, BachmannMNS, BachmannO}.
   \begin{defn}\label{defn:gappedgs}
   	A $\tau$-ground state $\omega$ is said to be a \emph{gapped ground state}
   	if the corresponding GNS-Hamiltonian $H_\omega$ has a spectral gap at $0$, that is,
   	there exists $\gamma >0$ such that 
   	\begin{equation*}
   	\inf \{ \lambda \in spec(H_\omega)\setminus\{0\} \} >\gamma.
   	\end{equation*}
   	
   \end{defn}
   \begin{defn}\label{defn:ff}
   	Let $\Phi$ be the interaction corresponding to $\tau$.
   	A $\tau$-ground state $\omega$ is called \emph{frustration-free}  if
   	\begin{equation*}
   	\omega( \Phi(X)) = \inf spec(\Phi(X)) \quad \text{ for all } \quad X \in \mathcal{P}_0(\Gamma).
   	\end{equation*}
   \end{defn}
   
   It is always possible to shift each interaction term $\tilde{\Phi}(X) = \Phi(X) - \inf spec(\Phi(X)) \geq 0$
   without altering the derivation $\delta$. 
   The frustration-free condition then becomes a \emph{zero-energy} condition, $\omega(\tilde{\Phi}(X)) = 0$ for all $X$.
   In the zero-energy case, it follows immediately from 
   Lemma \ref{lem:SRC} that 
   if the local Hamiltonians have a uniform spectral gap, that is, for all $\Lambda \in \mathcal{P}_0(\Gamma)$ there exists $\gamma$ uniform in $\Lambda$ such that 
   \begin{equation*} 
   \inf \{  \lambda \in spec(H_\Lambda)\setminus\{0\} \}\geq  \gamma >0,
   \end{equation*}
   then $H_\omega$ has a spectral gap.
   The algebraic structure of zero-energy ground states was studied extensively in \cite{GottsteinW}.
   
   A current challenge in mathematical physics is estimating lower bounds for the spectral gap of local Hamiltonians.
   In the frustration-free case, the state-of-the-art in estimating lower bounds is given by the martingale method \cite{NachtergaeleSG}
   and its generalizations \cite{BachmannHNYSG,BishopNYSG}.
   It if often possible to estimate lower bounds on the spectral gap for the 
   local Hamiltonians with boundary conditions \cite{AKLT, FNWFCS, Knabe}.
   In certain cases, this implies a spectral gap for the infinite system.
   
   \begin{lemma}
   	Let $\Phi$ be an interaction satisfying a finite $F$-norm
   	and $\Lambda_n \in \mathcal{P}_0(\Gamma)$ be an increasing and exhausting sequence in $\Gamma$.
   	Suppose the following are true: the local Hamiltonians with boundary $\widetilde{H}_{\Lambda_n} = H_{\Lambda_n} + W_{\partial \Lambda_n}$
   	have simple ground state vector $\Omega_n$,
   	$\widetilde{H}_{\Lambda_n}$ have a uniform spectral gap $\gamma$,
   	and the limit $\omega = \wslim\omega_n$ exists where 
   	$\omega_n$  is some extension of the state  $\langle \Omega_n,\ \cdot \ \Omega_n\rangle $   to the quasi-local algebra $\mathcal{A}$.
   	Then, $\omega$ is a gapped ground state.
   \end{lemma}
   
   \begin{proof}
   	Recall $\delta(A) = i [ H_\omega, A]$ for all $A \in \calA_{loc}$.
   	Thus, $H_\omega$ has a spectral gap $\gamma$ if for all $A\in \calA_{loc}$ 
   	such that $\omega(A) =0$ we have
   	\[ \omega( A^* \delta(A)) \geq \gamma \omega(A^* A).\] 
   	
   	Since $\gamma$ is the uniform spectral gap we have that $ \widetilde{H}_{\Lambda_n} \geq \gamma( I -\ket{\Omega_n}\bra{\Omega_n} )$ for all $n$.
   	Let $A \in \calA_{loc}$. Then, by a calculation similar to \eqref{eqn:dercauchy} we have that 
   	\begin{equation*}
   	\lim_{n \ra \infty}\lim_{m \ra \infty} \omega_n( A^*[H_{\Lambda_m} - H_{\Lambda_n} , A]) = 0.
   	\end{equation*}
   	Then, for $A\in \calA_{loc}$ such that $\omega(A) =0$ we have that
   	\begin{align*}
   	\omega( A^* \delta(A)) 
   	& =  \lim_{n\ra \infty} \omega_n( A^*\delta(A) )=  \lim_{n\ra \infty} \lim_{m \ra \infty} \omega_n( A^*[H_{\Lambda_m}, A] )\\
   	& = \lim_{n\ra \infty} \lim_{m \ra \infty} \langle \Omega_n, A^*[\widetilde{H}_{\Lambda_n}, A] \Omega_n\rangle  + \omega_n( A^*[H_{\Lambda_m} - H_{\Lambda_n} , A]) \\
   	& \geq \lim_{n\ra \infty} \gamma \langle \Omega_n, A^*(I -\ket{\Omega_n}\bra{\Omega_n} )A \Omega_n\rangle\\
   	& = \gamma \omega(A^*A) - \lim_{n \ra \infty} \abs{\langle \Omega_n, A \Omega_n\rangle}^2,
   	\end{align*}
   	where we use the fact that $ [W_{\partial\Lambda_n}, A] = 0$ for large enough $n$.
   	The last term vanishes since $  0 = \omega(A) = \lim_{n\ra \infty} \langle \Omega_n, A \Omega_n \rangle$.
   \end{proof}

   The existence of the spectral gap implies weak correlation structure in the ground state.

   \begin{thm}(Exponential decay of correlations, \cite{HastingsKLR,NachtergaeleSLR}) 
   	Let $b>0$ and $\Phi$ be an interaction satisfying a finite $F_b$-norm and $\tau_t$ be the corresponding dynamics.
   	Suppose $\omega$ is a gapped ground state with spectral gap $\gamma>0$ and 
   	$H_\omega$ has a non-degenerate ground state space.
   	Then, for all local observables $A\in\calA_X$ and $B \in \calA_Y$ it holds that 
   	\begin{equation*}
   	\abs{ \omega(AB) - \omega(A)\omega(B) } \leq C(A,B,\gamma) e^{- \mu d(X,Y)}
   	\end{equation*}
   	where $ \mu = \frac{ b \gamma}{4 \|\Phi\|_b C_b + \gamma}$.
   \end{thm}
   
   Let $\omega$ be a pure state.
   Let $\omega|_\Lambda$ denote the restriction of the state $\omega$ to $\calA_\Lambda$
   and $\rho_\Lambda$ be the density matrix of $\omega|_\Lambda$.
   The \emph{entanglement entropy} is defined as 
   \begin{equation*}
   S_\omega(\Lambda) \equiv - \Tr ( \rho_\Lambda \ln \rho_\Lambda) =  - \omega( \ln \rho_\Lambda).
   \end{equation*}
   Hastings \cite{HastingsAL} proved that for quantum spin chains, 
   if $\omega$ is a translation invariant pure gapped ground state with non-degenerate ground state space for the GNS-Hamiltonian
   then the entanglement entropy is bounded uniform in the volume.
   This result was recently improved.
   
   \begin{thm} \cite{BrandaoHAL}
   	Let $\nu =1$ and  $\omega$ be a translationally invariant pure state.
   	Suppose there exists constants $C, \mu >0 $ such that for each  $j \in \ZZ^{>0}$,  $A \in \calA_{ [a,b]}$ and $ B \in \calA_{ [a-j, b+j]^c}$  we have
   	\begin{equation*}
   	\abs{ \omega(AB) - \omega(A)\omega(B)} \leq C \|A\| \|B\| e^{-\mu j}.
   	\end{equation*}
   	Then, there is a constant $K$ such that $S_\omega([a,b]) \leq K$ for all $a,b \in \ZZ$.
   \end{thm}
   
   \section{Gapped ground state phases}
   
   Let $\omega_0$ and $ \omega_1$ be gapped ground states for a quantum spin system with quasi-local algebra $\calA$.
   Then, it is commonly said that $\omega_0$ and $\omega_1$ are in the same \emph{gapped ground state phase}
   if there exists a continuous family of finite range Hamiltonians $H(s)$ such that 
   $H(s)$ has a non-vanishing spectral gap above the ground state, uniform in $s$, and 
   $ \omega_0$ and $ \omega_1$ are ground states of $H(0)$ and $H(1)$ \cite{ChenGW, BachmannMNS}.
   In this section, mainly following \cite{BachmannMNS}, we discuss how under conditions analogous to the above setting the ground states $\omega_0$ and $\omega_1 $ are connected via a quasi-local dynamics.
   
   Let $\Gamma = \ZZ^\nu$ and consider a family of interactions
   $\Phi_s : \mathcal{P}_0(\Gamma) \ra \calA$ for $0\leq s \leq 1$.
   
   \emph{Assumption 1.} We assume the family is differentiable in the parameter $s$ and short-range such that
   for some $a, M>0$
   \begin{equation*}
   \sup_{x,y \in \Gamma} e^{ a d(x,y)} \sum_{\substack{ X\subset Z:\\ x,y\in X} }\sup_s \| \Phi_s (X) \| + \abs{X} \| \partial_s \Phi_s (X) \| \leq M.
   \end{equation*}
   
   \emph{Assumption 2.} We assume there exists an increasing and exhaustive sequence $\Lambda_n \in \mathcal{P}_0(\Gamma)$ such that the
   local Hamiltonians  $H_{\Lambda_n}(s) = \sum_{X \subset \Lambda_n} \Phi_X(s)$ have a uniform spectral gap in the following sense. 
   Let $\Sigma_\Lambda(s) \equiv \operatorname{spec}(H_\Lambda(s))$ and assume that the spectrum decomposes into two
   non-empty disjoint  sets, $\Sigma_\Lambda(s) =\Sigma_\Lambda^{(0)}(s)  \cup \Sigma_\Lambda^{(1)}(s) $ and there are disjoint intervals $I_i(s) = [a_i(s), b_i(s)]$ for $i = 0,1$  
   such that $a_i(s), b_i(s)$ are continuous functions, $a_0(s) \leq  b_0(s) < a_1(s) \leq b_1(s)$ and
   $\Sigma_i(s) \subset I_i(s)$.
   The spectrum is gapped in the sense that 
   there exists $\gamma>0$ with $\min \{  d( \Sigma_{\Lambda_n}^1(s),\Sigma_{\Lambda_n}^0(s) ), d( \Sigma_{\Lambda_n}^1(s),\Sigma_{\Lambda_n}^2(s)) \}\geq \gamma >0$
   where $\gamma$ is uniform  in $s\in [0,1]$ and $n$.
   
   \emph{Assumption 3.} Consider the finite volume dynamics $\tau_t^{H_s(\Lambda)} (A) = e^{i t H_\Lambda(s)} Ae^{-i t H_\Lambda(s)}$.  We  assume there exist $a>0$ and $ v_a, K_a >0$ such that
   the following exponential Lieb-Robinson bound holds 
   \[  \|[\tau_t^{H_\Lambda(s)}(A), B]\| \leq K_a \|A\| \|B\| e^{v_a t} \sum_{x \in X}\sum_{y\in Y} e^{-a d(x,y)}. \]
   
   Let $w_\gamma \in L^1(\RR)$ be a function satisfying 	
   \begin{enumerate}
   	\item $w_\gamma$ is real-valued and $ \int dt w_\gamma(t) = 1$,
   	\item the Fourier transform $\widehat{w}_\gamma$ is supported in the interval $[-\gamma, \gamma]$.
   \end{enumerate}
   For the existence of such a function see Lemma 2.6, \cite{BachmannMNS}.
   
   The spectral flow, or quasi-adiabatic continuation as it was first introduced in \cite{HastingsLSM, HastingsW}, is constructed as follows.
   Consider the self-adjoint operator
   \begin{equation*}
   D_\Lambda(s) = \int_{-\infty}^{\infty} dt w_\gamma(t) \int_0^t du e^{i  u H_\Lambda(s)} H'_\Lambda(s) e^{-i u H_\Lambda(s)}.
   \end{equation*}
   \begin{prop}\label{prop:spectralflow}(Spectral Flow)
   	There is a norm-continuous family of unitaries $U_\Lambda(s)$ such that the spectral projections $P_\Lambda(s)$ onto the subset $\Sigma_\Lambda^1(s)$ are given by
   	\begin{equation*}
   	P_\Lambda(s) = U_\Lambda(s) P(0) U_\Lambda(s)^*. 
   	\end{equation*} 
   	The unitary family is given by the unique solution to
   	$-i \frac{d}{ds} U_\Lambda(s) = D_\Lambda(s) U_\Lambda(s)$ and $U_\Lambda(0) = I$.
   \end{prop}
   
   Define the \emph{spectral flow} automorphism family as
   \begin{equation*}
   \alpha_s^\Lambda(A) = U^*_\Lambda(s) A U_\Lambda(s).
   \end{equation*}
   Since the $U_\Lambda(s)$ are solutions to a time-dependent Schroedinger's equation, 
   the spectral flow will have a cocycle property. 
   
   Let $g(x)$ denote the function given by \eqref{eqn:subexpg} for $p=2$, that is,
   
   \begin{equation*}
   g(r) = \left\{ \begin{array}{ll}
   \frac{r}{\ln^p(r)} & \mbox{ if }  x> e^p\\
   \left( \frac{e}{p}\right)^p & \mbox{ if } x\leq e^p
   \end{array} \right. 
   \end{equation*}
   
   and $F$ be an $\mathcal{F}$-function such that there exists $ 0< \delta < 2/7$ with
   \begin{equation*}
   \sup_{r\geq 1} \frac{e^{-\delta g(r)}}{F(r)} < \infty.
   \end{equation*}
   Define the $\mathcal{F}$-function 
   \begin{equation*}
   F_\Psi(r) \equiv e^{ - \mu g\left(\frac{\gamma}{8 v_a} r\right)} F\left( \frac{\gamma}{8 v_a} r \right).
   \end{equation*}
   The main property of the spectral flow is that it is a quasi-local dynamics connecting the spectral subspaces.
   
   \begin{thm}
   	Let Assumptions 1, 2 and 3 hold.
   	Then, there exists a time-dependent volume-dependent interaction $\Psi_\Lambda(s)$ 
   	and a such that 
   	\begin{equation}\label{eqn:specflowFnorm}
   	\| \Psi_\Lambda\|_{F_\Psi} \equiv 
   	\sup_{x,y \in \Gamma} \frac{1}{F_\Psi(d(x,y))}\sum_{\substack{Z \subset \Lambda \\x,y \in Z}} \sup_{0\leq s \leq 1} \| \Psi_\Lambda(Z,s)\| <\infty
   	\end{equation} 
   	and 
   	\begin{equation*}
   	D_\Lambda(s) = \sum_{Z \subset \Lambda } \Psi_\Lambda(Z,s).
   	\end{equation*}
   	Further, if $\Lambda_n \in \mathcal{P}_0(\Gamma)$ is an increasing and exhausting sequence as given in Assumption 2 and there exist positive constants $ b_1, b_2$ and $ p$ such that 
   	\begin{equation*}
   	d(\Lambda_m, \Lambda_n^c) \geq b_1 (n-m), \quad \text{and} \quad \abs{\Lambda_n}\leq b_2 n^p
   	\end{equation*}
   	then \eqref{eqn:specflowFnorm} holds uniformly in $\Lambda_n$.
   \end{thm}
   
   We can use the results similar to the previous sections to argue the following.
   There exists a strongly continuous cocycle of automorphisms $\alpha_s$ on $\calA$ as the strong limit of $ \alpha_s^{\Lambda_n}$ 
   for an increasing and exhausting sequence $\Lambda_n$.
   In addition, for all $A\in \calA_X$ and $ B\in \calA_Y$ such that $ d(X,Y) >0$ and $ 0\leq s \leq 1$ we have that 
   \begin{equation*}
   \| [\alpha_s(A), B] \| \leq 2 \frac{\|A\| \|B\|}{C_{F_\Psi}} \left(e^{2 \|\Psi\|_{F_\Psi} \abs{s}} -1\right)\sum_{x \in X}\sum_{y\in Y} F_\Psi(d(x,y)) .
   \end{equation*}
   
   \begin{lemma}
   	If $\Phi_s$ is a translation invariant interaction for all $s$ then $ \alpha_s \circ T_x = T_x \circ \alpha_s$ for all $s$ and $x\in\Gamma$.
   \end{lemma}
   
   \begin{proof}
   	This follows from
   	\begin{align*}
   	T_x( D_\Lambda(s)) & = T_x\left(\int_{-\infty}^{\infty} dt w_\gamma(t) \int_0^t du e^{i  u H_\Lambda(s)} H'_\Lambda(s) e^{-i u H_\Lambda(s)} \right)\\ 
   	& = \int_{-\infty}^{\infty} dt w_\gamma(t) \int_0^t du e^{i  u H_{\Lambda+x}(s)} H'_{\Lambda+x}(s) e^{-i u H_{\Lambda+x}(s)}\\
   	& = D_{\Lambda+x}(s).
   	\end{align*}
   \end{proof}

   Let $S_\Lambda(s)$ denote the set of states on $\calA_{\Lambda}$ that are mixtures of eigenstates with energy in $I_0(s)$, see Assumption 2.
   Define $S(s)$ as the set of all weak* limit points of the sets $S_\Lambda$.
   
   \begin{thm}[\cite{BachmannMNS}] \label{thm:autoeq}
   	The set $S(s)$ is automorphically equivalent to $S(0)$ for all $s \in [0,1]$.
   	In particular, 
   	\begin{equation}
   	S(s) = S(0 ) \circ \alpha_s
   	\end{equation}
   	where $\alpha_s$ is the spectral flow on $\calA$.
   \end{thm}
   
   %{\color{red} Include possible discussion on the following: 
   %	
   %- 	Classification of gapped phases. Results in 1d
   %	  
   %-	Stability of area laws.
   %	 
   %-	Set of infinite volume ground-states  is not invariant. 
   %
   %-   Stability of anyons in 2d topological phases}

   \section{Stability of frustration-free Hamiltonians}
   
   One aspect of topological order is long-range entanglement.
   Long-range entanglement is manifest in a quantum spin system 
   when there are no local order parameters.
   In the case of a gapped long-range entangled ground state, one expects stability of the spectral gap under sufficiently local perturbations.
   In this section, following \cite{MichalakisZ,NachSYffs}, we give precise conditions of topological order and strength of perturbations for which the spectral gap
   of a short-ranged Hamiltonian is stable.

   \begin{ex}
   	We begin by showing a model without long-range entanglement. 
   	Consider the Ising chain 
   	\begin{equation*}
   	H_{[-L,L]} = -\sum_{x=-L}^{L-1} \sigma^3_x\sigma_{x+1}^3.
   	\end{equation*}
   	Let $\ket{+}$, $\ket{-}$ denote a orthonormal basis of eigenstates for $\sigma^3$.
   	Then, the states $\ket{P} = \ket{++\cdots +}$ and $ \ket{M} = \ket{--\cdots -}$ are ground states for $H_{[-L,L]}$.
   	A direct computation shows that 
   	\begin{align*}
   	\bra{P} \sigma_0^3 P\rangle  - \bra{M} \sigma_0^3 M \rangle   = 2 \quad \text{for all} \quad L.
   	\end{align*}
   	Therefore, the ground state of the Ising chain is not topologically ordered.
   \end{ex}
   
   Let $\Phi$ be a finite range $R$, uniformly bounded non-negative interaction
   and consider perturbations of the form 
   \begin{equation}\label{eqn:perturbationlinear}
   H_\Lambda(\epsilon) =  \sum_{X \subset \Lambda} \Phi(X) + \epsilon \Psi(X)
   \end{equation}
   where $\Psi$ is an interaction satisfying a finite $F$-norm for $F(r) = e^{a r}$ and $a>0$.
   Let $P_\Lambda$ denote the ground state projector for $H_\Lambda$
   and $P_\Lambda(\epsilon)$ denote the ground state projector for $H_\Lambda(\epsilon)$.
   Let $\Lambda_L = [0,L]^\nu \cap \Gamma$ and $B_x(r) \equiv \{ y \in \Gamma: \|x - y\|_{\infty} <r\}$ denote the box of radius $r$ about $x$.
   
   \begin{assumption}\label{ass:pbs}
   	For all $L$, the Hamiltonian $H_{\Lambda_L}$ is defined with periodic boundary conditions.
   \end{assumption}
   
   \begin{assumption}\label{ass:ff} (Frustration-free)
   	We assume $\ker H_{\Lambda_L}\neq \{0\}$ whenever $diam(\Lambda)\geq R$.
   \end{assumption}

   \begin{assumption}\label{ass:localgap}(Local gap)
   	We assume there is a $\gamma>0$ such that 
   	\begin{equation*} 
   	\inf \{ \operatorname{spec}(H_{B_x(r)}) \setminus \{0\}\} > \gamma
   	\end{equation*}
   	for all $x\in \Gamma$ and $r \geq R$.
   	
   \end{assumption}
   
   \begin{assumption}\label{ass:LTQO}(Local Topological Quantum Order - LTQO)
   	We assume there is a $q > 2(\nu +1)$ and $\alpha \in (0,1)$
   	such that for all $r \leq L^\alpha < L$ and $l$ such that  $1 \leq  l \leq L - r$,
   	and all self-adjoint $A \in \calA_{ B_{x}(r)}$ we have
   	\begin{equation}\label{eqn:ltqo}
   	\| P_{B_{x}(r+l)} A P_{B_{x}(r+l) } - \omega_{B_x(r+l)}(A) P_{B_x(r+l)} \| \leq C \|A\| l^{-q}.
   	\end{equation}
   	where 
   	$\omega_{\Lambda}(A) = \operatorname{Tr}(P_\Lambda A)/\operatorname{Tr}(P_\Lambda)$
   	and $C$ is a constant.
   \end{assumption}
   
   Equation \eqref{eqn:ltqo} can be interpreted as saying that local observables take approximately the same expectation for any ground state.
   More precisely, the left hand side of \eqref{eqn:ltqo} can be rewritten for $A$ self-adjoint 
   \begin{align}\label{eqn:LTQOexpectation}
   &\| P_{B_{x}(r+l)} A P_{B_{x}(r+l) } - \omega_{B_x(r+l)}(A) P_{B_x(r+l)} \|\\
   & \qquad \qquad =  \sup_{ \| \psi \| = 1} \abs{ \langle P_{B_{x}(r+l)}\psi,  A P_{B_{x}(r+l)} \psi\rangle - \omega_{B_{x}(r+l)}(A) \langle P_{B_{x}(r+l)}\psi,   P_{B_{x}(r+l)} \psi\rangle}\\
   & \qquad\qquad  = \sup_{ \substack{\|\Omega\| = 1, \\ \Omega \in \calG_{B_{x}(r+l)}}} 
   \abs{ \langle \Omega,  A \Omega\rangle - \omega_{B_{x}(r+l)}(A)}.
   \end{align}
   where $\calG_{B_{x}(r+l)} \equiv P_{B_{x}(r+l)} \calH_{B_{x}(r+l)}$ is the ground state space.
   
   \begin{lemma}\label{lem:LTQOff}
   	Suppose Assumptions \ref{ass:ff} and \ref{ass:LTQO} hold.
   	Then, 
   	\begin{equation*}
   	\omega(A) \equiv \lim_{L \ra \infty } \frac{1}{\dim \ker H_{\Lambda_L}} \operatorname{Tr}( P_{\Lambda_L} A) \quad \text{  exists for all } \quad A \in \calA.
   	\end{equation*}
   	Further, if $\rho$ is a frustration-free ground state then $\rho = \omega$.
   \end{lemma}

   \begin{proof}
   	Consider the increasing and exhaustive sequence given by $\Lambda_L$,
   	and let $\omega_L$ be an extension of the state  $\frac{1}{\dim \ker H_{\Lambda_L}} \operatorname{Tr}( P_{\Lambda_L} \cdot ) $ to $\calA$.
   	By compactness, there exists a converging subsequence we again call $\omega_L$ 
   	such that  $ \omega = \wslim_{L\ra \infty} \omega_L$.    
   	By the frustration-free condition, there exists $L$ such that if $L'>L$ then $ \omega_{L'}(\Phi(X)) = 0$ for all $X \in \mathcal{P}_0(\Lambda_L)$.
   	It follows that $\omega(\Phi(X)) = \lim_{L \ra \infty} \omega_L(\Phi(X)) = 0$
   	and that $\omega$ is a frustration-free ground state.
   	
   	Suppose $\rho$ is a frustration-free ground state
   	and, abusing of notation, let $\rho_L = \rho|_{\calA_{ \Lambda_L}}$ denote the restriction and the corresponding reduced density matrix.
   	By the frustration-free condition, $P_L \rho_L P_L = P_L \rho_L = \rho_L P_L = \rho_L$,
   	that is, $\rho_L$ is supported on the ground state space.
   	Let $ \rho_L = \sum c_k \ket{\Omega_k}\bra{\Omega_k}$ be a convex decomposition for $\rho_L$.
   	
   	Consider a self-adjoint local observable $A \in \calA_{loc}$. 
   	Then, for $L$ large enough that $supp(A) \subset \Lambda_L$ we have that $\rho(A) = \rho_L(A)$.
   	Since $A$ is self-adjoint we can apply the LTQO bound as in Assumption \ref{ass:LTQO} and \eqref{eqn:LTQOexpectation} to get
   	\begin{align*}
   	\abs{\omega(A) - \rho(A)}& = \lim_{L \ra \infty } \abs{ \omega_{L} (A)  - \rho_L(A) }\\
   	&\leq \lim_{L \ra \infty } \sum c_k \abs{\omega_L( A ) - \langle \Omega_k, A \Omega_k \rangle}\\
   	&\leq \lim_{L \ra \infty }  \sup_{\substack{\| \Omega\| = 1 \\ \psi \in \mathcal{G}_L}}
   	\abs{\omega_L( A ) - \langle \Omega, A \Omega \rangle} \qquad  \ra \quad  0.
   	\end{align*}
   	
   	Therefore, since any observable can be written as a linear combination of two self-adjoint observables
   	and since $\calA_{loc}$ is dense in $\calA$, $\rho = \omega$.
   	
   	Combining the above two arguments, we have shown that for the sequence $\{\omega_L\}_{L =1}^\infty$ 
   	every subsequence has a convergent subsequence to a frustration-free ground state.
   	Further, since there is a unique ground state satisfying the frustration-free condition, 
   	these subsequences must all converge to $\omega$.  
   	Therefore, $\wslim_{L\ra\infty} \omega_L = \omega.$
   \end{proof}
   
   Let $E_\Lambda(\epsilon) = \inf \operatorname{spec}(H_\Lambda(\epsilon))$.
   Generally, a perturbation will cause a splitting at the zero ground state energy of $H_\Lambda$.
   It is convenient to modify the definition of spectral gap in the perturbed Hamiltonian as
   \begin{equation*}
   \gamma_\delta(H_\Lambda(\epsilon)) \equiv \sup \{ \eta > 0 : (\delta, \delta +\eta) \cap \operatorname{spec}(H_\Lambda(\epsilon) - E_\Lambda(\epsilon) I ) = \emptyset \}.
   \end{equation*}

   \begin{thm}\label{thm:stablegap}\cite{MichalakisZ, NachSYffs}
   	Let $H_{\Lambda_L}(\epsilon)$ be as defined above
   	and suppose Assumptions \ref{ass:ff}, \ref{ass:localgap} and \ref{ass:LTQO} hold.
   	Then, for every $0 < \gamma_0 < \gamma$ there exists an $ \epsilon_0 >0$ such that if $\abs{\epsilon}<\epsilon_0$,
   	for sufficiently large $L$ we have
   	\begin{equation*}
   	\gamma_{\delta_L}(H_{\Lambda_L}(\epsilon)) \geq \gamma_0, \quad \mbox{ where } \quad \delta_L \leq C L^{-p} \mbox{ for some } p>0.
   	\end{equation*}
   \end{thm}

   \chapter[% 
   Kitaev models
   ]{% 
   	Kitaev's quantum double models
   }%
   \label{ch:qdoub}
   
   In this chapter we introduce the  Kitaev's  quantum double model for abelian group $G$ in the setting of quantum spin systems. 
   The Hamiltonian interaction terms are commuting projectors.
   The ground state space is frustration-free in the sense that local Hamiltonians have non-trivial kernel.
   When the underlying surface of the model is a surface of genus $g$, 
   the ground state degeneracy is a topological invariant depending only on the genus.
   There are no local order parameters in the sense that local observables take identical expectation values in the ground state 
   and the LTQO Assumption \ref{ass:LTQO} is satisfied.
   The algebraic structure describing the statistics of the elementary excitations of the model is 
   representation theory of the quantum double, $\operatorname{Rep}(\mathcal{D}(G))$.
   
   \section{The model}\label{sec:qdoub}
   
   \subsection{Dynamics}
   
   Let $G$ be a finite group and consider the bonds (or edges) $\calB$ of the square lattice $\ZZ^2$, i.e.\ the edges between 
   nearest neighbors of points (or vertices)  in $\mathbb{Z}^2$.
   We give $\calB$ an orientation by having edges either point up or right.
   To each edge $e \in \calB$ we associate a $\abs{G}$-dimensional Hilbert space with an orthonormal basis labeled by group elements 
   and denoted by $\ket{g}$. 
   We will use the notation $\bar{g}$ to denote the inverse element $g^{-1}$.
   In general, the model can be defined on any oriented metric graph, see~\cite{BombinMD,KitaevQD}. Reversing 
   the orientation on a given edge corresponds to the unitary transformation that maps $\ket{g}$  to $\ket{\bar{g}}$ in the state space of that
   edge. 
   Let $ \mathcal{P}_0(\calB)$ denote  the set of finite subsets of $\mathcal{B}$.
   
   Recall, the quantum spin system on $\Lambda$ is defined on the Hilbert space $ \calH_\Lambda := (\CC^{\abs{G}})^{\otimes \abs{\Lambda}} $
   and the algebra of observables is $\calA_\Lambda = M_{\abs{G}^{\abs{\Lambda}}}(\CC)$.
   If $ \Lambda_1 \subset \Lambda_2$ there is a natural inclusion
   $ i_{\Lambda_2, \Lambda_1}:\calA_{\Lambda_1} \hookrightarrow \calA_{\Lambda_2}$ mapping $ A \mapsto A\otimes I_{\Lambda_2\backslash \Lambda_1}$.
   The maps $i_{\Lambda_2, \Lambda_1}$ are isometric morphisms 
   so we will often abuse notation by identifying $i_{\Lambda_2,\Lambda_1}(A)$ with simply $A.$
   For this net of algebras, we define the local algebra of observables 
   and the quasi-local algebra of observables as, respectively,
   \begin{equation*}
   \calA_{loc} = \bigcup_{\Lambda \in \mathcal{P}_0(\calB) }\calA_\Lambda, \quad \text{and} \quad \calA = \overline{\calA_{loc}}^{\| \cdot \|}.
   \end{equation*}

   To define the model we specify the local Hamiltonians and the Heisenberg dynamics on $\calA$.
   The interaction terms of the local Hamiltonian are non-trivial only on certain subsets of $\calB$, called stars and plaquettes.
   We define a \emph{star} $v$ to be a set of four edges sharing a vertex.
   Similarly, a \emph{plaquette} $f$ is the set of four edges forming a unit square in the lattice.
   Interaction terms are defined for each star and plaquette by
   \begin{equation*}
   A_v \equiv \frac{1}{\abs{G}} \sum_{g\in G} A_v^g, \qquad \text{ and } \qquad  B_f \equiv B_f^e,
   \end{equation*}
   where the terms $A_v^g$ and $B_f^h$ are defined linearly by their action on an orthonormal basis, as shown in the following figures:
   \begin{figure}[h]
   	\centering
   	\includegraphics[width=0.45\textwidth]{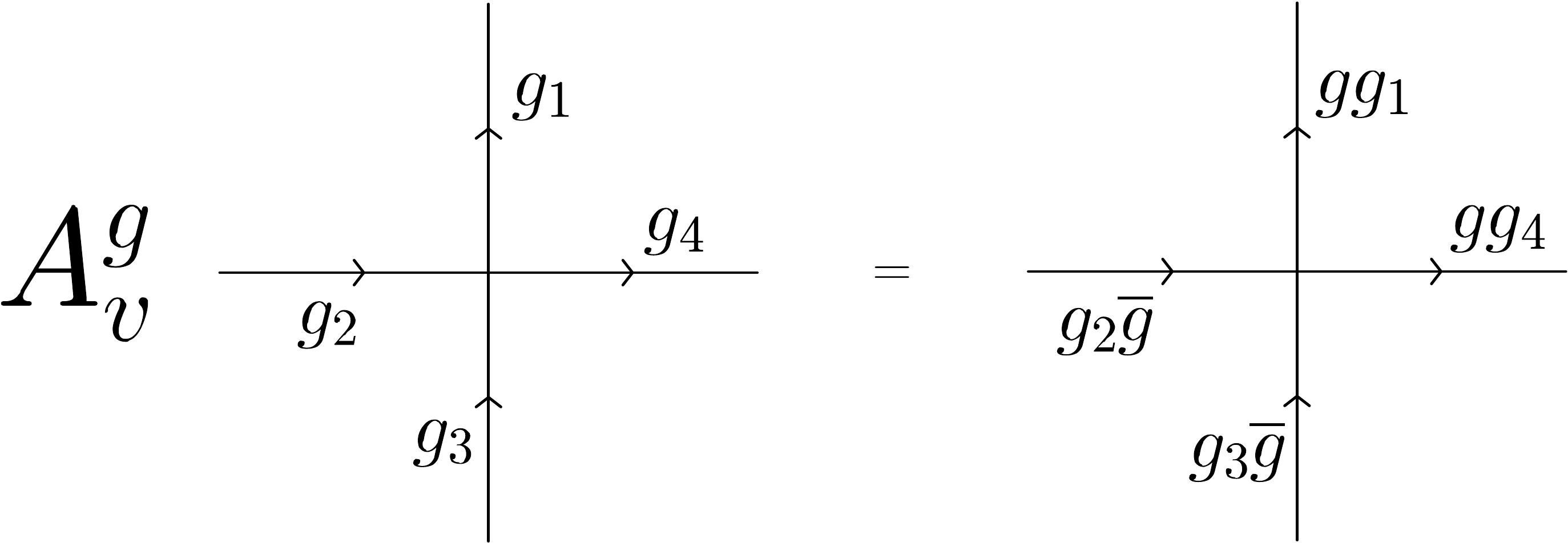}
   	\hfill
   	\includegraphics[width=0.45\textwidth]{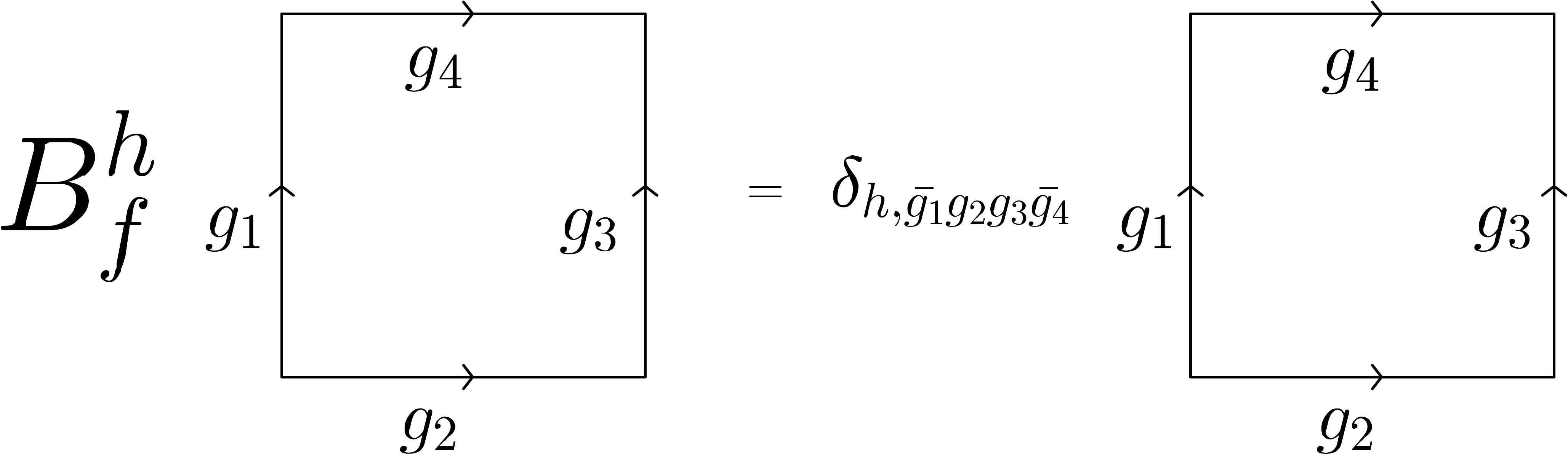}
   \end{figure}
   
   A direct computation gives that the operators $A_v$ and $B_f$ satisfy the following relations:
   \begin{align*}
   &A_v^g A_v^{g'} = A_v^{gg'},  & \left(A_v^{\bar{g}}\right)^* = A_v^g, \\
   &B_f^h B_f^{h'} = \delta_{h,h'} B_f^h,  &B_f^{h*} = B_f^h, \\
   &A_v^g B_f^h = B_f^{gh\bar{g}} A_v^g &\text{(if $v$ and $f$ share edges)}.
   \end{align*}
   In all other cases the operators commute.
   It follows that
   \begin{align*}
   A_v^2 &= \frac{1}{\abs{G}^2} \sum_{g,g' \in G} A_v^g A_v^{g'}= \frac{1}{\abs{G}} \sum_{g \in G} \left( \frac{1}{\abs{G}} \sum_{g'\in G} A_v^{ g g'}\right) = A_v
   \end{align*}
   and 
   \begin{equation*}
   A_v^* = \frac{1}{\abs{G}^2} \sum_{g \in G} A_v^{\bar{g}} = A_v.
   \end{equation*}
   A similar calculation as above holds for the plaquette terms.
   Thus, the interactions terms are mutually commuting projectors
   \begin{equation*}
   A_v = A^*_v=A_v^2, \quad \quad B_f = B_f^* = B_f^2, \quad \quad [A_v, B_f] = 0 \text{ for all } v,f.
   \end{equation*}
   
   We caution the reader that in the case of the toric code model (which corresponds to $G = \mathbb{Z}_2$) one usually shifts the 
   local interaction terms by a constant. This has no effect on the dynamics, but the algebraic relations are slightly different. Explicitly,
   the common convention is to define the toric code model in terms of star and plaquette operators $A_v^{tc}$ and $B_f^{tc}$ given
   by $2 A_v -I = A_v^{tc}$ and $ 2 B_f -I = B_f^{tc}$.
   
   For $\Lambda \subset \calB$ denote the subset of stars and plaquettes contained in $\Lambda$ as
   \begin{equation*}
   \mathcal{V}_\Lambda = \{ v \subset \Lambda: v \text{ is a star} \}, \qquad 
   \mathcal{F}_\Lambda = \{ f \subset \Lambda: f \text{ is a plaquette} \}.
   \end{equation*}
   If $\Lambda \in \mathcal{P}_0(\calB)$, the local Hamiltonians for the quantum double models defined by Kitaev \cite{KitaevQD} are given by
   \begin{equation*}
   \sum_{v\in \mathcal{V}_\Lambda} (I - A_v) + \sum_{f\in \mathcal{F}_\Lambda} (I - B_f) = H_\Lambda \in \calA_\Lambda.
   \end{equation*}
   Since the interaction terms are uniformly bounded and of finite range, 
   the existence of global dynamics $t \mapsto \tau_t \in \operatorname{Aut}(\calA)$ is readily established by Theorem \ref{thm:infvoldyn}.
   
   For our analysis it will be enough to consider squares $\Lambda_L\subset \calB$ consisting of all edges in $[-L,L]^2$.  
   We will denote $H_L = H_{\Lambda_L}$ and $\calH_L = \calH_{\Lambda_L}$.
   The generator of the dynamics is the closure of the operator
   \begin{equation*}
   \delta(A) = \lim_{L \ra \infty} [H_L, A],
   \end{equation*}
   where $\calA_{loc}$ is a core for $\delta$, and $\tau_t(A) = e^{i t \delta}(A) $ for all $A \in \calA_{loc}$.
   
   \begin{figure}
   	\includegraphics[width=0.4\textwidth]{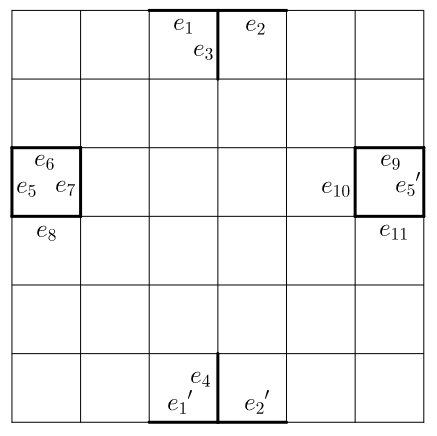}
   	\caption{The set $\Lambda_3$ is pictured. In the case of periodic boundary conditions we identify edges, $e_1 = e_1'$, $ e_2= e_2'$ and $e_5 = e_5'$.
   		The set $\mathcal{V}_3^{per}$ will contain the star $v = \{e_1, e_2, e_3,e_4\}$ and $\mathcal{F}_3^{per}$ will contain the plaquette $f = \{e_5, e_9, e_{10}, e_{11} \}.$ }
   	\label{fig:periodic}
   \end{figure}
   
   We will often consider the case of periodic boundary conditions for $\Lambda_L$,
   where the edges on the outer boundary of $\Lambda_L$ are identified.
   Let $\mathcal{F}_{L}^{per}$ and $ \mathcal{V}_{L}^{per}$ be the corresponding sets of stars and plaquettes, see Figure \ref{fig:periodic}.
   Define the local Hamiltonians with periodic boundary conditions as
   \begin{equation*}
   H_{L}^{per} = \sum_{v\in \mathcal{V}_{L}^{per}} (I - A_v) + \sum_{f\in \mathcal{F}_{L}^{per}} (I - B_f)
   \end{equation*}

   \subsection{Ground state space}
   We consider the subspace defined by 
   \begin{equation*}
   \mathcal{G}_\Lambda = \{ \psi \in \calH_{\Lambda} : A_v \psi = \psi, B_f \psi = \psi, \forall v \in \mathcal{V}_\Lambda, \forall f \in \mathcal{F}_\Lambda\}.
   \end{equation*}
   If the set $\mathcal{G}_\Lambda$ is non-trivial, that is it consists of more than the zero vector, then it will agree with the ground state space of $H_\Lambda$. 
   In the following, using the framework of lattice gauge theory \cite{Oeckl}, we show that the dimension of $\mathcal{G}_\Lambda$ coincides with the number of flat $G$-connections up to conjugacy of $\Lambda$.

   \begin{defn}
   	A $G$-connection is a map $c: \Lambda \ra G$ .
   \end{defn}
   
   There is a natural isomorphism $ \calH_L \cong \CC[ c: \Lambda_L \ra G]$ given by the map $c \mapsto \ket{c(x_1) c(x_2) \cdots c(x_{\abs{\Lambda}})}$
   where $x_i$ is an arbitrary enumeration of the edges in $\Lambda_L$.
   
   \begin{defn} 
   	A $G$-connection $c$ is called \emph{flat} if for every face $f = (x_1 x_2 x_3 x_4)$, labeled clockwise, we have 
   	\begin{equation}\label{eqn:trivialmonodromy}
   	\sigma_f( c(x_1))\sigma_f( c(x_2))\sigma_f( c(x_3))\sigma_f( c(x_4)) = e
   	\end{equation}
   	where $ \sigma_f( c(x)) =\left\{  \begin{array}{ll}
   	c(x) & \mbox{ if x  is oriented clockwise wrt f }\\
   	\bar{c}(x) & \mbox{ otherwise } 
   	\end{array}\right .
   	$
   	
   	The product on the left hand side of equation \ref{eqn:trivialmonodromy} will be referred to as the monodromy of $c$ about the face $f$.
   \end{defn}
   
   \begin{figure}
   	\includegraphics[width = .4\textwidth]{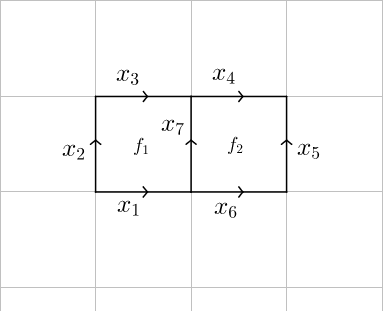}
   	\caption{ The faces $f_1 = ( x_1 x_2 x_3 x_7)$ and $f_2 = ( x_4 x_5 x_6 x_7)$ are labeled clockwise and share the edge $e_7$.}
   	\label{fig:monodromy}
   \end{figure}
   
   Let $c$ be a flat connection.
   Let $f_1 = ( x_1 x_2 x_3 x_7)$ and $f_2 = ( x_4 x_5 x_6 x_7)$ be faces sharing the edge $x_7$ and consider the
   closed loop $\gamma = (x_1 x_2 \cdots x_6)$, see Figure~\ref{fig:monodromy}.
   The orientation of $x_7$ in $f_1$ is opposite that of $f_2$ such that $\sigma_{f_1}( c(x_7)) = \sigma_{f_2}( c(x_7))^{-1}$.
   It follows that 
   \begin{align*}
   \sigma_\gamma( c(x_1))\sigma_\gamma( c(x_2)) \cdots \sigma_\gamma( c(x_6)) &= \sigma_{f_1}( c(x_1))\sigma_{f_1}( c(x_2))\sigma_{f_1}( c(x_3))  \sigma_{f_1}(c(x_7)) \\
   & \quad \quad \quad \quad \sigma_{f_2}({c}(x_7)) \sigma_{f_2}( c(x_4)) \sigma_{f_2}( c(x_5)) \sigma_{f_2}( c(x_6))\\
   & = e.
   \end{align*}
   The following lemma readily follows:
   \begin{lemma}
   	Let $c$ be a flat $G$-connection.
   	Then, for any simple closed loop in the lattice $\gamma = (x_1 x_2 \ldots x_n)$ we have 
   	\begin{equation*}
   	\sigma_\gamma(c(x_1)) \cdots \sigma_\gamma(c(x_n)) = e.
   	\end{equation*}
   \end{lemma}
   
   \begin{defn}
   	Let $V_\Lambda$ be the vertex set of $\Lambda$ and $h: V_\Lambda \ra G$.
   	A \emph{gauge transformation} is a map $\Delta_h: \CC[c : \Lambda \ra G] \ra \CC[c : \Lambda \ra G]$ 
   	defined by 
   	\begin{equation*}
   	(\Delta_h c)(x) = h(v_1) c(x) h(v_2)^{-1}
   	\end{equation*}
   	where $ x =(v_1 v_2) $ is a oriented from $v_1$ to $ v_2$.
   \end{defn}

   \begin{lemma}\label{lem:gauge}
   	Let $ \calH_{\Lambda}^f$ denote the subspace generated by flat connections.
   	If $\Delta_h$ is a gauge transformation then $\Delta_h(\calH_{\Lambda}^f) = \calH_{\Lambda}^f$.
   	In particular, $\Delta_h$ is invertible and $(\Delta_h)^{-1} = \Delta_{h^{-1}}$.
   \end{lemma}
   
   \begin{proof}
   	Let $c$ be a $G$-connection and $f = (x_1 x_2 x_3 x_4)$ be a face in $\Lambda$.
   	If $x_i$ is oriented clockwise with respect to $f$ then we write $ x_i = ( v_i v_{i+1})$.
   	If $x_i$ is oriented counter clockwise with respect to $f$ then we write $x_i = (v_{i+1} v_i)$.
   	In both cases $\sigma_f (\Delta_h c(x_i) ) = h(v_i) c(x_i) h(v_{i+1})^{-1}$.
   	It follows that, 
   	\begin{align*}
   	\sigma_f ( \Delta_hc(x_1))\sigma_f ( \Delta_hc(x_2))\sigma_f ( \Delta_hc(x_3))\sigma_f ( \Delta_hc(x_4)) = e.
   	\end{align*}
   	Therefore, $\Delta_h c $ is a flat $G$-connection.
   	Thus, $\Delta_h$ an invertible linear map on $ \calH_{\Lambda}^f$
   	where $ \Delta_{h^{-1}} = (\Delta_h)^{-1}$.
   \end{proof}
   
   The notions of flat $G$-connections and gauge transformations are directly related to the interaction terms.
   Indeed, a connection $c$ is flat if and only if 
   the state $\ket{c} = \ket{c(x_1) \cdots c(x_{\abs{\Lambda}})}$ satisfies $ B_f \ket{c} = \ket{c}$ for all $ f \in \mathcal{F}_{\Lambda}$.
   Let $P_c = \ket{c} \bra{c}$, then we have that $\prod_{f \in \mathcal{F}_\Lambda} B_f P_c = P_c$.
   Let $\Delta_h$ be a gauge transformation.
   A direct calculation shows that in the case of periodic boundary conditions
   \begin{equation*}
   \Delta_h  = \prod_{v \in \mathcal{V}_L^{per}} A_v^{h(v)}.
   \end{equation*}
   Define an equivalence relation on flat connections by $ c \sim c'$ 
   if there exists a gauge transformation $\Delta_h$ such that $\Delta_h c = c'$.
   In particular, our construction shows that the ground state degeneracy of $H_{\Lambda}^{per}$ is equal to the number of flat $G$-connections up to gauge transformation.
   For open boundaries conditions, the  dimension grows exponentially with the perimeter of $\Lambda$ since the edges lying on the perimeter of $\Lambda$ do not belong to any star in $\mathcal{V}_\Lambda$,
   and thus, the product $\prod_{v \in \mathcal{V}_\Lambda} A_v^{h(v)}$  does not uniquely define a gauge transformation.

   The argument above generalizes for any cellulation of a closed orientable surface, see \cite{Oeckl} for a precise definition.
   In this case, the dimension is a topological invariant.
   \begin{thm} \cite{KitaevQD}
   	Let $ \Lambda_g$ denote a cellulation of a closed  surface of genus $g$. 
   	Then, 
   	\begin{equation*}
   	\ker(H_{\Lambda_g})  =  \{ \psi \in \calH_{{\Lambda_g}}:  A_f \psi = \psi, B_p \psi = \psi, \forall f \in \mathcal{F}_{\Lambda_g},\forall v \in\mathcal{V}_{\Lambda_g} \}
   	\end{equation*}
   	and  $\dim(\ker(H_{\Lambda_g})) = \dim( \mathcal{H}_{\Lambda_g}^f/\sim).$
   	Furthermore, $\dim(\ker(H_{\Lambda_g})) $ is independent of the choice of cellulation $\Lambda_g$
   	and if $G$ is a finite abelian group then $ \dim(\ker(H_\Lambda))  = \abs{G}^{2g}$.
   \end{thm}
   
   Further, \cite{FreedmanM, BravyiK} considered  cellulations of $\RR P^2$ for the toric code model, $G =\ZZ_2$.
   They showed that the ground state space encodes exactly one qubit,  $\dim  \mathcal{G}_{\Lambda^{\RR P^2}}  = 2$.

   \subsection{Ribbon operators}
   
   \begin{figure}
   	\includegraphics[width=0.4\textwidth]{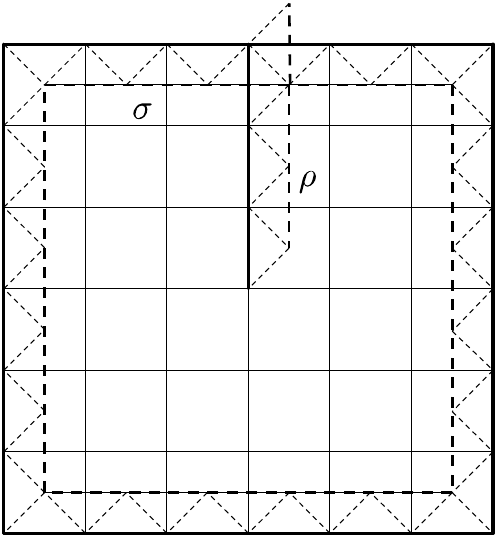}
   	\hfill
   	\includegraphics[width=0.4\textwidth]{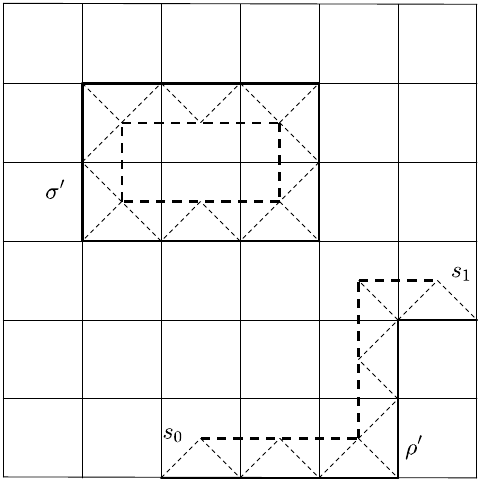}
   	\caption{The region $\Lambda_3$ is depicted with typical configurations of ribbons.
   		On the left, $\rho$ connects a site in $\mathcal{S}_3$ to a site on the boundary,
   		and intersects $\sigma = \partial \Lambda_3$, the boundary ribbon of $\Lambda_3$.
   		On the right, $\rho'$ is an open ribbon connecting site  $s_0 \in \mathcal{S}_3$ to $s_1 \in \mathcal{S}_3$, while $\sigma'$ is a closed ribbon.}
   	\label{fig:ribbons}
   \end{figure}
   
   We now restrict to $G$ being a finite \emph{abelian} group.
   
   By abuse of notation, use $v$ and $f$ to also denote a vertex and face of the lattice $\ZZ^2$.
   A \emph{site} is a pair $s=(v,f)$ of a vertex $v \in \ZZ^2$ and neighboring face $f$. 
   Let $\mathcal{S}_L$ denote the set of all sites $s=(v,f)$ such that $v \in \ZZ^2 \cap [-L,L]^2$ and the corresponding face $f \in \mathcal{F}_L$.
   We say that a site $s=(v,f)$ is on the boundary of $\Lambda_L$ if $v \in \ZZ^2 \cap [-L,L]^2$ and 
   the corresponding face $f \in \mathcal{F}_{L+1} \setminus \mathcal{F}_{L}$.
   As we will see, excitations of the model are located at sites.
   A \emph{ribbon} $\rho$ is a sequence of adjacent sites connecting two sites $s_0$ and $s_1$. 
   We assume ribbons avoid self-crossing 
   and label $\partial_0\rho =  s_0$ as the start of the ribbon and $\partial_1 \rho = s_1$ as its end. 
   In particular, note that ribbons carry a direction (see~\cite{BombinMD} how this relates to the direction of the lattice).
   We also assume that ribbons have at least length two.
   A ribbon is said to be \emph{open} if $s_0 \neq s_1$ and \emph{closed} if $s_0 = s_1$, see Figure~\ref{fig:ribbons}.

   \begin{figure}
   	\begin{center}
   		\includegraphics[width=\textwidth]{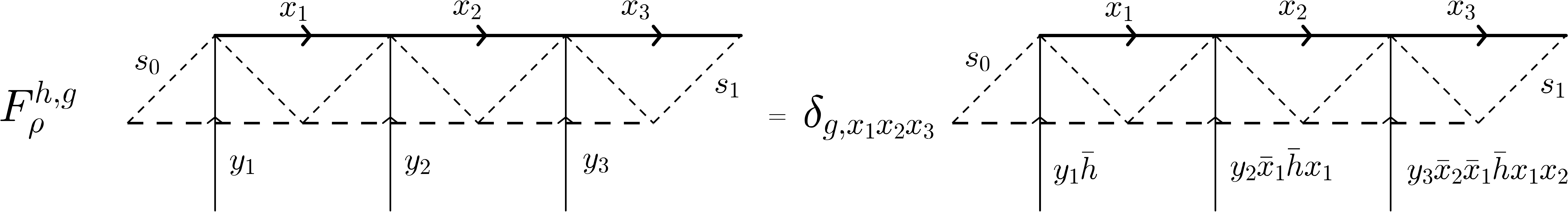}
   	\end{center}
   	\caption{Definition of the ribbon operator $F_\rho^{h,g}$.}
   	\label{fig:defnrib}
   \end{figure}
   
   For any ribbon $\rho$ and $g,h \in G$ the ribbon operator $F_\rho^{g,h}$ is defined as in Figure~\ref{fig:defnrib}.
   The ribbon operators can also be defined recursively as concatenations of elementary triangle operators, see \cite{BombinMD}. 
   If $\rho_0$ and $\rho_1$ are two ribbons such that $\partial_1 \rho_0 = \partial_0\rho_1$ 
   then for the ribbon $\rho = \rho_0\rho_1$,  where the product is defined by concatenation of ribbons,
   the ribbon operator satisfies $F_\rho^{h,g}  = \sum_{k\in G} F_{\rho_0}^{h,k} F_{\rho_1}^{\bar{k}  h k, \bar{k} g}$.
   
   The operators $F_\rho^{h,g}$ can be thought of as creating excitations at the ribbon endpoints.
   However, in general $F^{h,g}_\rho$ will yield a superposition of different excitation types, and it is more convenient to choose a different basis.
   Essentially, what one does is to decompose the space of excitations as invariant subspaces with respect to a \emph{local} action of the quantum double symmetry $\mathcal{D}(G)$ at each site.
   This symmetry is implemented by the star and plaquette operators $A^g$ and $B^h$~\cite{BombinMD,KitaevQD}.
   In this new basis the ribbon operators are labeled by pairs $(\chi, c) \in \widehat{G} \times G$, where $\widehat{G}$ is the group of characters, and we define
   \begin{equation}\label{eqn:ribbondef}
   F_\rho^{\chi,c} := \sum_{g \in G} \overline{\chi}(g) F_\rho^{\bar{c}, g}.
   \end{equation}
   If $\rho$ consists of a single edge, then the family of ribbon operators generate the local algebra on that edge.
   Similarly, for any finite subset $\Lambda$, the family of ribbon operators supported in $\Lambda$ generate the local observable algebra $\calA_\Lambda$.
   
   We list some properties of the ribbon operators that we will use frequently.
   These properties can be verified by direct computation, see~\cite{BombinMD,FiedlerN,KitaevQD}:
   \begin{enumerate}
   	\item For operators acting along the same ribbon:
   	\begin{equation}\label{eqn:ribprop1}
   	F^{\chi,c }_\rho F^{\sigma, d}_\rho = F_{\rho}^{\chi\sigma,cd} \qquad \text{and} \qquad (F^{\chi,c}_\rho)^* = F^{\bar{\chi}, \bar{c}}_\rho.
   	\end{equation}
   	\item If $\rho$ is an open ribbon connecting sites $\partial_0 \rho = (v_0, f_0)$ and $ \partial_1 \rho= (v_1,f_1)$ then for all $k \in G$ we have
   	\begin{align}
   	A_{v_0}^k F_\rho^{\chi,c} &= \chi(k) F_\rho^{\chi,c} A_{v_0}^k, 
   	& A_{v_1}^k F_\rho^{\chi,c}&= \bar{\chi}(k) F_\rho^{\chi,c} A_{v_1}^k, 
   	\label{eqn:ribstarrel} \\
   	B_{f_0}^k F_\rho^{\chi,c} &= F_\rho^{\chi,c} B_{f_0}^{k\bar{c}}, 
   	& B_{f_1}^k F_\rho^{\chi,c} &= F_\rho^{\chi,c} B_{f_1}^{ck}
   	.\label{eqn:ribplaqrel}	
   	\end{align}
   	In all other cases, the star and plaquette interaction terms commute with the open ribbon operators.
   	From equations \eqref{eqn:ribstarrel} and \eqref{eqn:ribplaqrel}, we can compute 
   	the commutation relation with the local Hamiltonian $H_L$ as
   	\begin{equation}\label{eqn:ribHamrel}
   	[H_L, F_{\rho}^{\chi,c}] = F_{\rho}^{\chi,c} \bigg(B_{f_0} - B_{f_0}^{\bar{c}}  + B_{f_1} - B_{f_1}^c 
   	+ \sum_{k \in G} \left( 1 - \chi(k) \right) A_{v_0}^k + \left(1 - \bar{ \chi}(k)\right) A_{v_1}^k \bigg).
   	\end{equation}
   	\item Let $\Omega \in \calG_L$ be a ground state and $\rho$ be an open ribbon.  We can compute the energy introduced by the ribbon operators from the relations \eqref{eqn:ribstarrel}, \eqref{eqn:ribplaqrel}, and \eqref{eqn:ribHamrel},
   	\begin{equation}\label{eqn:ribenergy}
   	H_L F_\rho^{\chi,c} \Omega = C_\rho ( 2- \delta_{\chi, \iota} - \delta_{c, e}  ) F_{\rho}^{\chi,c}  \Omega.
   	\end{equation}
   	where 
   	\[ C_\rho = 
   	\left\{ \begin{array}{ll}
   	2 & \mbox{ if } \partial_i \rho \in \mathcal{S}_L \text{ for }   i = 0,1 \\
   	1 & \mbox{ if } \partial_i \rho \in \mathcal{S}_L, \  \partial_{i+1}\rho \notin \mathcal{S}_L \\
   	0 & \mbox{ if } \partial_i \rho \notin \mathcal{S}_L \text{ for }   i=0,1.
   	\end{array}\right. \]
   	Thus, $\rho$ generates \emph{excitations} at its endpoints.
   	\item If $\rho$ is a closed ribbon then for all $\Omega \in \calG_{L}$
   	\begin{equation}\label{eqn:closedribbon}
   	F_\rho^{\chi,c} \Omega = \Omega \qquad \quad  \forall (\chi, c) \in \widehat{G}\times G
   	\end{equation}
   	and 
   	\begin{equation}\label{eqn:ribclosedcom}
   	[ F_\rho^{\chi,c}, A_v^g] = [ F_\rho^{\chi,c}, B_f^h] = 0 \qquad \quad \forall v\in \mathcal{V}_\mathcal{B}, f\in \mathcal{F}_\mathcal{B}.
   	\end{equation}
   	\item If $\rho$ is a concatenation of $\rho_0$ and $\rho_1$, that is, $\partial_1 \rho_0 = \partial_0 \rho_1$  and we write $\rho = \rho_0 \rho_1$ then the ribbon operators obey
   	\begin{equation}\label{eqn:ribbonconcatenation}
   	F_\rho^{\chi, c} = F_{\rho_0}^{\chi,c} F_{\rho_1}^{\chi,c}.
   	\end{equation}
   	The above operation
   	\item A complete set of eigenvectors of $H_L$ for $\calH_L$ is,
   	\begin{equation*}
   	\left\{ \prod_{i} F_{\rho_i}^{\chi_i, c_i} \Omega :  \Omega\in \calG_L, \forall \mbox{ list of triples } (\chi_i, c_i, \rho_i)  \mbox{ where } \chi_i\in \widehat{G}, c \in G, \rho_i \mbox{ is a ribbon } \right\}.
   	\end{equation*}
   	\item  When two ribbons intersect once (as in left Figure~\ref{fig:ribbons}),
   	\begin{equation}\label{eqn:ribbonrelation}
   	F_\rho^{\chi,c} F_\sigma^{\xi,d} = \chi(d)\overline{\xi}(c)  F_\sigma^{\xi,d}F_\rho^{\chi,c}.
   	\end{equation}
   	In the case of multiple crossings, one can induct on the formula above by decomposing the ribbons $\rho$ and $\sigma$ into sections where only one crossing occurs using the concatenation formula \eqref{eqn:ribbonconcatenation}.
   	\item  Ribbon operators satisfy \emph{path independence} in the ground state, that is, 
   	if $\rho$ and $ \sigma $ are ribbons such that $\partial_i \rho = \partial_i\sigma$ for $ i= 0,1$ then
   	\begin{equation}\label{eqn:ribbonpathind}
   	F_\rho^{\chi,c}\Omega  = F_\sigma^{\chi,c} \Omega \qquad \text{ for all } \qquad \Omega\in\calG_L.
   	\end{equation}
   \end{enumerate}
   
   From the properties of ribbon operators above, 
   we can see the elementary excitations formed at the endpoints of the ribbons are \emph{abelian anyons},
   that is, a full exchange of particles results in multiplication by a phase.
   Indeed, let $\Omega$ be a ground state of $H_\Lambda$.
   Consider open ribbons $\rho$ and $ \sigma$ and suppose $\rho'$ is a closed ribbon 
   connecting $ \partial_1\rho$ to itself such that $\partial_1 \sigma$ is in interior of the region enclosed by $\rho'$. 
   Let $\psi = F_\rho^{\chi,c} F_\sigma^{\xi,d} \Omega$ and consider the action of braiding determined by applying $F_{\rho'} $.
   From properties \eqref{eqn:ribbonrelation} and \eqref{eqn:ribbonpathind} we have
   \begin{align*}
   F_{\rho'}^{\chi,c} \psi & = 	F_{\rho'}^{\chi,c} F_\rho^{\chi,c} F_\sigma^{\xi,d} \Omega 
   =  \chi(d)\bar{\xi}(c) F_\rho^{\chi,c} F_\sigma^{\xi,d} F_{\rho'}^{chi,c}  \Omega\\
   & =  \chi(d)\bar{\xi}(c) \psi.
   \end{align*}

   \section{Topological order}
   
   The notion of topological order in quantum many body systems still does not have a adequate mathematical formulation.
   In this section, we review two notions of topological order that the quantum double models satisfy, namely, long-range entanglement and  the topological entanglement entropy.
   
   The following lemma and corollary  show that to distinguish between ground states, 
   observables need to be supported on regions at least the size of $\operatorname{diam}\Lambda -2$.
   In the language of stabilizer codes (see \cite{NielsenC}), this shows that the ground state of the quantum double model has a macroscopic code distance.
   In particular,  the Local Topological Quantum Order Assumption \ref{ass:LTQO} is satisfied.
   
   Let $c$ be a $G$-connection and define the projector $P_c = \ket{c}\bra{c}$.
   Recall that $\ket{c}$ forms an orthonormal basis of $\calH_{L}$.
   Let $\Delta_h$ be a gauge transformation and recall that in the case of periodic boundary conditions
   $\Delta_h = \prod_{v} A_v^{h(v)}$ and $(\Delta_h)^{-1} = \Delta_{h^{-1}}$.
   It follows that the operators $\Delta_h P_c$ span the observable algebra $\calA_{ L}$.

   \begin{lemma}(Lemma 12.1.2, \cite{NaaijkensD})
   	For all $l > 2$ and for all $\Omega \in \mathcal{G}^{per}_{\Lambda_{L+l}}$ we have that 
   	\begin{equation*}
   	\langle \Omega, \Delta_h P_c \Omega \rangle = \left\{ \begin{array}{ll}
   	1/\dim \ker H^{per}_L & \mbox{ if } c \text{ is a flat connection} \\
   	0 & \mbox{ otherwise }
   	\end{array}\right.
   	\end{equation*}
   \end{lemma}
   
   \begin{proof}
   	Let $c$ be a $G$-connection.  
   	If $c$ is not flat then there is some face $f$ for which the monodromy is not trivial.
   	It follows that $B_f P_c = 0$.
   	Since $\Omega$ satisfies $B_f \Omega = \Omega$ we have that
   	\begin{equation*}
   	\langle \Omega, \Delta_h P_c \Omega\rangle = \langle \Omega, \Delta_h P_c B_f \Omega\rangle  = 0.
   	\end{equation*}
   	
   	Let $c_1 $ and $c_2$ be flat connections.
   	It follows that there is a gauge transformation $\Delta_h = \prod_{v \in \mathcal{V}_\Lambda} A_v^{h(v)}$ such that $\Delta_h P_{c_1} = P_{c_2}$.
   	Notice that $A_{v}^{g'} A_v = A_v$ so that $ A_{v}^{g'} \Omega = \Omega$.
   	Thus,
   	\begin{align*}
   	\langle \Omega,  P_{c_2} \Omega\rangle &= \langle \Omega, \Delta_h P_{c_1}\Omega\rangle
   	= \langle \Omega,  \prod_{v \in \mathcal{V}_\Lambda}  A_v^{h(v)} P_{c_1}\Omega\rangle
   	= \langle \Omega,  P_{c_1} \Omega\rangle.
   	\end{align*}

   	Consider the projectors $ \prod_{f \in \mathcal{F}_\Lambda} B_f $ and $ \sum_{c: \text{ flat }} P_c$.
   	A direct calculation shows that $ P_c \leq \prod_{f \in \mathcal{F}^{per}_{L}} B_f $ for all $c$ flat
   	and $ B_f \leq  \sum_{c: \text{ flat }} P_c$ for all $f$.
   	Therefore,
   	\begin{equation*}
   	\prod_{f \in \mathcal{F}^{per}_{L}} B_f = \sum_{c: \text{ flat }} P_c.
   	\end{equation*}
   	The number of flat $G$-connections is counted by the $ \dim \ker H^{per}_L$.
   	
   	Therefore, $\omega(P_c) = 1/\dim \ker H^{per}_L$.
   \end{proof}

   \begin{cor}\label{cor:qdLTQO}
   	The Kitaev's quantum double models satisfy the Local Topological Quantum Order Assumption \ref{ass:LTQO} for all $q >6$ and $ \alpha \in (0,1)$.
   	Furthermore, the right hand side of \eqref{eqn:ltqo} in Assumption \ref{ass:LTQO} can be replaced with $\|A\| h(l)$ where $h$ is the step function
   	$h(l) = 2$ if $l \leq 2$ and $ h(l) =0$ if $x >2$.
   \end{cor}

   The area law for entanglement entropy is expected to be satisfied generally in ground states of gapped Hamiltonians. 
   For topologically ordered states in the thermodynamic limit, there exists a universal constant correction term to the area law, $S(\rho) \simeq C \abs{\partial \Lambda} - \gamma$, where $\gamma$ is called the topological entanglement entropy (TEE) and $\Lambda$ is a disk shaped region  \cite{KitaevP}.
   The TEE is related to the number of distinct superselection sectors \cite{NaaijkensKL}
   and for the abelian quantum double models this relation is precise.
   If $\mathcal{D}$ is the total quantum dimension, which in the case of abelian anyons $\mathcal{D}$ counts the number of superselection sectors, then  $\gamma =  \log \mathcal{D}$.
   
   Suppose $G =\ZZ_d$ be the finite cyclic group of order $d$.
   Let $\Lambda = [-L,L] \cap \calB$ be a rectangular region in the plane and let $ \calG_{\Lambda}^{per}$ denote the ground state space of $H_\Lambda^{per}$.

   \begin{lemma}\label{lem:qdTEE}(\cite{HammaIZ, BachmannTEE})
   	Let $X \subset \Lambda$ be a simple region.
   	Then,
   	\begin{equation*} 
   	S(X) = (\abs{\partial{X}} - 1) \log d.
   	\end{equation*}
   \end{lemma}

   \section{Infinite system}
   
   In the thermodynamic limit, it is expected that single excitation states will exist. 
   To construct a single excitation state, consider a procedure where an excitation pair is created from the ground state by application of a ribbon operator.
   By extending one endpoint of the open ribbon operator to infinity, one of the pair excitations is effectively moved off to infinity.
   Naaijkens \cite{Naaijkens11} showed how to make the above procedure rigorous in the operator algebraic framework of infinite quantum spin systems. 
   In particular, the single excitation states with different corresponding charge types belong to distinct superselection sectors.
   The superselection structure corresponding to all single excitation states can be classified 
   as the complete set of localized endomorphisms for the model \cite{NaaijkensKL}.
   Analysis of the superselection structure recovers a faithful and full functor to the modular tensor category $\operatorname{Rep}(\mathcal{D}(G))$ \cite{Naaijkens11}.
   In the following section, we discuss the aforementioned results.
   
   \subsection{Frustation-free ground state}
   
   Let $\Omega_L\in \calG_L$ be a sequence of finite volume ground states of $H_L$.
   A frustration-free ground state of the quantum double model $\omega^0$
   can be constructed as follows.
   Consider a family of states $\{\omega_L\}_{L=2}^\infty$ as $L\ra \infty$, 
   where $ \omega_L$ is an arbitrary extension
   of the state $\langle \Omega_L , \cdot \ \Omega_L  \rangle$ to the quasi-local algebra $\calA$ (in particular, we could choose a product state).
   By arguments similar to \ref{lem:LTQOff}, the sequence $\omega_L(A)$ is eventually constant for any local observable $A$.
   Thus, the limit $\omega_0 \equiv \lim_{L \ra \infty } \omega_L$ exists.
   For any $v$ and $f$, choose $L$ large enough such that $v \in \mathcal{V}_{\Lambda_L}$ and $f \in \mathcal{F}_{\Lambda_L}$.
   Since $\omega_L$ is a ground state for the finite model it follows that $\omega^0(I - A_v) = \omega_L(I - A_v) = 0$ and $\omega^0(I - B_f) = \omega_L(I - B_f) = 0$.
   This is exactly the  frustration-free or zero-energy property, Definition \ref{defn:ff}.

   \begin{prop}\label{prop:qdffgs}(\cite{AlickiFH,Naaijkens11,FiedlerN})
   	Let $\omega^0$ be the frustration-free ground state of the quantum double model obtained as above.
   	Then,
   	\begin{enumerate}
   		\item if $\omega$ is a frustration-free ground state then $\omega = \omega^0$,
   		\item $\omega^0$ is a pure state,
   		\item Let $(\pi_0, \Omega_0, \calH_0)$ be a GNS-representation for $\omega^0$ and $H_0$ be the GNS Hamiltonian.
   		Then, $spec(H_0) = 2 \ZZ^{\geq 0}$ with a simple ground state eigenvector $\Omega_0$.
   	\end{enumerate}
   \end{prop}
   
   The existence and uniqueness of the frustration-free ground state is a general property for frustration-free Hamiltonians
   satisfying a local topological quantum order condition, Lemma \ref{lem:LTQOff}.
   The last property follows from an application of strong resolvent convergence, Lemma \ref{lem:SRC}, 
   and shows that the frustration-free ground state is a gapped ground state.
   %As a consequence of local indistinguishablity, the gap is stable under local perturbations \cite{BravyiHM,MichalakisZ}.
   
   \subsection{Single excitation states}
   
   We now come to states in the thermodynamic limit that describe a single excitation.
   Such states may be constructed on the quasi-local algebra by moving one of the excitations in a pair off to infinity.
   Let $\rho$ be a ribbon extending to infinity such that $\partial_0\rho = s$ and $\partial_1 \rho = \infty$, where $\partial_1 \rho = \infty$ means that the ribbon goes to infinity in any direction.
   We assume that it does not ``loop back'', in the sense that if $\rho_n$ is the ribbon consisting of the first $n$ parts of $\rho$, then for any fixed point in the lattice, the distance to the endpoint of $\rho_n$ that is not fixed goes to infinity as $n$ goes to infinity.
   
   We denote $\rho_L = \rho \cap \Lambda_L$.
   Define the state $\omega^{\chi,c}_s$  on $\calA_{loc}$, and its unique continuous extension to $\calA$, by
   \begin{equation}\label{eqn:singleexcitation}
   \omega_s^{\chi,c} (A) := \lim_{L\ra \infty} \langle F_{\rho_L}^{\chi,c } \Omega_L, A   F_{\rho_L}^{\chi,c } \Omega_L\rangle.
   \end{equation}
   
   The limit converges because the sequence is eventually constant for fixed local $A$.
   That is, by concatenation \eqref{eqn:ribbonconcatenation} and unitarity in the ribbon operators, there exists $L>0$ such that for all $L' >L$ we have 
   $(F_{\rho_{L'}}^{\chi,c})^* A F_{\rho_{L'}}^{\chi,c} = (F_{\rho_{L}}^{\chi,c})^* A F_{\rho_{L}}^{\chi,c}$,
   and by local indistinguishability, the state is independent of the choice of sequence $\Omega_L$. 
   By path independence in the ground state, the state $\omega_s^{\chi,c}$ is also independent of the 
   path that $\rho$ takes to infinity and depends only on the basepoint $s$.
   
   \subsection{Superselection structure}

   For the quantum double models the relevant superselection criterion is as follows.
   Let $\Lambda \subset \mathcal{B}$ be an infinite cone region (the precise shape is not that important).
   We consider representations $\pi$ which satisfy the following criterion for \emph{any} such $\Lambda$:
   \begin{equation}\label{eqn:conecrit}
   \pi_0 \upharpoonright \calA_{\Lambda^c} \cong \pi \upharpoonright \calA_{\Lambda^c}.
   \end{equation}
   Here $\pi_0$ is the GNS representation of the frustration-free ground state and $\pi_0 \upharpoonright \calA_{\Lambda^c}$ means that we restrict the representation to $\calA_{\Lambda^c}$, the $C^*$-algebra generated by all local observables supported outside $\Lambda$.
   Physically, to detect the charge of a state in the representation $\pi$, one needs to measure the value of a ``Wilson loop''.
   If such loops around the charge are not allowed (as in the selection criterion, due to the absence of the cone), the charge cannot be detected.
   
   The superselection structure of the quantum double model can be analyzed in the same spirit as the DHR program.
   The sector structure is summarized in the following proposition
   
   \begin{prop}\label{prop:singleexc}(\cite{FiedlerN,Naaijkens11})
   	Let $ (\pi_s^{\chi,c}, \Omega_s^{\chi,c}, \calH_s^{\chi,c})$ be the GNS triple for $\omega_s^{\chi,c}$.
   	Then,
   	\begin{enumerate}
   		\item $\pi_s^{\chi,c}$ are irreducible representations satisfying the criterion~\eqref{eqn:conecrit},
   		\item $\pi_{s}^{\chi,c} \cong \pi_{s'}^{\chi,c} $,
   		\item if $(\chi,c) \neq (\chi',c')$ then $\pi_s^{\chi,c} $ and $\pi_s^{\chi',c'}$ belong to different superselection sectors (and hence are inequivalent),
   		\item if $\pi$ is irreducible and satisfies \eqref{eqn:conecrit} then there exists $\rho$ and $(\chi, c)$ such that $\pi \cong \pi_s^{\chi,c}$.
   	\end{enumerate}
   \end{prop}
   Pushing this analysis further, all properties of the charges such as their fusion and braiding rules can be recovered~\cite{FiedlerN}.
   It follows that the structure is completely described by the representation theory of the quantum double, $\operatorname{Rep}(\mathcal{D}(G))$.
   It is interesting to see that the charge superselection structure is closely related to the classification of ground states of the quantum double, as will become even clearer in the next chapter.

   \part{Results}
   
   \chapter[% 
   Ground states for Kitaev's models
   ]{% 
   	The complete set of infinite volume ground states for Kitaev's abelian quantum double models
   }%
   \label{ch:gs}
	
	We study the charges, or superselection sectors, of the quantum double model.
	To this end, we study the excitations of the model, which can be obtained by using 
	what are called ribbon operators \eqref{eqn:ribbondef}.
	Recall the ribbon operators we defined in Section \ref{sec:qdoub} and some of the essential properties of these operators were discussed.
	A good understanding of these operators and how they can be used to build up the local Hilbert spaces will be essential to our proof.
	In particular, we need to be able to detect such excitations with local projections, which we will call \emph{charge projections}.
	We will prove the projections that measure the total charge in a box are supported on the boundary of the box.
	
	After introducing the charge projectors, we turn to the main topic of interest, the infinite volume ground states. 
	Recall from Section \ref{sec:qdoub} there exists a the unique frustration-free ground state, $\omega_0$.
	We will also construct a family of pure non-translation invariant ground states that do not satisfy the  frustration-free property, namely, the single excitation states.
	It turns out that they can be obtained by judiciously choosing boundary terms of finite volume Hamiltonians, and taking weak$^*$ limits of finite volume ground states. As such, they can be identified with the superselection sectors of the theory, 
	see~\cite{FiedlerN,Naaijkens11}.
	Finally, we will show that in fact all pure ground states are equivalent to such states.
	
	The results of this chapter are based off the work of the author in collaboration with Bruno Nachtergaele and Pieter Naaijkens \cite{ChaNN}.
	
	\pagebreak
	
	\section{Excitations and superselection sectors}\label{sec:supersel}

	\subsection{Local and global charge projectors}\label{sec:finvolume}
	We can detect the presence of an excitation and the charge type localized at a site $s=(v,f)$ 
	with the orthogonal projectors 
	\begin{align*}
	D_v^\chi &:= \frac{1}{\abs{G}} \sum_{g\in G} \overline{\chi}(g) A_v^g  &\text{ for } \chi \in \widehat{G}, \\
	D_f^c &:= B_f^c  &\text{ for } c \in G.
	\end{align*}
	These can be obtained by considering the action of the quantum double at each site~\cite{BombinMD}.
	The first detects ``electric'' charges labeled by the characters of $G$, while the latter project on the ``magnetic'' charges labeled by group elements.
	The electric charges are located on the vertices, while the magnetic charges are located on faces.
	Since we only consider abelian models, these two types of excitations can be treated separately.
	One can check that, in the case of abelian groups $G$, the projectors commute.
	Thus, there is no ambiguity in defining the operator $D_s^{\chi, c} := D_v^\chi D_f^c$ for the site $s =(v,f)$.
	
	The local charge projectors have the following properties. They follow readily using the properties listed in the previous chapter.
	\begin{align}
	D_v^\chi \Omega&= \delta_{\chi,\iota} \Omega, & D_f^c \Omega&=  \delta_{c, e} \Omega &\text{ for all  } \Omega\in \calG_L \label{eqn:localprojgs}\\
	D_v^{\chi} F_\rho^{\xi, d} & = F_\rho^{\xi, d} D_v^{\chi\bar{\xi}},  &
	D_f^{c} F_\rho^{\xi, d} &= F_\rho^{\xi, d} D_f^{c\bar{d}}
	&\text{ if } \partial_0\rho = (v,f) \neq \partial_1\rho  \label{eqn:localprojribbonrelation1}\\
	D_v^{\chi} F_\rho^{\xi, d} & = F_\rho^{\xi, d} D_v^{\chi\xi},  &
	D_f^{c} F_\rho^{\xi, d} &= F_\rho^{\xi, d} D_f^{dc}
	&\text{ if } \partial_1\rho = (v,f) \neq \partial_0\rho\label{eqn:localprojribbonrelation2} \\
	D_v^{\chi}D_v^{\xi} &= \delta_{\chi,\xi} D_s^{\chi}, &
	D_f^{c}D_f^{d} &= \delta_{c,d} D_f^{c}\\
	\sum_{\chi\in\widehat{G}} D_v^{\chi} &= I, &  \sum_{ c \in G} D_f^{c} & = I  &
	\implies \sum_{ \chi\in \widehat{G}, c \in G} D_s^{\chi,c} = I \label{eqn:localprojcomplete}.
	\end{align}
	Thus the projections onto the charges at a site $s$ form a complete set of orthogonal projections, by equation~\eqref{eqn:localprojcomplete}.
	The ground state projector onto the ground state subspace $\calG_L$ is a product of all local charge projectors with trivial charge, $\prod_{s\in \mathcal{S}_L} D_s^{\iota,e}$.
	
	A \emph{global} (or \emph{total}) \emph{charge projector} selects the total charge $(\chi,c)$ in the region $\Lambda_L$.
	Heuristically, if for each face $f\in \mathcal{F}_L$ there is a local charge $c_f$, 
	then the total charge of magnetic type is $\prod_{f\in \mathcal{F}_L} c_f = c$ (again, we use that $G$ is abelian).
	For example, if magnetic charges $c$ and $\bar{c}$ are located on two faces and all other faces carry trivial charge, then the total charge in the region is $ c \bar{c} = e$.
	Thus, for a charge $c$ the \emph{conjugate charge} is given by the inverse group element $\bar{c}$.
	Here ``conjugate charge'' is standard terminology: it means that you can combine the two charges to obtain a trivial charge. It has nothing to do with conjugation in the group.
	
	Note that the example above in particular shows that trivial global charge does not mean that there are no excitations in the region $\Lambda_L$.
	Rather, it means that all charges add up to the trivial charge (similarly, the total charge of a state with an electron and a positron is trivial).
	An open ribbon operator with both endpoints in $\mathcal{S}_L$  
	generates a charge $c$ (see equation~\eqref{eqn:localprojribbonrelation1}), and a conjugate charge $\bar{c}$ (from equation~\eqref{eqn:localprojribbonrelation2}) pair at each of its endpoints. 
	Thus, charge is created locally at its endpoints but the total charge of the initial state is preserved.
	However, an open ribbon operator with only one of its endpoints in the region $\Lambda_L$ 
	does \emph{not} conserve the global charge in the region.
	The same is true for charges of electric type, where the multiplication is in the dual group $\widehat{G}$.
	
	We define the global charge projectors by
	\begin{equation*}
	D_L^\chi:= \sum_{\prod_{v}\chi_v = \chi} \ 
	\prod_{v\in\mathcal{V}_L} D_v^{\chi_v}, \qquad 
	D_L^{c} := \sum_{\prod_{f}c_f= c}  \ 
	\prod_{f \in \mathcal{F}_L} D_f^{c_f},
	\end{equation*}
	where the sums are over all configurations $\{ \chi_v \}_{v \in \mathcal{V}_L}$ such that $ \prod_{v}\chi_v = \chi$
	and configurations $\{ c_f\}_{f \in \mathcal{F}_L}$ such that $\prod_{f}c_f= c$, respectively.

	To project onto (non-trivial) electric $\epsilon$ or magnetic $\mu$ charge types, we have the projectors, respectively,
	\begin{align*}
	D_L^\epsilon  &:= \sum_{\chi\neq \iota}D_L^\chi = I - D_L^\iota, &
	D_L^\mu &:= \sum_{c\neq e} D_L^c = I - D_L^e.
	\end{align*}
	From the definitions above it appears that the global charge projectors are supported on the entire region $\Lambda_L$. 
	We will show that they are actually boundary operators.
	
	To this end, we consider closed ribbon operators encircling the boundary of $\Lambda_L$ given by the operators
	\begin{equation}\label{eqn:boundribbon}
	V_L^\epsilon := \frac{1}{\abs{G}} \sum_{c\in G} \left(I - F_{\partial L}^{\iota, c}\right), \qquad
	V_L^\mu := \frac{1}{\abs{G}} \sum_{\chi \in \widehat{G}} \left( I -  F_{\partial L}^{\chi,e} \right),
	\end{equation}
	where we use $\partial L$ to denote the closed ribbon running along the boundary of $\Lambda_L$, see Figure~\ref{fig:ribbons}.
	Due to the anyonic nature of the charges, if there is a charge present in the region $\Lambda_L$, this is a non-trivial operation for at least one of the charges.
	In principle, this can be detected and used to determine the total charge in $\Lambda_L$.
	
	We first show that $V_L^\varepsilon$ and $V_L^\mu$ are in fact projections (see~\cite[Sect. B.9]{BombinMD}).
	\begin{prop}
		The operators $V_L^\epsilon$ and $V_L^\mu$ are orthogonal projections.
	\end{prop}
	
	\begin{proof}
		From equation \eqref{eqn:ribprop1}, we have that 
		\[\sum_{\chi,\chi'} F_{ \partial L}^{\chi, e} F_{ \partial L}^{\chi', e}= 
		\sum_{\chi,\chi'} F_{ \partial L}^{\chi \chi' , e} = \abs{G} \sum_{\chi} F_{\partial L}^{\chi,e}.\]
		Thus,
		\begin{align*}
		(V_L^\epsilon)^2 &= \frac{1}{\abs{G}^2} \sum_{\chi,\chi'} \left(I - F_{\partial L}^{\chi, e } -  F_{\partial L}^{\chi', e } 
		+ F_{\partial L}^{\chi, e } F_{\partial L}^{\chi', e } \right) \\
		&= \frac{1}{\abs{G}} \sum_{\chi} \left( I - F_{\partial L}^{\chi, e }  \right) = V_L^\epsilon.
		\end{align*}
		Also with equation~\eqref{eqn:ribprop1} we find
		\begin{align*}
		(V_L^\epsilon)^* &= \frac{1}{\abs{G}} \sum_{\chi} I - F_{\partial L}^{\chi, e *} = \frac{1}{\abs{G}} \sum_{\chi} I - F_{\partial L}^{\bar{\chi}, e} = V_L^\epsilon.
		\end{align*}
		The calculation to show that $V_L^\mu = (V_L^\mu)^2 = (V_L^\mu)^*$ is similar.
	\end{proof}
	
	The following lemma states that localized ribbon operators do not change the total charge.
	In other words, local operation cannot change the charged sector.  
	
	\begin{lemma}\label{lem:ribchargeinv}
		Let $L' > L$.  Then, $[D_{L'}^\chi,A] = [D_{L'}^c, A] =0 $ for all
		$ \chi\in \widehat{G}, c\in G$ and $A \in \calA_L$.  
	\end{lemma}
	
	\begin{proof}
		Suppose $A = F_\rho^{\xi,d}$ is a ribbon operator.
		If $\rho$ is a closed ribbon, then by equation~\eqref{eqn:ribclosedcom},
		$[D_{L'}^\chi, F_\rho^{\xi,d}] = [D_{L'}^c,  F_\rho^{\xi,d}] =0 $.
		
		If $\rho$ is an open ribbon let $\partial_i\rho = s_i = (v_i,f_i)$ for $i = 0,1$. 
		Since, $L'>L$ it follows that $v_i \in \mathcal{V}_{L'}$ and $s_i \in \mathcal{S}_{L'}$ for $i =0,1$.
		Thus,
		\begin{align*}
		D_{L'}^\chi F_\rho^{\xi,d} 
		&= \sum_{\prod_{v}\chi_v = \chi} \ 
		\prod_{v\in\mathcal{V}_{L'}} D_v^{\chi_v} F_\rho^{\xi,c}
		= F_\rho^{\chi,c}
		\sum_{\prod_{v}\chi_v = \chi} \ 
		\prod_{\substack{v\in \mathcal{V}_{L'}\\ v\neq v_0,v_1}} D_v^{\chi_v} D_{s_0}^{\chi_{v_0} \bar{\xi} } D_{s_1}^{\chi_{v_1} \xi}\\
		& = F_\rho^{\xi, d}  D_{L'}^\chi.
		\end{align*}
		The lemma follows from the fact that ribbon operators generate $\calA_L$.
		
		A similar calculation holds to show $[D_L^c, A] = 0$.
	\end{proof}
	
	We can now prove our claim that the operators $D_L^{\epsilon}$ and $D_L^\mu$ are supported on the boundary of $\Lambda_L$, by showing that they are equal to the projections in the proposition above. We do this by showing that equality holds on a spanning set of vectors.
	
	\begin{lemma}\label{lem:globprojboundaryop}
		The global charge projectors $D_L^\epsilon$ and $D_L^\mu$, defined above, 
		are supported on the boundary of the region $\Lambda_L$.  
		In particular, they are equal to the boundary ribbon operators on $\calH_L$, defined in equations~\eqref{eqn:boundribbon}:
		\begin{equation*}
		D_L^\epsilon= V_L^\epsilon \in \calA_{\Lambda_L\setminus\Lambda_{L-1}}, \qquad  D_L^\mu=V_L^\mu \in \calA_{\Lambda_L\setminus\Lambda_{L-1}}.
		\end{equation*}
	\end{lemma}

	\begin{proof}
		The first part of the lemma follows directly from Lemma \ref{lem:ribchargeinv} and an application of Schur's lemma.
		
		Let $\Omega \in \calG_L$.  
		From equation \eqref{eqn:closedribbon}, $F_{\partial L}^{\chi,c} \Omega = \Omega$ for all $\chi, c$.
		Thus, $V_L^\epsilon \Omega = V^\mu_L \Omega = 0$. 
		On the other hand, from eqn.~\eqref{eqn:localprojgs}, 
		$D_L^\epsilon \Omega= \sum_{\prod_{s}\chi_s \neq \iota} 
		\prod_{v} D_v^{\chi_v} \Omega 
		= \sum_{\prod_{v}\chi_v \neq \iota}  \prod_{v} \delta_{\chi_v, \iota} \Omega = 0$,
		where each term in the sum vanishes because the requirement $\prod_{v}\chi_v \neq \iota$ 
		forces a non-trivial charge to exist at least one star.
		Similarly, $D_L^\mu \Omega = 0$.
		Hence the operators agree on the subspace $\calG_L$.
		
		Let $\rho$ be a ribbon and suppose $\rho$ has both endpoints contained in $\Lambda_L$.
		Label its endpoints as $\partial_0\rho = (v_0, f_0)$ and $ \partial_1 \rho = (v_1, f_1)$.
		By path independence in the ground state, equation~\eqref{eqn:ribbonpathind}, we can assume without loss of generality that $\rho$ 
		does not intersect the boundary ribbon, implying $V_L^\epsilon F_\rho^{\chi,c} \Omega= F_\rho^{\chi,c}V_L^\epsilon \Omega = 0$.
		Also, by Lemma \ref{lem:ribchargeinv} $D_L^\epsilon F_\rho^{\chi,c} \Omega= F_\rho^{\chi,c}D_L^\epsilon \Omega = 0$.
		
		Now suppose $\rho$ has both endpoints contained on the boundary.
		Since the endpoints of $\rho$ are not in $\mathcal{S}_L$,
		the vertex and face projectors corresponding to the endpoints of $\rho$ will not be involved in the products defining $D_L^\epsilon$.
		It follows from \eqref{eqn:ribstarrel}
		that $D_L^\epsilon F_\rho^{\chi,c} = F_\rho^{\chi,c} D_L^\epsilon$.
		On the other hand, from the ribbon crossing relations \eqref{eqn:ribbonrelation} we can compute
		\begin{align*}
		V_L^\epsilon F_{\rho}^{\chi,c} &= \bigg(\frac{1}{\abs{G}}\sum_{g\in G} I - F_{\partial L}^{\iota,g}\bigg)F_{\rho}^{\chi,c}\\
		& = F_{\rho}^{\chi,c} \bigg(\frac{1}{\abs{G}}\sum_{\chi\neq\iota} I - \chi(g) \chi(\bar{g}) F_{\partial L}^{\iota,g}\bigg)\\
		&= F_{\rho}^{\chi,c} V_L^\epsilon.
		\end{align*}

		Finally, we consider the action of $V_L^\epsilon$ and $D_L^\epsilon$ 
		on the spanning set of vectors of the form $\prod_{i}F_{\rho_i}^{\chi_i,c_i} \Omega$.
		From the above arguments, without loss of generality, 
		we can consider the product of ribbon operators for ribbons which have one endpoint on the boundary
		and one in the interior of $\Lambda$, since the other path operators commute with both $V_L^\epsilon$ and $D_L^{\epsilon}$. 
		We also assume $\partial_0 \rho_i \in \mathcal{S}_L$.
		This can always be achieved by reversing the direction of the ribbon (called ``inversion'' in~\cite{BombinMD}).
		By concatenation of ribbons \eqref{eqn:ribbonconcatenation}, without loss of generality we can assume
		that all endpoints are distinct, that is,
		$\partial_i\rho_k \neq \partial_j\rho_l$ for all $i,j = 0,1$ and $k, l = 0,1,\ldots,n$.
		Let $\{(v_i,f_i)\}_{i=1}^n$ enumerate the endpoints of $ \rho_i$ contained in the interior of $\Lambda_L$.
		
		Applying the ribbon operator relations \eqref{eqn:ribbonrelation},
		\begin{align*}
		V_L^\epsilon \bigg( \prod_{i=1}^n F_{\rho_i}^{\chi_i,c_i}\bigg) \Omega
		&= \frac{1}{\abs{G}} \bigg( \sum_{g\in G} I - F_{\partial L}^{\iota,g}\bigg) 
		\bigg( \prod_{i=1}^n F_{\rho_i}^{\chi_i,c_i}\bigg) \Omega\\
		& = \frac{1}{\abs{G}}\bigg( \prod_{i=1}^n F_{\rho_i}^{\chi_i,c_i}\bigg)
		\bigg( \sum_{g\in G} I - \prod_{i=1}^n \chi_i(g) F_{\partial L}^{\iota,g}\bigg)\Omega\\
		& =\bigg( \prod_{i=1}^n F_{\rho_i}^{\chi_i,c_i}\bigg) \left( I - \delta_{\prod_{i=1}^n \chi_i,\iota} \right)\Omega,
		\end{align*}
		where we use orthogonality of characters, $ \frac{1}{|G|} \sum_{g\in G} \prod_{i=1}^n \chi_i(g) = \delta_{\prod_{i=1}^n \chi_i,\iota}$.
		On the other hand we have that 
		\begin{align*}
		D_L^\epsilon \bigg( \prod_{i=1}^n F_{\rho_i}^{\chi_i,c_i}\bigg) \Omega
		&= \bigg( \sum_{\prod_{v}\chi_v \neq \iota} 
		\prod_{v} D_v^{\chi_v} \bigg) \bigg( \prod_{i=1}^n F_{\rho_i}^{\chi_i,c_i}\bigg) \Omega\\
		&=\bigg( \prod_{i=1}^n F_{\rho_i}^{\chi_i,c_i}\bigg)  \bigg( \sum_{\prod_{v}\chi_v \neq \iota} 
		\prod_{i=1}^n D_{v_i}^{\chi_{i}\chi_{v_i}} \prod_{v \neq v_i} D_v^{\chi_v}\bigg) \Omega\\
		&= \bigg( \prod_{i=1}^n F_{\rho_i}^{\chi_i,c_i}\bigg) \left( I - \delta_{\prod_{i=1}^n \chi_i,\iota} \right)\Omega,
		\end{align*}
		where for the last equality we apply an extension of \eqref{eqn:localprojribbonrelation1}.
		
		We have shown $D_L^\epsilon \left( \prod_{i=1}^n F_{\rho_i}^{\chi_i,c_i} \Omega\right) =V_L^\epsilon \left( \prod_{i=1}^n F_{\rho_i}^{\chi_i,c_i} \Omega\right)$ for any arbitrary spanning vector of $\calH_L$.
		Therefore, $D_L^\epsilon  = V_L^\epsilon$ as operators on $\calH_L$.
		A similar argument gives  $D_L^\mu = V_L^\mu $ as operators on $\calH_L$.
	\end{proof}
	
	We briefly summarize the physical interpretation of the properties of ribbon operators and charge projectors.
	The ribbon operators create a pair of excitations, one at each end of the ribbon.
	In addition, the different types of excitations are labeled by pairs $(\chi,c) \in \widehat{G} \times G$ and the charges at each end are \emph{conjugate} to each other.
	In other words, the \emph{global} charge does not change after applying a ribbon operator.
	Moreover, when acting with a ribbon operator on the ground state, the resulting state only depends on the endpoints of the ribbon, i.e.\ the location of the excitations.
	A similar thing is true for configurations of multiple charges, up to a phase.
	This phase can be explained by the anyonic nature of the charges: exchanging two of them gives an overall phase, much like interchanging two fermions yield a minus sign.
	The local Hilbert spaces can be obtained completely by such operations, so that we can define a basis by specifying the charge at each site (with an additional constraint on the total charge).
	This observation will play an important role in our proof.
	
	\section{Single excitation ground states}
	Recall that for the quantum double models there exists a unique ground state satisfying the frustration-free property Proposition \ref{prop:qdffgs}.
	
	\begin{prop}(\cite{AlickiFH,Naaijkens11,FiedlerN})
		Let $\omega^0$ be the frustration-free ground state of the quantum double model obtained as the limit of finite volume ground states.
		Then,
		\begin{enumerate}
			\item if $\omega$ is a frustration-free ground state then $\omega = \omega^0$,
			\item $\omega^0$ is a pure state,
			\item Let $(\pi_0, \Omega_0, \calH_0)$ be a GNS-representation for $\omega^0$ and $H_0$ be the GNS Hamiltonian.
			Then, $spec(H_0) = \ZZ^{\geq 0}$ with a simple ground state eigenvector $\Omega_0$.
		\end{enumerate}
	\end{prop}
	
	In the finite volume, elementary excitations are constructed by violating one of the frustration-free ground state conditions.
	These excitations must come in pairs since they are generated by open ribbon operators.
	By introducing a boundary condition to the Hamiltonian,
	we generate ground states which have one excitation in the bulk and one on the boundary.
	In the infinite volume, this is equivalent to moving one of the excitations off to infinity, 
	thereby isolating a single excitation in the bulk.
	By construction, these states will be ground states in the infinite volume.
	One way to understand this intuitively is that even though there is an excitation, we cannot lower the energy of the system with local operations.
	It is possible to move the excitation around (thereby locally decreasing the energy), but we cannot get rid of it completely with local operations, and the moved excitation will increase the local energy density at its new location.
	
	The idea is to use the projections that were introduced in the previous section to define boundary conditions, which can compensate for the existence of an excitation in the bulk.
	Recall that these projections are supported on the boundary of $\Lambda_L$.
	
	\begin{defn}
		\label{def:gsboundary}
		Define the following Hamiltonians with boundary condition
		\begin{align*}
		H_L^{\epsilon} &:= H_L - V_L^\epsilon,\\
		H_L^{\mu} &:= H_L - V_L^\mu,\\
		H_L^{\epsilon, \mu} & := H_L - V_L^\epsilon - V_L^\mu. 
		\end{align*}
		We will sometimes use the index $k$ to denote either $\epsilon, \mu,$ or $ (\epsilon,\mu)$,
		and set $V^{\epsilon,\mu}_L := V_L^\epsilon + V_L^\mu$.
	\end{defn}
	
	Recall that the boundary terms $V_L^k$ are linear combination of closed ribbon operators and thus commute with each interaction term, and hence the Hamiltonian:
	\begin{equation*}
	[V_L^k, B_f] = [ V_L^k, A_v] = [ V_L^k, H_L] = 0 \quad \text{ for all } \quad k, f, v.
	\end{equation*}
	In what follows we will show that $H_L^k \geq 0$, despite it being a difference of positive operators.
	From equation \eqref{eqn:closedribbon}, if  $\Omega_L \in \calG_L$ then $H_L^k \Omega_L = 0$ for all $k$.
	
	Now consider an open ribbon $\rho$ connecting a site $ \partial_0\rho = (v,f)\in \mathcal{S}_L$ to a site on the boundary,
	for instance $\rho$ as the ribbon in Figure~\ref{fig:ribbons}, and its corresponding ribbon operator, $F_\rho^{\chi,c}\in \calA_L$.
	Then, for $(\chi,c)$ we have that
	\begin{align*}
	H_L F_\rho^{\chi,c} \Omega_L &= F_\rho^{\chi,c} \left(I- \frac{1}{\abs{G}} \sum_{g\in G} \chi(g) A_v^g + I - B_f^{\bar{c}}\right) \Omega_L\\
	& = F_\rho^{\chi,c} \left( 2  -  \delta_{\bar{\chi},\iota}- \delta_{\bar{c},e}\right) \Omega_L.
	\end{align*} 
	In the last line we used orthonormality: $\langle \chi_1, \chi_2 \rangle := \frac{1}{|G|} \sum_{g \in G} \bar{\chi}_1(g) \chi_2(g) = \delta_{\chi_1, \chi_2}$. Similar calculations, as in the proof of Lemma~\ref{lem:globprojboundaryop}, yield:
	\begin{align*}
	V_L^\epsilon F_\rho^{\chi,c} \Omega_L &= F^{\chi,c}_\rho \left[\frac{1}{\abs{G}} \sum_{d\in G} \left(I - \chi(d) F^{\iota, d}_{\partial L}\right)\right] \Omega_L\\
	& = F^{\chi,c}_\rho ( I - \delta_{\chi, \iota}) \Omega_L,
	\end{align*}
	and for the magnetic charges,

	\begin{align*}
	V_L^\mu F_\rho^{\chi,c} \Omega_L &= F^{\chi,c}_\rho \left[\frac{1}{\abs{G}} \sum_{\xi\in \widehat{G}}\left( I - \xi(c) F^{\xi,e}_{\partial L}\right)\right] \Omega_L\\
	& = F^{\chi,c}_\rho ( I - \delta_{c, e}) \Omega_L.
	\end{align*}
	
	Therefore, together with equation~\eqref{eqn:ribenergy}, we find
	\begin{equation*}
	H_L^{\epsilon, \mu} F_\rho^{\chi, c} \Omega_L = 0.
	\end{equation*}
	If $\rho$ connects  two sites in $\mathcal{S}_L$ then $[V_L^k, F_\rho^{\chi,c}] = 0$.
	Thus, combined with a similar calculation from above for multiple charges on the boundary, we can conclude that 
	\[ \bigg\langle \bigg( \prod_{i} F_{\rho_i}^{\chi_i, c_i} \bigg) \Omega_L,  H_L^{\epsilon, \mu} \bigg( \prod_{i} F_{\rho_i}^{\chi_i, c_i} \bigg) \Omega_L \bigg\rangle \geq 0.\]
	Therefore, we have the following lemma.
	\begin{lemma}\label{lem:gsspan}
		Let $\calG_L^k$ be the ground state space of $H_L^k$ for $k = \epsilon, \mu$ and $(\epsilon,\mu)$.  Then,
		\begin{enumerate}
			\item $H_L^{\epsilon,\mu} \geq 0$ and $H_L^k \calG_L^k = 0$,
			\item $\calG_L^{\epsilon}$ is spanned by $\{ F_\rho^{\chi,e} \Omega_L: \rho $ connects an interior site to the boundary, $\Omega_L \in \calG_L, \chi\in \widehat{G}\}$.
			
			$\calG_L^{\mu}$ is spanned by $\{ F_\sigma^{\iota,c} \Omega_L: \sigma $ connects an interior site to the boundary, $\Omega_L \in \calG_L,  c\in G \}$.
			
			$\calG_L^{\epsilon,\mu}$ is spanned by $\{ F_\rho^{\chi,e}F_\sigma^{\iota,c}  \Omega_L: \rho,\sigma $ connects interior sites to the boundary, $\Omega_L \in \calG_L, (\chi,c)\in\widehat{G}\times G\}$.
		\end{enumerate}
	\end{lemma}
	
	From the decompositions given in Lemma~\ref{lem:gsspan}, 
	it is clear that 
	\begin{equation*}
	\text{if } L' > L \text{ then } H_L^{ \epsilon, \mu} ( \calG_{L'}^{\epsilon,\mu}) = 0.
	\end{equation*}
	Note that $\calG^\mu_L$ and $\calG_L^\epsilon$ are subspaces of $\calG_L^{\epsilon,\mu}$. 
	This result allows us to decompose the ground state space into different sectors corresponding to the different charges:
	\begin{cor}
		The ground state space has a natural decomposition
		\begin{equation*}
		\calG_L^{\epsilon,\mu} = \bigoplus_{\chi\in\widehat{G},c\in G} D_L^{\chi,c}\calG_L^{\epsilon,\mu}.
		\end{equation*}
	\end{cor}
	
	\begin{proof}
		This follows from Lemma \ref{lem:gsspan} and the relation \eqref{eqn:localprojcomplete}.
	\end{proof}

	Note that by construction we have that $H_L^{\epsilon,\mu} \geq 0$,
	$ \delta(A) = \lim_{L\ra \infty} [H_L, A] = \lim_{L\ra \infty} [H_L^{\epsilon,\mu}, A]$  for all $ A \in \calA_{loc}$, 
	and  $\omega_s^{\chi,c}( H_L^{\epsilon,\mu}) = 0$ from Lemma~\ref{lem:gsspan}.
	From the following basic lemma, it follows that the single excitation state $\omega_s^{\chi,c}$ (see Equation \eqref{eqn:singleexcitation}) is an infinite volume ground state. 
	
	\begin{lemma}\label{lem:finitegslim}
		Let $\omega \in \calA_{+,1}^*$ and $\tilde{H}_L\in \calA_L$ be a sequence of positive operators 
		such that $\delta(A) = \lim_{L \ra \infty} [\tilde{H}_L, A]$ for all $A \in \calA_{loc}$.
		If $\omega (\tilde{H}_L) = 0$ for all $L$ then $\omega$ is a ground state, that is, 
		$\omega(A^* \delta(A)) \geq 0$ for all $A \in \calA_{loc}$.
	\end{lemma}
	
	\begin{proof}
		Let $\rho_L \in \calA_L$ be the reduced density matrix for $\omega$ on $\calA_L$, that is, 
		$\omega(A) = \Tr(\rho_L A)$ for all $A \in \calA_L$. 
		From the condition, $ \omega(\tilde{H}_L) = 0$ and $\tilde{H}_L \geq 0$, it follows that $\tilde{H}_L \rho_L = 0$ for all $L$.
		Therefore, by boundedness of $\omega$,
		\begin{align*}
		\omega(A^* \delta(A)) & = \lim_{L \ra \infty} \omega( A^* [ \tilde{H}_L, A])= \lim_{L \ra \infty}  \omega(A^* \tilde{H}_L A ) - \omega(A^* A \tilde{H}_L) \\
		& = \lim_{L \ra \infty} \Tr(\rho_L A^* \tilde{H}_L A)  - \Tr( \tilde{H}_L \rho_L A^*A)\\
		& \geq 0
		\end{align*}
		for all $A \in \calA_{loc}$.
	\end{proof}
	
	The states $\omega^{\chi,c}_s$ were first introduced in~\cite{FiedlerN,Naaijkens11}.
	They showed the states can be constructed from the frustration-free ground state via an automorphism,
	$ \omega_s^{\chi,c} = \omega^0 \circ \alpha^{\chi,c}_\rho$, where
	\begin{equation}\label{eqn:chargemorp}
	\alpha^{\chi,c}_\rho(A) = \lim_{L\ra \infty} F_{\rho_L}^{\chi,c *} A \ F_{\rho_L}^{\chi,c}.
	\end{equation}
	The limit converges in norm for each $A \in \calA$ and defines an outer automorphism.
	
	\section{The complete set of ground states}\label{sec:results}
	In this section we prove the main result of this chapter, that is, a complete classification of the ground states of the quantum double model
	for abelian groups.
	Our strategy is to find a boundary condition such that any 
	infinite volume ground state has zero energy for the local Hamiltonian with this boundary condition.
	It turns out that this is possible with the boundary conditions introduced in Section~\ref{sec:finvolume}.
	The classification of infinite volume ground states then
	simplifies to a classification of infinite volume limits of finite volume ground states.
	These finite volume ground states are well understood by the results in the previous section, and this allows us to obtain our classification.
	This strategy is similar to the solution of the complete ground state for the XXZ chain given in \cite{KomaN}.
	
	We begin by introducing some notation.
	\begin{defn}
		Let $K := \{ \omega \in \calA_{+,1}^*: \omega(A^* \delta(A))\geq 0\}$ denote the set of infinite volume ground states,
		where $\delta$ is the generator of the dynamics for the quantum double model for abelian group $G$.
		Similarly, for the set of all finite volume ground state functionals of $H_L^{\epsilon,\mu}$ we write $ K_L := \{\omega_L:\calA_L \ra \CC \mid \omega_L(H_L^{\epsilon,\mu})=0\} $.
	\end{defn}
	
	The first step is to show that any infinite volume ground state minimizes the energy of the finite volume Hamiltonians $H_L^{\epsilon,\mu}$ of Definition~\ref{def:gsboundary}.
	Here we will use the formulation of the boundary term in terms of a sum of products of local charge projections. This gives us precise control on the location of possible excitations.
	\begin{lemma}\label{lem:gshambound}
		Let $\omega\in \calA_{+,1}^*$.
		Then, $\omega \in K$ if and only if for all $L \geq 2$
		\begin{equation*} 
		\omega(H_L - D_L^\epsilon - D_L^\mu) = 0.
		\end{equation*} 
	\end{lemma}
	
	\begin{proof}
		($\impliedby$) This follows from Lemma \ref{lem:finitegslim}.
		
		($\implies$) We will show that $\omega( \sum_{v\in\mathcal{V}_L} \left(I - A_v\right) ) = \omega(V_L^\epsilon)$ and $\omega(\sum_{f\in\mathcal{F}_L} \left(I - B_f\right))  = \omega( V_L^\mu)$.
		The result then follows from Lemma~\ref{lem:globprojboundaryop}.
		
		Let $L\geq 2$ be given.
		Consider an arbitrary enumeration of the set of plaquettes, $\mathcal{F}_L = \{f_i\}_{i=1}^{n_L}$, and
		a configuration of magnetic charges, $\{ c_i\in G\}_{i=1}^{n_L}$
		such that $\prod_i c_i = e$.
		In the following, sums and products indexed by $i,j$ and $k$ will run from $1$ to $n_L$ unless otherwise stated.
		Pair $f_i$ with a neighboring vertex $v_i$ and 
		let $\rho_i $ be a ribbon 
		such that $\partial_0 \rho_i = (v_i, f_i) $  and $\partial_1 \rho_i = (v_{i+1},f_{i+1})$.
		With this choice, consider the operator $A =\big( \prod_i F_{\rho_i}^{\iota, \tilde{c}_i}\big)\big(\prod_i(I - B_{f_i})\big)$, 
		where the family $\{ \tilde{c}_i \}$ is chosen such that 
		\begin{equation*}
		B_{f_i} \bigg(\prod_k F_{\rho_k}^{\iota, \tilde{c}_k}\bigg) = \bigg(\prod_k F_{\rho_k}^{\iota, \tilde{c}_k}\bigg) B_{f_i}^{c_i}, \quad \forall i.
		\end{equation*}
		Indeed, the $\tilde{c}_i$'s must be such that $c_{i+1} = \tilde{c}_{i} \bar{\tilde{c}}_{i+1}$.
		The condition that $\prod_i c_i = e$ guarantees that such a family $\{ \widetilde{c}_i \}_i$ always exists,
		for instance, $\tilde{c}_i = \prod_{j\leq i}\bar{c}_j.$

		We want to apply the ground state condition to the operator $A$, hence we compute
		\begin{align*}
		A^*\delta(A) &= A^*[H_L, A] \\
		&=  \prod_{i} (I - B_{f_i})\bigg(\prod_i F_{\rho_i}^{\iota, \tilde{c}_i}\bigg)^* \Big(\sum_{j} \bigg[  - B_{f_j},\prod_i F_{\rho_i}^{\iota, \tilde{c}_i} \bigg]\Big) \prod_i (I - B_{f_i})\\
		& = \prod_{i} (I - B_{f_i})\bigg(\prod_i F_{\rho_i}^{\iota, \tilde{c}_i}\bigg)^*\prod_i F_{\rho_i}^{\iota, \tilde{c}_i} \bigg(\sum_{j} B_{f_j} - B_{f_j}^{c_{j}}\bigg) \prod_{i} (I - B_{f_i})\\
		&= \prod_{i} (I - B_{f_i}) \bigg(\sum_{j} B_{f_j} - B_{f_j}^{c_{j}}\bigg) \\
		&=  - \prod_{i} (I - B_{f_i}) \bigg(\sum_{j} B_{f_j}^{c_{j}}\bigg).
		\end{align*}
		The operator $\prod_{i} (I - B_{f_i}) \bigg(\sum_{j} B_{f_j}^{c_{j}}\bigg)$ is a product of
		commuting positive operators and, hence, it is positive. But this implies that $A^*[H_L, A] \leq 0$.
		
		Because of the ground state condition, equation \eqref{eqn:gs}, and the calculation above, $\omega(A^*[H_L, A]) = 0$.
		We can then sum over each configuration $c_i$ with trivial product.
		Note that if we fix $c_j$ for $j=1, \ldots, n_L-1$, this fixes $c_{n_L}$ by the condition that their product should be trivial.
		Hence the summation over all configurations gives
		\begin{equation*}
		0 = \sum_{(c_1, \dots, c_{n_L-1}) \in G^{n_{L}-1}} \omega\bigg(\prod_{i} (I - B_{f_i}) \bigg(B_{f_{n_L}}^{\overline{\prod c_i}} + \sum_{k=1}^{n_L-1} B_{f_k}^{c_{k}}\bigg)\bigg).
		\end{equation*}
		Here we separated the $n_{L}$ face from the others in the summation, since its magnetic charge is fixed by the others.
		We now do the summation over $c_1$.
		Note that as $c_1$ runs over the group $G$, so does $\overline{\prod_{i=1}^{n_L-1} c_i}$.
		Also note that for any $j$, $\sum_{c_j\in G} B_{f_j}^{c_j} = I$.
		This yields, by repeating this procedure,
		\begin{align*}
		0 & = \sum_{(c_2, \dots c_{n_L-1}) \in G^{n_{L}-2}} \omega\bigg(\prod_{i} (I - B_{f_i}) \bigg(2 I + \sum_{k=2}^{n_L-1} B_{f_k}^{c_{k}}\bigg)\bigg)\\
		& = c(G, n_L) \omega\bigg(\prod_{i} (I - B_{f_i})\bigg),
		\end{align*}
		where $c(G, n_L)$ is some non-zero constant depending only on $|G|$ and the number of plaquettes.
		Therefore, 
		\begin{equation}\label{eqn:doubgscond3}
		\omega\bigg(\prod_{i=1}^{n_L} (I - B_{f_i})\bigg) = 0.
		\end{equation}
		Equation \eqref{eqn:doubgscond3} generally holds for a finite subset $\Lambda\subset \calB$,
		where we assume that the subset is contained in some box $\Lambda_L$.
		We will need this fact for the following argument.
		
		We proceed by induction to show that 
		\begin{equation}\label{eqn:doubgscond1}
		\omega\bigg( \sum_{f\in \mathcal{F}_L} I-B_f\bigg)  = \omega\bigg( I - \prod_{f\in \mathcal{F}_L} B_f\bigg).
		\end{equation}
		For the case of two faces, $f_1$ and $f_2$, we have from equation~\eqref{eqn:doubgscond3}
		\begin{align*}
		0  = \omega( (I- B_{f_1})(I-B_{f_2}))= \omega( I - B_{f_1} - B_{f_2} + B_{f_1}B_{f_2}),
		\end{align*}
		so that $\omega( I - B_{f_1}) + \omega(I- B_{f_2}) = \omega(I- B_{f_1}B_{f_2})$.
		
		Suppose that equation~\eqref{eqn:doubgscond1} holds if $X$ is a finite collection of faces with $\abs{X} \leq n$.
		Now let $X$ be a finite collection of $n$ faces making up a region in $\Lambda_L$ and 
		enumerate the elements, $X = \{f_i\}_{i=1}^n$ and let $ f_{n+1} \notin X$ be a face in $X$ but otherwise arbitrary.
		From equation~\eqref{eqn:doubgscond3} it follows that 
		\[
		\omega\Big( \Big(\prod_{i\in X} I-B_{f_i}\Big)(I-B_{f_{n+1}} )\Big) = 0.
		\]
		Expanding the product and using the hypothesis we have,
		
		\begin{align*}
		0 & =  \omega\Big( \Big(\prod_{f_i\in X} I-B_{f_i}\Big)(I-B_{f_{n+1}} )\Big) \\
		& = \omega\Big( I - \sum_{i=1}^{n+1} B_{f_i} + \sum_{i<j \leq n+1} B_{f_i}B_{f_j} - \sum_{i<j<k \leq n+1} B_{f_i}B_{f_j}B_{f_k} +   \ldots + 
		(-1)^{n+1} \Big( \prod_{f_i\in X} B_{f_i} \Big) B_{f_{n+1}}\Big)\\
		&= 1 + \left[ \sum_{i=1}^{n+1} \omega\left( I - B_{f_i} \right) - \binom{n+1}{1} \right] - 
		\left[ \sum_{i < j \leq n+1} \omega\left(I - B_{f_i} B_{f_j}\right) - \binom{n+1}{2}\right] \\
		& \qquad \qquad \qquad +  \ldots + (-1)^{n+1} \omega\Big( \Big( \prod_{f_i\in X} B_{f_i} \Big) B_{f_{n+1}}\Big)\Big) \\
		&= -(-1)^{n+1} + \sum_{i=1}^{n+1} \omega\left( I - B_{f_i} \right)  + \sum_{i < j \leq n+1} \omega\left(I - B_{f_i} B_{f_j}\right) 
		+ \ldots  + (-1)^{n+1} \omega\Big( \Big( \prod_{f_i\in X} B_{f_i} \Big) B_{f_{n+1}}\Big)\Big),
		\end{align*}
		where in the last step we use the elementary equation $\sum_{k=1}^{n-1} (-1)^k \binom{n}{k} = - (1+(-1)^n)$.
		We can then apply the induction hypothesis to all but the last terms.
		Note that for the term with $k$ products of $B_{f_i}$, after applying the summation in the induction hypothesis,
		each term $\omega(I-B_{f_i})$ appears exactly $\binom{n}{k-1}$ times.
		Hence we obtain
		\begin{align*}
		0 & =  -(-1)^{n+1}  + \sum_{k=1}^n (-1)^{k+1} \binom{n}{k-1} \left( \sum_{i=1}^{n+1} \omega\left(I - B_{f_i}\right) \right) 
		+ (-1)^{n+1} \omega\Big( \Big( \prod_{f_i\in X} B_{f_i} \Big) B_{f_{n+1}}\Big)\Big) \\
		& = -(-1)^{n+1}\omega\left(I - \prod_{i=1}^{n+1} B_{f_i} \right) + (-1)^{n+1} \omega\left( \sum_{i=1}^{n+1} \left(I - B_{f_i}\right)\right),
		\end{align*}
		where we used that $\sum_{k=1}^{n} (-1)^{k+1} \binom{n}{k-1} = (-1)^{n+1}$.
		Therefore equation~\eqref{eqn:doubgscond1} holds. 
		
		Now consider a configuration of magnetic charges, $\{ c_i\in G\}_{i=1}^{n_L}$
		such that $\prod_i c_i = e$ and $\rho_i$ and $\widetilde{c}_i$ are as defined earlier.
		Let  $ A' = \prod_{i=1}^{n_L} F_{\rho_i}^{\iota, \widetilde{c}_i} \prod_{i=1}^{n_L} B_{f_i}^{c_i}$
		and let $l = \#\{i: c_i \neq e\}$.
		We compute
		
		\begin{align*}
		A'^*\delta(A')& = A'^*[H_L, A'] \\
		& = \prod_i B_{f_i}^{c_i}\bigg( \prod_i F_{\rho_i}^{\iota, \widetilde{c}_i} \bigg)^*\sum_{j}  \bigg[ - B_{f_j},  \prod_i F_{\rho_i}^{\iota, \widetilde{c}_i}\bigg]\prod_i B_{f_i}^{c_i}\\
		&=\prod_i B_{f_i}^{c_i} \sum_{j} \left(B_{f_j} - B_{f_j}^{c_{j}}\right)\\
		& = - l \prod_i B_{f_i}^{c_i}\leq 0.
		\end{align*}
		Therefore, applying the ground state condition gives
		\begin{equation}\label{eqn:doubgscond2}
		\text{ if } l > 0 \qquad \text{ then } \qquad  \omega\bigg(\prod_{i=1}^{n_L}B_{f_i}^{c_i}\bigg) = 0.
		\end{equation}
		Finally, applying the equivalence in Lemma~\ref{lem:globprojboundaryop}
		with equations~\eqref{eqn:doubgscond1} and~\eqref{eqn:doubgscond2} gives the result,
		\begin{align*}
		\omega\bigg( \sum_{i=1}^{n_L} I-B_{f_i}\bigg) & = \omega\bigg( I - \prod_{i=1}^{n_L} B_{f_i}\bigg)
		= \omega\bigg( I - \sum_{\prod_i c_i =e }  \prod_{i=1}^{n_L} B_{f_i}^{c_i} \bigg)
		= \omega(D_L^\mu) = \omega(V_L^\mu).
		\end{align*}
		
		A similar argument gives
		\begin{equation*}
		\omega \bigg( \sum_{v\in \mathcal{V}_L} I  - A_v \bigg) = \omega( V_L^\epsilon).
		\end{equation*}
		This concludes the proof.
	\end{proof}
	
	We now give precise statements and proofs for the main result of the chapter,
	starting with the definitions of the infinite volume ground state subsets:
	\begin{defn}\label{def:chargedgs}
		Define the following convex subset of states for each $(\chi,c) \in \widehat{G} \times G$: 
		\begin{equation*}
		\begin{split}
		K^{\chi,c} := \bigg\{ \omega^{\chi,c} \in \calA_{+,1}^*:& \exists \omega\in K \text{ such  that } 
		\lim_{L\ra\infty} \omega(D_{L}^{\chi,c}) > 0 \text{ exists, and } \\
		& \omega^{\chi,c} = \wslim_{L\ra\infty} \frac{\omega(\ \cdot \ D_{L}^{\chi,c} )}{\omega(D_{L}^{\chi,c})} \bigg\}.
		\end{split}
		\end{equation*}
	\end{defn}
	
	By Lemma \ref{lem:ribchargeinv}, $D_{L'}^{\chi,c}$ is a supported on the boundary.
	It follows that if $L'>L$ and  $\omega(D_{L'}^{\chi,c}) > 0$ then 
	\begin{equation*}
	\frac{ \omega( \ \cdot \ D_{L'}^{\chi,c})}{\omega(D_{L'}^{\chi,c})}\bigg|_{\calA_L} = 
	\frac{\omega( D_{L'}^{\chi,c} \cdot \ D_{L'}^{\chi,c})}{\omega(D_{L'}^{\chi,c})}\bigg|_{\calA_L}. 
	\end{equation*} 
	
	In particular, we have that $\omega(\ \cdot \ D_{L'}^{\chi,c})$ is a positive linear functional and $\omega( H_L^{\epsilon,\mu} D_{L'}^{\chi,c}) = 0$.
	Thus, by Lemma \ref{lem:gshambound}, $K^{\chi,c} \subset K$ is a subset of the set of infinite volume ground states.
	The interpretation of a state in $K^{\chi,c}$ is that it has a \emph{global} excitation of type $(\chi,c)$, hence the projection onto the charge $(\chi,c)$ in the region $\Lambda_L$ has a positive expectation value as $L$ goes to infinity.
	The assumption that $  \lim_{L\ra\infty} \omega(D_{L}^{\chi,c})$ exists is always satisfied, as follows from the next lemma.
	
	\begin{lemma}\label{lem:asymptoticcoef}
		The limit $ \lambda_{\chi,c}(\omega) := \lim_{L\ra\infty} \omega(D_{L}^{\chi,c})$ exists for all ground states $\omega$ and 
		we have $ \lambda_{\chi,c}(\omega) \geq 0$.
		Furthermore, if $\omega^{\chi,c}\in K^{\chi,c}$ then  $\lambda_{\sigma,d}(\omega^{\chi,c}) = \delta_{(\sigma,d),(\chi,c)}$.
	\end{lemma}
	
	\begin{proof}
		Let $L''> L' > L$.
		First, we claim that 
		\begin{align}
		\label{eqn:chargecont1} D_L^\chi(\calG_{L''}^{\epsilon,\mu}) &\subset D_{L'}^\chi(\calG_{L''}^{\epsilon,\mu}) & & \text{if } \chi \neq \iota\\
		\label{eqn:chargecont2} D_L^c(\calG_{L''}^{\epsilon,\mu}) &\subset D_{L'}^c(\calG_{L''}^{\epsilon,\mu}) & & \text{if } c \neq e\\ 
		\label{eqn:chargecont3} D_{L'}^\iota(\calG_{L''}^{\epsilon,\mu}) &\subset D_{L}^\iota(\calG_{L''}^{\epsilon,\mu}) & &\text{if } \chi = \iota\\ 
		\label{eqn:chargecont4} D_{L'}^e(\calG_{L''}^{\epsilon,\mu}) &\subset D_{L}^e(\calG_{L''}^{\epsilon,\mu}) & &\text{if } c = e 
		\end{align}
		(see Lemma~\ref{lem:gsspan} for a description of $\calG_{L''}^{\epsilon,\mu}$).
		Note the reversal of $L$ and $L'$ in the last two equations.
		The reason is that while in the first two equations, the \emph{presence} of a charge in the region is measured, in the last two equations it is the \emph{absence} of any charge that is important.
		
		To see why these equations are true, consider first $\chi \neq \iota$ and note that the subspace $D_{L}^\chi(\calG_{L''}^{\epsilon,\mu})$ is spanned by 
		$\{ F_{\rho}^{\chi,e}F_{\sigma}^{\iota,c}  \Omega:\forall \Omega\in \calG_{L''}; c\in G; \rho, \sigma$ path connecting sites from the interior of $ \Lambda_{L''}$ to that boundary such that $\partial_0\rho \subset \Lambda_{L}  \}$.
		The same statement is true if we replace $L$ by $L'$.
		Thus, $D_L^\chi$ selects for a $\chi$-excitation in the region $\Lambda_L$ 
		whereas $D_{L'}^\chi$ selects for a $\chi$-excitation in the region $\Lambda_{L'}$.
		The latter condition is weaker.
		Therefore, $D_{L}^\chi|_{\calG_{L''}^{\epsilon,\mu}} \leq D_{L'}^\chi|_{\calG_{L''}^{\epsilon,\mu}} $ as projections.
		A similar argument gives $D_{L}^c|_{\calG_{L''}^{\epsilon,\mu}} \leq D_{L'}^c|_{\calG_{L''}^{\epsilon,\mu}} $ as projections.
		
		If $\chi=\iota$, $D_{L'}^{\iota}$ selects for a trivial $\epsilon$-type charge (i.e., the absence of an electric charge) in the region $\Lambda_{L'}$ 
		while $D_{L}^{\iota}$ selects for a trivial $\epsilon$-type charge in the region $\Lambda_{L}$.
		The latter condition is weaker.  
		Therefore, $D_{L'}^\iota|_{\calG_{L''}^{\epsilon,\mu}} \leq D_{L}^\iota|_{\calG_{L''}^{\epsilon,\mu}} $ as projections.
		A similar argument gives $D_{L'}^e|_{\calG_{L''}^{\epsilon,\mu}} \leq D_{L}^e|_{\calG_{L''}^{\epsilon,\mu}} $ as projections. 
		This shows that~\eqref{eqn:chargecont1}--\eqref{eqn:chargecont4} hold.
		
		Let $\omega \in K$ be an infinite volume ground state.
		As remarked below Definition~\ref{def:chargedgs}, $\omega|_{\calA_{L''}} \in K_{L''}$ is a ground state functional for $H_{L''}^{\epsilon,\mu}$.
		Consider the sequence $\{ \omega(D_L^{\chi}) \}_{L=2}^\infty$.
		If $\chi \neq \iota$, the inclusion~\eqref{eqn:chargecont1} gives that $ \omega( D_{L'}^\chi - D_{L}^\chi) \geq 0$,
		thus the sequence is increasing.
		The sequence is also bounded, $\omega(D_L^\chi) \leq \| D_L^\chi \| = 1$.
		Hence we have a uniformly bounded and increasing sequence, and therefore the limit $ \lim_{ L\ra \infty} \omega(D_L^\chi)$ exists.
		A similar argument gives that the limit $\lim_{L\ra \infty} \omega(D^{\chi,c}_L)$ exists if $\chi\neq \iota$ and $c \neq e$.
		
		If $\chi \neq \iota$ and $c =e$, where there is a non-trivial electric charge and the magnetic charge is trivial, we can rewrite the projector $D_L^{\chi,e}$ as 
		\begin{equation*}
		D_L^{\chi, e} = D_L^{\chi} D_L^e =  D_L^{\chi} \bigg( I - \sum_{c \neq e} D_L^c \bigg)  = D_L^{\chi} - \sum_{c \neq e} D_L^{\chi,c}.
		\end{equation*}
		This is enough to show the limit $\lim_{L\ra \infty} \omega(D^{\chi,e}_L)$ exists.
		By similar arguments, 
		the limits exist for the cases $\chi = \iota$ with $c \neq e$, and when $\chi = \iota$ with $ c = e$.
		The limits are always positive, since $\omega(D_L^{\chi,c}) \geq 0$ for all $L$.
		
		To prove the second claim,
		let $\omega^{\chi,c}\in K^{\chi,c}$ and choose $\omega\in K$ such that
		$$
		\omega^{\chi,c} = \wslim_{L'\ra\infty} \frac{\omega(\ \cdot \ D_{L'}^{\chi,c} )}{\omega(D_{L'}^{\chi,c})}.
		$$
		We use freely that the charge projectors commute.
		Equations \eqref{eqn:chargecont1}-\eqref{eqn:chargecont2} imply that if $\chi \neq \iota$ and $c \neq e$ then $\omega(D_L^{\chi,c}D_{L'}^{\chi,c}) = \omega(D_L^{\chi,c})$.
		Equations \eqref{eqn:chargecont3}-\eqref{eqn:chargecont4} imply that $\omega(D_L^{\iota,e} D_{L'}^{\iota,e}) = \omega(D_{L'}^{\iota,e})$.
		If $\chi \neq \iota$, \eqref{eqn:chargecont1} and \eqref{eqn:chargecont3} imply that $\omega(D_L^{\chi,e} D_{L'}^{\chi,e})  = \omega(D_L^\chi D_{L'}^e)$. 
		From the Cauchy-Schwarz inequality, it follows that
		\begin{align*}
		\abs{ \omega(D_L^{\chi,e} D_{L'}^{\chi,e}) - \omega(D_L^{\chi,e} )} 
		& = \abs{ \omega( D_L^\chi (D_{L'}^e - D_L^e)}\\
		& \leq \sqrt{\omega(D_L^\chi )}\sqrt{\abs{ \omega(( D_{L'}^e - D_L^e)^2) }}\\
		&= \sqrt{\omega(D_L^\chi )}\sqrt{\abs{ \omega( D_{L'}^e + D_L^e - 2 D_L^e D_{L'}^e)}}\\
		&= \sqrt{\omega(D_L^\chi )}\sqrt{\abs{ \omega( D_{L'}^e - D_L^e) }}
		\end{align*}
		and similarly for $c \neq e$, $\abs{ \omega(D_L^{\iota,c} D_{L'}^{\iota,c}) - \omega(D_L^{\iota,c} )} \leq \sqrt{\omega(D_L^c)}\sqrt{\abs{\omega(D_{L'}^\iota - D_L^\iota)}}$.
		
		Let $\epsilon >0$ be given.  
		The previous arguments show that for all $(\chi,c) \in \widehat{G}\times G$, there exists $l$ such that if  $L'>L>l$ then
		\begin{equation}\label{eqn:fullcharge}
		\abs{ \omega(D_L^{\chi,c}D_{L'}^{\chi,c}) - \omega(D_{L}^{\chi,c})} < \epsilon.
		\end{equation}
		Thus,
		$
		\lambda^{\chi,c}(\omega^{\chi,c})
		= \lim_{L\ra\infty} \lim_{L'\ra\infty}  \frac{\omega(D_L^{\chi,c} D_{L'}^{\chi,c})}{\omega(D_{L'}^{\chi,c})}
		=\lim_{L\ra\infty} \lim_{L'\ra\infty} \frac{\omega( D_{L}^{\chi,c})}{\omega(D_{L'}^{\chi,c})}  =  1.
		$
		Therefore, combining $\sum_{\sigma,d} \lambda_{\sigma,d}(\omega^{\chi,c}) = 1$ 
		and $\lambda_{\chi,c}(\omega^{\chi,c}) =1$ 
		gives $\lambda_{\sigma,d}(\omega^{\chi,c}) = \delta_{(\sigma,d),(\chi,c)}$.
	\end{proof}
	
	From the arguments given in the previous lemma, we can achieve a slightly stronger bound which will be used later.
	Let $\epsilon>0$ be given. 
	Then, for all $(\chi,c)\in \widehat{G} \times G$, we show that there exists $l$ such that if $L' > L >l$ then 
	\begin{equation}\label{eqn:chargemoment}
	\abs{ \omega\left( ( D_{L'}^{\chi,c} - D_{L}^{\chi,c})^2\right) } <\epsilon.
	\end{equation}
	If $\chi \neq \iota$ and $ c\neq e$, \eqref{eqn:chargecont1} and \eqref{eqn:chargecont2} give that $\omega\left( ( D_{L'}^{\chi,c} - D_{L}^{\chi,c})^2\right) = \omega( D_{L'}^{\chi,c} - D_{L}^{\chi,c})$. 
	Similarly, $\omega\left( ( D_{L'}^{\iota,e} - D_{L}^{\iota,e})^2\right) = \abs{\omega( D_{L'}^{\iota,e} - D_{L}^{\iota,e})}$.
	If $\chi \neq \iota$ then 
	\begin{align*}
	\abs{ \omega\left( ( D_{L'}^{\chi,e} - D_{L}^{\chi,e})^2\right) }
	& = \abs{\omega( D_{L'}^{\chi,e} + D_L^{\chi,e} - 2 D_L^{\chi,e}D_{L'}^{\chi,e})}\\
	& \leq \abs{\omega( D_{L'}^{\chi,e}  - D_L^\chi D_{L'}^e)} +  \abs{\omega( D_L^{\chi,e} - D_L^\chi D_{L'}^e)}\\
	& = \abs{\omega\left( D_{L'}^e (D_{L'}^\chi -D_L^\chi)\right)} + \abs{ \omega( D_L^{\chi} (D_{L'}^e - D_L^e))},
	\end{align*}
	with a similar bound holding if $ \chi = \iota$ and $c \neq e$.  Thus, \eqref{eqn:chargemoment} holds.
	
	Lemma \ref{lem:asymptoticcoef} serves to distinguish ground states with different charges and makes it possible to decompose any ground state into charged ground states. 
	\begin{thm}\label{thm:qdoubgs1}
		Let $\omega\in K$ be a ground state.
		Then there exists a convex decomposition of $\omega$  as
		\begin{equation}\label{eqn:gsdecomp4}
		\omega = \sum_{\chi\in \widehat{G},c \in G} \lambda_{\chi,c}(\omega) \omega^{\chi,c} \qquad \text{ where } \quad \omega^{\chi,c}\in K^{\chi,c}.
		\end{equation}
		Furthermore, we can calculate the coefficients explicitly as
		\begin{equation*}
		\lambda_{\chi,c}(\omega) = \lim_{L\ra\infty} \omega(D_{L}^{\chi,c}).
		\end{equation*}
		If $ \lambda_{\chi,c}(\omega) >0$ then 
		\begin{equation}
		\label{eqn:wslimit}
		\omega^{\chi,c} = \wslim_{L\ra\infty} \frac{\omega(\ \cdot \ D_{L}^{\chi,c} )}{\omega(D_{L}^{\chi,c})}.
		\end{equation}
	\end{thm}
	
	\begin{proof}
		For convenience, in this proof we denote $\lambda_{\chi,c}(\omega) = \lambda_{\chi,c}$.
		
		By Lemma \ref{lem:asymptoticcoef}, the values $\lambda_{\chi,c} \geq 0$  are well-defined,
		so we will have to show that the limit in equation~\eqref{eqn:wslimit} exists and that the decomposition in equation~\eqref{eqn:gsdecomp4} agrees with the state $\omega$.
		
		Let $L' > L \geq 2$.  
		Since for each $L'$ the charge projections add up to the identity, by equation~\eqref{eqn:localprojcomplete}, we have $\omega = \sum_{\chi,c} \omega( \ \cdot \ D_{L'}^{\chi,c})$ for all $L'$.
		By Lemma \ref{lem:gshambound}, $\omega|_{\calA_L}$
		is a finite volume ground state for $H_L^{\epsilon,\mu}$ for all $L\geq 2$ (see also the remark after Definition~\ref{def:chargedgs}).
		
		Now suppose $\lambda_{\chi,c}>0$.
		Let $\epsilon >0$ be given and suppose $\epsilon$ is small enough such that $\lambda_{\chi,c} > \epsilon >0$. 
		By Lemma~\ref{lem:asymptoticcoef}, and inequalities~\eqref{eqn:fullcharge} and~\eqref{eqn:chargemoment}, there exists $L>0$ such that if $L''> L' >L$ then
		\[
		\abs{\omega( D_{L'}^{\chi,c}) - \lambda_{\chi,c}}  <   \epsilon, \quad
		\abs{ \omega( D_{L''}^{\chi,c}-  D_{L'}^{\chi,c})} < \epsilon, \quad \text{ and }\quad
		\abs{ \omega\left( ( D_{L''}^{\chi,c} -  D_{L'}^{\chi,c})^2\right) } < \epsilon.
		\]
		We also demand that $\omega(D^{\chi,c}_{L'}) > 0$ for all $L' > L$, which can always be achieved by choosing $L$ big enough.
		Note that $\lambda_{\chi,c} \leq 1$, so we can restrict to $\epsilon < 1$.
		Let $A\in \calA_{L}$, then 
		\begin{align*}
		\left\vert \frac{\omega(A D_{L'}^{\chi,c} )}{\omega(D_{L'}^{\chi,c})}  - \frac{\omega(A D_{L''}^{\chi,c} )}{\omega(D_{L''}^{\chi,c})} \right\vert 
		& =  \frac{1}{\omega(D_{L'}^{\chi,c}) \omega(D_{L''}^{\chi,c})}
		\left\vert  \omega(A D_{L'}^{\chi,c} )\omega(D_{L''}^{\chi,c}) -  \omega(A D_{L''}^{\chi,c} )\omega(D_{L'}^{\chi,c}) \right\vert\\
		&  \leq \frac{1}{\omega(D_{L'}^{\chi,c}) \omega(D_{L''}^{\chi,c})}
		\Big( \left| \omega(A D_{L'}^{\chi,c} )\right| \left\vert   \omega(D_{L''}^{\chi,c}) - \omega( D_{L'}^{\chi,c} ) \right\vert \\
		&  + \omega(D_{L'}^{\chi,c})\abs{\omega(A (D_{L'}^{\chi,c} - D_{L''}^{\chi,c}) )}\Big). 
		\end{align*}
		Recall that by Lemma \ref{lem:ribchargeinv},
		we have that $\omega(A D_{L'}^{\chi,c}) = \omega(D_{L'}^{\chi,c} A D_{L'}^{\chi,c})$.
		It follows that $\abs{\omega(A D_{L'}^{\chi,c})} \leq \|A\| \omega( D_{L'}^{\chi,c})$.
		We also note the estimate $\frac{1}{\omega(D_{L''}^{\chi,c})} \leq \frac{1}{\lambda_{\chi,c} - \epsilon } $.
		The last term can be estimated using the Cauchy-Schwarz inequality,
		\[
		\left| \omega(A(D_{L'}^{\chi,c} - D_{L''}^{\chi,c}) \right| 
		\leq \sqrt{\omega(A^*A)} \sqrt{\abs{\omega\left( (D_{L'}^{\chi,c} - D_{L''}^{\chi,c})^2\right)}}
		\leq \|A\| \sqrt{\epsilon}.
		\]
		Combining these estimates we obtain the bound
		\begin{equation}\label{eqn:cauchyseq}
		\left\vert \frac{\omega(A D_{L'}^{\chi,c} )}{\omega(D_{L'}^{\chi,c})}  - \frac{\omega(A D_{L''}^{\chi,c} )}{\omega(D_{L''}^{\chi,c})} \right\vert \leq \frac{\|A\| (\epsilon + \sqrt{\epsilon})}{\lambda_{\chi,c} - \epsilon}.
		\end{equation}
		Because $\lambda_{\chi,c} > 0$, this goes to zero as $\epsilon$ goes to zero.
		
		Thus for each pair $(\chi,c)$ there is a sequence of states
		$\omega_L^{\chi,c} = \frac{\omega( \ \cdot \ D_{L}^{\chi,c})}{\omega(D_{L}^{\chi,c})},$
		converging in the weak$^*$ limit (or $\omega^{\chi,c}$ is the zero functional if $\lambda_{\chi,c} = 0$).
		We have the following properties (cf.\ equations (4.48--4.50) in~\cite{KomaN}):
		\begin{align*}
		&\wslim_{L\ra\infty} \omega_L^{\chi,c}  := \omega^{\chi,c} \quad \text{ exists}, \\
		&\wslim_{L\ra \infty} \left\vert \omega - \sum_{\chi,c} \lambda_{\chi,c}\omega_L^{\chi,c} \right\vert = 0\\
		& \omega^{\chi,c}_{L'}( H_{L}^{\epsilon,\mu}) = 0  \quad \text{ for all } L'> L\geq 2.
		\end{align*}
		The last property follows directly from the fact that $\omega_{L'}^{\chi,c}\in K_{L}$ for all $L'>L$
		and $\omega^{\chi,c}$ is an infinite volume ground state for all $(\chi,c)$.
		Thus we have proven the ground state decomposition as in equation~\eqref{eqn:gsdecomp4}.
	\end{proof}
	
	\begin{cor}\label{cor:face}
		For all $(\chi,c)\in \widehat{G}\times G$, $K^{\chi,c}$ is a face in the set of all states.
		In particular, if $\omega^{\chi,c} \in K^{\chi,c}$ is an extremal point of $K^{\chi,c}$ 
		then $\omega^{\chi,c}$ is a pure state.
	\end{cor}
	
	\begin{proof}
		Let $\omega^{\chi,c} \in K^{\chi,c}$ and suppose $\phi \leq \lambda \omega^{\chi,c}$. 
		Since the set of ground states is a face
		this implies $\phi$ is a ground state.
		By Theorem \ref{thm:qdoubgs1}, decompose $\phi$ as 
		$\phi = \sum \lambda_{\sigma,d}(\phi)  \phi^{\sigma, d}$.
		By Lemma \ref{lem:asymptoticcoef}, 
		\begin{align*}
		\lambda_{\sigma,d}(\phi) = \lim_{L\ra \infty} \phi(D_L^{\sigma,d})  \leq \lambda \lim_{L\ra\infty}  \omega^{\chi,c}(D^{\sigma,d}_L) = \lambda \delta_{(\sigma,d), (\chi,c)}.
		\end{align*}
		Therefore, $\phi = \phi^{\chi,c} \in K^{\chi,c}$, which shows $K^{\chi,c}$ is a face
		in the set of all states.
		
		Suppose $\omega^{\chi,c}$ is an extremal state of $K^{\chi,c}$
		and that $\omega^{\chi,c}$ can be decomposed as a convex combination of states
		\[ \omega^{\chi,c}  = c_1 \omega_1 + c_2 \omega_2.\]
		It follows that $c_i \omega_i \leq \omega^{\chi,c}$, 
		so by the face property of $K^{\chi,c}$, $\omega_i \in K^{\chi,c}$.
		The supposition that the state $\omega^{\chi,c}$ is extremal in $K^{\chi,c}$ gives that $\omega_1 = \omega_2 =\omega^{\chi,c}$. 
	\end{proof}
	
	We note that the definition we take for a face does not require it to be a closed set in the weak$^*$ topology (but see also Theorem~\ref{thm:weakclosure}).
	The decomposition above suggests that each ground state can be decomposed into ground states that are related to the superselection sectors.
	Indeed, the \emph{pure} states in $K^{\chi,c}$ are equivalent to one of the charged states that we constructed before.
	Two states $\omega_1$ and  $\omega_2$ are said to be \emph{equivalent} if their corresponding GNS representations are unitarily equivalent.
	
	\begin{thm}\label{thm:eqstates}
		If $\omega^{\chi,c} \in K^{\chi,c}$ is a pure state then $\omega^{\chi,c}$ 
		is equivalent to a single excitation ground state $\omega_s^{\chi,c}$,
		as defined in equation~\eqref{eqn:singleexcitation}. 
	\end{thm}
	\begin{proof}
		First, we notice that $\omega_s^{\chi,c} \in K^{\chi,c}$ for all sites $s$ since 
		$ \frac{\omega_s^{\chi,c}( A D_{L'}^{\chi,c})}{\omega_s^{\chi,c}(D_{L'}^{\chi,c})} = \omega_s^{\chi,c}(A)$
		for all $A \in \calA_L$ and $L'>L$, where $L$ is chosen large enough so that $ s \in \mathcal{S}_L$.
		
		Let $\omega^{\chi,c} \in K^{\chi,c}$ and 
		let $\omega \in K$ be a ground state such that 
		$\omega^{\chi,c} = \wslim \frac{\omega( \cdot D^{\chi,c}_L)}{\omega(D_L^{\chi,c})}$ and $\lambda_{\chi,c}>0$.
		Let $\epsilon >0$ be given and suppose $\epsilon$ is small enough such that $\lambda_{\chi,c} > \epsilon >0$. 
		By Lemma \ref{lem:asymptoticcoef} and equation \eqref{eqn:fullcharge}, there is an $L>0$  such that if $L''\geq L' \geq L$ then $\abs{ \omega(D^{\chi,c}_{L''}) - \omega(D_{L'}^{\chi,c}D^{\chi,c}_{L''})}  < \epsilon$ and $\abs{ \omega(D^{\chi,c}_{L''}) - \omega(D_{L'}^{\chi,c})} < \epsilon$.
		
		Fix an operator $ A \in \calA_{loc} \cap \calA_{(L)^c}$.
		Then, there is an $L'>L+1$ such that $A$ is supported on the annulus $A \in \calA_{L'-2} \cap \calA_{(L)^c}$,
		and an $L''> L'+1$ such that 
		\begin{equation*}\label{eqn:bound1}
		\abs{ \omega^{\chi,c}( A) - \frac{\omega(A D^{\chi,c}_{L''})}{\omega(D^{\chi,c}_{L''})}} < \|A\|\epsilon.
		\end{equation*}
		The estimate \eqref{eqn:cauchyseq} also holds.
		Applying \eqref{eqn:fullcharge} and the estimates 
		$$
		\frac{1}{\omega(D_{L'}^{\chi,c})} \leq \frac{1}{\lambda_{\chi,c} - \epsilon } \quad \mbox{ and } \quad
		\frac{1}{\omega(D_L^{\chi,c}D_{L'}^{\chi,c})} \leq \frac{1}{\lambda_{\chi,c} - \epsilon }
		$$
		give
		\begin{align*}
		\abs{ \frac{\omega(A D_{L'}^{\chi,c})}{\omega(D_{L'}^{\chi,c})} - \frac{\omega(A D_{L}^{\chi,c} D_{L'}^{\chi,c})}{\omega(D_{L}^{\chi,c}D_{L'}^{\chi,c})}} 
		&\leq  \abs{  \frac{\omega(A D_{L'}^{\chi,c})}{\omega(D_{L'}^{\chi,c})} - \frac{\omega(A D_{L}^{\chi,c} D_{L'}^{\chi,c})}{\omega(D_{L'}^{\chi,c})}} + \abs{ \frac{\omega(A D_{L}^{\chi,c} D_{L'}^{\chi,c})}{\omega(D_{L'}^{\chi,c})}- \frac{\omega(A D_{L}^{\chi,c} D_{L'}^{\chi,c})}{\omega(D_{L}^{\chi,c}D_{L'}^{\chi,c})}} \\
		&\leq \frac{\|A\| (\epsilon + \sqrt{\epsilon})}{(\lambda_{\chi,c} - \epsilon)^2}.
		\end{align*}
		Combining the above estimates, it follows that for a site $s$ contained in $\Lambda_L$
		\begin{equation}\label{eqn:eqstates1}
		\begin{split}
		& \abs{ \omega^{\chi,c}(A) - \omega_s^{\chi,c}(A) } \leq 
		\abs{\omega^{\chi,c}(A) - \frac{\omega(A D^{\chi,c}_{L''})}{\omega(D^{\chi,c}_{L''})}}
		+ \abs{ \frac{\omega(A D^{\chi,c}_{L''})}{\omega( D^{\chi,c}_{L''})} -  \frac{\omega(A D_{L'}^{\chi,c})}{\omega(D_{L'}^{\chi,c})}} \\
		& \, + \abs{ \frac{\omega(A D_{L'}^{\chi,c})}{\omega(D_{L'}^{\chi,c})} - \frac{\omega(A D_{L}^{\chi,c} D_{L'}^{\chi,c})}{\omega(D_{L}^{\chi,c}D_{L'}^{\chi,c})}}
		+\abs{ \frac{\omega(A D_{L}^{\chi,c} D_{L'}^{\chi,c})}{\omega(D_{L}^{\chi,c}D_{L'}^{\chi,c})} -  \frac{ \omega(D_{L'}^{\chi,c} D_{L}^{\chi,c} A D_{L}^{\chi,c}D_{L'}^{\chi,c})}{\omega(D_{L}^{\chi,c}D_{L'}^{\chi,c})}}\\
		& \qquad \qquad + \abs{  \frac{ \omega(D_{L'}^{\chi,c} D_{L}^{\chi,c} A D_{L}^{\chi,c}D_{L'}^{\chi,c})}{\omega(D_{L}^{\chi,c}D_{L'}^{\chi,c})} - \omega_s^{\chi,c}(A) } \\
		& \quad\quad\quad \leq \|A\|\epsilon  +  2\frac{\|A\| (\epsilon + \sqrt{\epsilon})}{(\lambda_{\chi,c} - \epsilon)^2}
		+  \abs{  \frac{ \omega(D_{L'}^{\chi,c}D_{L}^{\chi,c}A D_{L}^{\chi,c}D_{L'}^{\chi,c})}{\omega(D_{L}^{\chi,c}D_{L'}^{\chi,c})} - \omega_s^{\chi,c}(A) },
		\end{split}
		\end{equation}
		where the fourth term vanishes by Lemma \ref{lem:ribchargeinv}.
		The last term will be shown to be identically zero.
		
		Recall that $\omega_s^{\chi,c}|_{\calA_{L'}} = \omega^0( F_{\rho_{L'}}^{\chi,c *} \ \cdot \ F_{\rho_{L'}}^{\chi,c})|_{\calA_{L'}}$ 
		where $\rho_{L'}$ is a ribbon that connects the site $s$ to the boundary of $\Lambda_{L'}$.
		Denote the subspace $\calG_{L',L}^{\chi,c} := D_{L}^{\chi,c} D_{L'}^{\chi,c} (\calG_{L'}^{\epsilon,\mu})$;
		it is spanned by simple vectors of the form 
		$F_{\sigma}^{\chi,e}F_{\sigma'}^{\iota, c} \Omega$ where $ \Omega \in \calG_{L'}$, and $\sigma$ and $\sigma'$ are a ribbons connecting sites in 
		$\mathcal{S}_L$ to the boundary of $\Lambda_{L'}$, (see Lemma \ref{lem:gsspan}).
		
		First, we consider the case $\psi \in\calG_{L',L}^{\chi,c}$ and $\psi =  F_\sigma^{\chi,c} \Omega$.
		Indeed, for each $\sigma$ as above, we can find a new ribbon, $\sigma' = \sigma_1 \sigma_2 \sigma_3$, see Figure~\ref{fig:eqstates}, connecting 
		$y$ to the boundary of $ \Lambda_{L'}$ with the following properties 
		\begin{align}\label{eqn:pathdecomp}
		&F_{\sigma_1}^{\chi,c}\in \calA_L \text{ and } \sigma_1\cap \Lambda_L =\sigma \cap \Lambda_L\\
		&F_{\sigma_2}^{\chi,c} \in \calA_{L'-1}\cap \calA_{(L+1)^c} \text{ and }  
		\sigma_2 \cap \Lambda_{L'-2}\cap \Lambda_{(L+2)^c}  = \rho_{L'} \cap  \Lambda_{L'-2}\cap \Lambda_{(L+2)^c} \\
		&F_{\sigma_3}^{\chi,c} \in \calA_{L'} \cap  \calA_{( L'-1)^c} 
		\text{ and } \sigma_3\cap \Lambda_{L''} \cap  \Lambda_{( L'-1)^c} =\sigma \cap \Lambda_{L''} \cap  \Lambda_{( L'-1)^c}\\
		&F_\sigma^{\chi,c} \Omega = F_{\sigma_1}^{\chi,c} F_{\sigma_2}^{\chi,c} F_{\sigma_3}^{\chi,c} \Omega. \label{eqn:pathdecomp2}
		\end{align} 
		Here we used that the state only depends on the endpoints of the ribbon, not on the path.
		Decompose 
		$\psi = F_{\sigma}^{\chi,c} \Omega
		= F_{\sigma_{1}}^{\chi,c} F_{\sigma_{2}}^{\chi,c} F_{\sigma_{3}}^{\chi,c}\Omega $.
		
		\begin{figure}
			\begin{center} 
				\includegraphics[width=.5\textwidth]{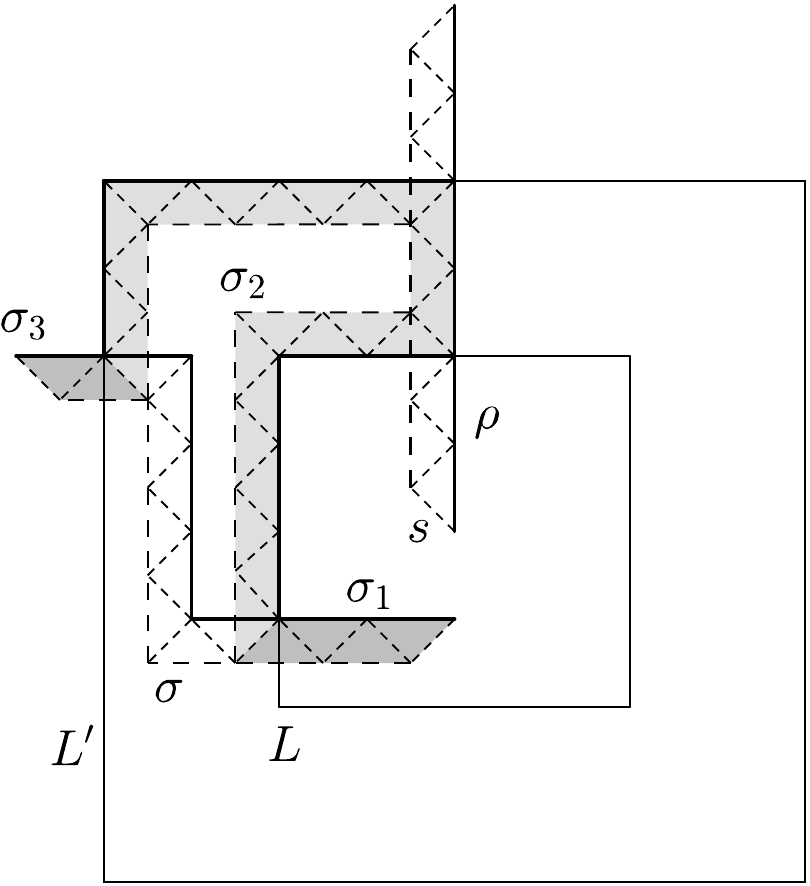}
			\end{center}
			\caption{A depiction of the ribbons $\rho$, $\sigma$ and $\sigma'=\sigma_1\sigma_2\sigma_3$. 
				The ribbon $\sigma'$ is shaded, $\sigma_2$ is distinguished by a lighter shade.}
			\label{fig:eqstates}
		\end{figure}
		
		Suppose $A$ is a product of ribbon operators.
		If $A$ is not a product of closed paths, 
		then its action on $\calG_{L',L}^{\chi,c}$ yields a subspace with strictly higher energy with respect to $H_{L''}^{\chi,c}$.
		Thus, if $\psi \in \calG_{L',L}^{\chi,c}$ then $\langle \psi, A \psi \rangle = 0 = \omega_s^{\chi,c}(A)$.
		If $A$ is a product of closed paths, then $A$ leaves the frustration free ground state invariant,
		$A\Omega = \Omega$.
		Let $k\in\CC$  be such that $A F_{\rho_{L'}}^{\chi,c} = kF_{\rho_{L'}}^{\chi,c} A$;
		$k$ can be computed from ribbon intertwining relations \eqref{eqn:ribbonrelation}.
		It follows that 
		\begin{align*}
		\omega_s^{\chi,c}(A)= \langle F_{\rho_{L'}}^{\chi,c} \Omega, A F_{\rho_{L'}}^{\chi,c} \Omega \rangle = k \langle \Omega, A \Omega \rangle = k.
		\end{align*}
		Now going back to the properties of $\sigma_{l}$ in equations \eqref{eqn:pathdecomp}-\eqref{eqn:pathdecomp2}, we have that 
		\begin{equation*}
		[A, F_{\sigma_1}^{\chi,c}] = 0, \quad [A, F_{\sigma_3}^{\chi,c}] = 0, 
		\quad \text{ and } \quad A F_{\sigma_2}^{\chi,c} = kF_{\sigma_2}^{\chi,c} A .
		\end{equation*}
		Thus, if $A$ is a product of closed ribbon operators then
		\begin{align*}
		A \psi & =  A F_{\sigma_{1}}^{\chi,c} F_{\sigma_{2}}^{\chi,c} F_{\sigma_{3}}^{\chi,c}\Omega =  k F_{\sigma_{1}}^{\chi,c} F_{\sigma_{2}}^{\chi,c} F_{\sigma_{3}}^{\chi,c}A \Omega  = k  \psi
		\end{align*}
		and $\psi$ has eigenvalue $k$.
		For the case $\psi = F_{\sigma}^{\chi,e}F_{\tau}^{\iota,c}\Omega$,
		the decomposition $\sigma = \sigma_1 \sigma_2 \sigma_3$ and $ \tau = \tau_1 \tau_2 \tau_3$ 
		as above then we can choose $ \sigma_2$ and $\tau_2$ to coincide on the annulus $\Lambda_{L'-2}\cap \Lambda_{(L+2)^c}$.
		Therefore, the same argument as above shows $A \psi = k \psi$.
		
		For a general $\psi \in\calG_{L',L}^{\chi,c}$, $\psi$ is a linear combination of the simple vectors $ F_\sigma^{\chi,e} F_{\sigma'}^{\iota,c}\Omega$.
		Thus, by linearity $A\psi = k \psi$ for all $\psi \in\calG_{L',L}^{\chi,c}$.
		Therefore, if $\psi$ is normalized 
		\begin{equation*}
		\langle \psi, A \psi \rangle = k = \omega_s^{\chi,c}(A).
		\end{equation*}
		Note that we already established this equation for $A$ an open ribbon operator.
		
		Since ribbon operators span the algebra $\calA_{L'-2} \cap \calA_{(L+2)^c}$, 
		we extend the above argument by linearity 
		so that 
		\begin{equation*}
		\langle \psi, A \psi \rangle  = \omega_s^{\chi,c}(A)\quad \text{ for all } \quad A \in \calA_{L'-2} \cap \calA_{(L+2)^c}.
		\end{equation*} 
		
		A general mixed state supported on $\calG_{L',L}^{\chi,c}$ is of the form 
		$$ 
		\phi = \frac{ \omega_{L'}(D_{L'}^{\chi,c} D_{L}^{\chi,c} \ \cdot \ D_{L}^{\chi,c}D_{L'}^{\chi,c})}{\omega_{L'}(D_{L}^{\chi,c}D_{L'}^{\chi,c})} =
		\sum c_{\psi} \langle \psi, A \psi\rangle,
		$$ 
		where $ \omega_{L'} \in K_{L'}$ and each $\psi$ is a linear combination of vectors of the form $F_{\sigma}^{\chi,e}F_\xi^{\iota,c}\Omega$.
		Since the $c_\psi$ add up to one, it follows that $ \phi(A) = \omega_s^{\chi,c}(A)$ for all $A \in \calA_{L'-2} \cap \calA_{(L+2)^c}$.
		
		For the ground state $\omega$, Lemma \ref{lem:gshambound} gives $\omega|_{\calA_{L'}} \in K_{L'}$.
		Therefore,
		\begin{equation*}
		\frac{ \omega(D_{L'}^{\chi,c}D_{L}^{\chi,c}A D_{L}^{\chi,c}D_{L'}^{\chi,c})}{\omega(D_{L}^{\chi,c}D_{L'}^{\chi,c})} = \omega_s^{\chi,c}(A) \quad \quad \text{ for all } A \in \calA_{L'-2} \cap \calA_{(L+2)^c}.
		\end{equation*}
		
		Since $L'$ was chosen such that $L'>L$, and otherwise arbitrary, the estimate in \eqref{eqn:eqstates1} becomes
		\begin{equation*}
		\abs{ \omega^{\chi,c}(A) - \omega_s^{\chi,c}(A) } \leq  \|A\| \left(\epsilon +2\frac{(\epsilon + \sqrt{\epsilon})}{(\lambda_{\chi,c} - \epsilon)^2}\right) \qquad \text{ for all } \quad A \in \calA_{loc}\cap \calA_{L^c}.
		\end{equation*}
		
		Now suppose further that $\omega^{\chi,c}$ is a pure state.
		Proposition \ref{prop:singleexc} also gives that the states $\omega_s^{\chi,c}$ are pure states.
		Therefore, applying the criterion for equivalence of pure states (Corollary 2.6.11, \cite{BratteliR1})
		we have 
		\begin{equation*}
		\omega^{\chi,c} \approx \omega_s^{\chi,c} \text{ for all } (\chi,c) \in \widehat{G}\times G.
		\end{equation*}
		This completes the proof.
	\end{proof}
	
	The above two theorems give a complete characterization of the ground states of the quantum double model.
	The sets of states $K^{\chi,c}$ played an important role in the analysis.
	We end our discussion by finding the weak$^*$ closure of these sets in the set of all states.
	Recall that each state in $K^{\chi,c}$ has a charge $(\chi,c)$.
	Now consider a sequence of states where the $\chi$ charge is gradually moved off to infinity.
	The resulting weak$^*$ limit will be a state with only a charge $c$,
	so the weak closures of the sets $K^{\chi,c}$ are strictly larger than $K^{\chi,c}$ (unless $\chi = \iota$ and $c = e$).
	Finally, we see that this procedure suffices to give the  weak$^*$ closures of the sets of charged ground states.
	
	\begin{thm}\label{thm:weakclosure}
		The closures in the weak$^*$ topology are given by
		\begin{equation*}
		\overline{K^{\chi,c}}^w = \conv{K^0\cup K^{\chi,e} \cup K^{\iota,c} \cup K^{\chi,c}},
		\end{equation*}
		where $\operatorname{Conv}$ denotes the convex hull.
	\end{thm}
	\begin{proof}
		
		First, we show that $K^{\chi,e} \subset \overline{K^{\chi,c}}^w$.
		
		Let $\omega^{\chi,e} \in K^{\chi,e}$ and $\rho$ be a path extending to infinity based 
		at site $s$. Consider the automorphism $\alpha_\rho^{\iota, c}$ (see \eqref{eqn:chargemorp}) 
		that generates a charge of type $(\iota,c)$ located at the site $s$.
		We claim that the state $\omega^{\chi,e} \circ \alpha_\rho^{\iota,c} \in K^{\chi,c}$.
		To see this, from Theorem \ref{thm:qdoubgs1} and Lemma \ref{lem:asymptoticcoef},
		write $ \omega^{\chi,e}(A) = \lim_{L\ra \infty} \omega^{\chi,e}( A D_{L}^{\chi,e} )/\omega^{\chi,e}(D_{L}^{\chi,e})$
		for $A \in \calA$.
		Notice that for $L'>L$ large enough such that $s \subset \Lambda_L$, equation \eqref{eqn:localprojribbonrelation1} gives
		$F_{\rho_{L'}}^{\iota,c} D_L^{\chi,e} = D_L^{\chi,c}F_{\rho_{L'}}^{\iota,c}$
		and $(\alpha_{\rho}^{\iota,c })^{-1}(D_L^{\chi,e}) = D_L^{\chi,c}$.
		Thus,
		\begin{equation*}
		\omega^{\chi,e} \circ \alpha_\rho^{\iota, c} (A) 
		= \lim_{L\ra \infty} \frac{ \omega^{\chi,e}(  \alpha_\rho^{\iota, c}(A) D_{L}^{\chi,e} )}{\omega^{\chi,e}(D_{L}^{\chi,e})} 
		= \lim_{L\ra \infty} \frac{ \omega^{\chi,e}\circ \alpha_\rho^{\iota, c}(A D_{L}^{\chi,c}) }{\omega^{\chi,e}\circ \alpha_\rho^{\iota, c} (D_{L}^{\chi,c})}.
		\end{equation*}
		To finish the claim, we need to show $\omega^{\chi,e} \circ \alpha_\rho^{\iota, c} $ is a ground state.
		Recall that $ D_{L'}^{\chi,e}\calG_{L'}^{\epsilon, \mu}$ is spanned by simple vectors of the form
		$F_\sigma^{\chi,e} \Omega$ 
		where $ \Omega \in \calG_{L'}$ and 
		$\sigma$ is a ribbon connecting a site $s \in \mathcal{S}_{L'}$ to the 
		boundary of $\Lambda_{L'}$.
		Let $ \psi \in D_{L'}^{\chi,e}\calG_{L'}^{\epsilon, \mu}$ and write $ \psi = \sum_j b_j F_{\sigma_j}^{\chi,e}\Omega_j$.
		Since $F_{\rho_{L'}}^{\iota,c} \psi = \sum_j b_j F_{\rho_{L'}}^{\iota,c} F_{\sigma_j}^{\chi,e}\Omega_j \in D_{L'}^{\chi,c}\calG_L^{\epsilon,\mu}$, it follows that 
		$H_L^{\epsilon,\mu} F_{\rho_{L'}}^{\iota,c} \psi = 0$.
		Indeed, we compute $\omega^{\chi,e}\circ\alpha_{\rho}^{\iota, c}$ is an infinite volume ground state:
		\begin{align*}
		\omega^{\chi,e}(\alpha_{\rho}^{\iota, c}( H_L^{\epsilon,\mu}))
		&= \lim_{L'\ra\infty} \frac{ \omega^{\chi,e}(\alpha_{\rho}^{\epsilon,\mu}( H_L^{\epsilon,\mu} D_{L'}^{\chi,c}))}
		{\omega^{\chi,e}\circ\alpha_{\rho}^{\iota, c}(D_{L'}^{\chi,c})}
		= \lim_{L'\ra\infty} \frac{ \omega^{\chi,e}(\alpha_{\rho}^{\epsilon,\mu}( D_{L'}^{\chi,c} H_L^{\epsilon,\mu} D_{L'}^{\chi,c}))}
		{\omega^{\chi,e}\circ\alpha_{\rho}^{\iota, c}(D_{L'}^{\chi,c})}\\
		&= \lim_{L'\ra\infty} \frac{  \omega^{\chi,e}( D_{L'}^{\chi, e} F_{\rho_{L'}}^{\iota,c *} H_{L}^{\epsilon,\mu}F_{\rho_{L'}}^{\iota,c} D_{L'}^{\chi,e} )}{\omega^{\chi,e}( D_{L'}^{\chi, e})}\\
		& = 0,
		\end{align*}
		where the last equality is true since the state 
		$$
		\frac{  \omega^{\chi,e}( D_{L'}^{\chi, e} F_{\rho_{L'}}^{\iota,c *} \ \cdot \ F_{\rho_{L'}}^{\iota,c} D_{L'}^{\chi,e} )}{\omega^{\chi,e}( D_{L'}^{\chi, e})}
		$$
		is a mixed state supported on $D_{L'}^{\chi,c}\calG_L^{\epsilon,\mu}$.
		
		Now consider a sequence $s_n$ of sites such that $s_1 = s$ and $ s_n \ra \infty$ as $n\ra \infty$.
		Let $\rho_n$ be a ribbon extending to infinity based at the site $s_n$ 
		and then define the sequence of states 
		\[  \omega_{n}=  \omega^{\chi,e}\circ\alpha_{\rho_n}^{\iota,c} \in K^{\chi,c}.\]
		For $A \in \calA_{loc}$, choose $n$ large enough so that
		$ \alpha_{\rho_n}^{\iota, c}(A) = A$.
		It follows that,
		\begin{align*}
		\omega_{n}( A ) & = \omega^{\chi,e}\circ\alpha_{\rho_n}^{\iota,c}(A)= \omega^{\chi, e}(A).
		\end{align*}
		Therefore, $\wslim \omega_n = \omega^{\chi, e} \in K^{\chi,e}$.
		By similar arguments one can show the inclusion,
		\[\conv{K^0\cup K^{\chi,e} \cup K^{\iota,c} \cup K^{\chi,c}} \subset  \overline{K^{\chi,c}}^w.\]
		
		Now, to show the reverse inclusion, suppose $\widehat{\omega}^{\chi,c} \in \overline{K^{\chi,c}}^w$ 
		and let $\omega^{\chi,c}_\lambda \in K^{\chi,c}$ be a net in $K^{\chi,c}$ such that $ \wslim_{\lambda}\omega_\lambda^{\chi,c} = \widehat{\omega}^{\chi,c}$.   
		For each $ \lambda$, we can write
		$$
		\omega_\lambda^{\chi,c} = \lim_{L \ra \infty } \frac{\omega_\lambda( \ \cdot \ D_L^{\chi,c})}{\omega_\lambda(D_L^{\chi,c})} .
		$$
		The proof of Lemma \ref{lem:asymptoticcoef}
		gives that $ D^{\sigma, d}_L D^{\chi,c}_{L'} |_{\calG_{L'}^{\epsilon,\mu}} = 0$ if $(\sigma, d)$ is not in the set $\{ (\chi,c), (\chi,e), (\iota,c), (\iota, e) \}  $.
		Thus, in that case,
		\begin{align}\label{eqn:gswslim}
		\widehat{\omega}^{\chi,c}( D_L^{\sigma, d}) &= \lim_{\lambda } \omega_{\lambda}^{\chi,c}( D_L^{\sigma, d}) = \lim_{\lambda} \lim_{L'\ra\infty }\frac{\omega_{\lambda}^{\chi,c}(D_L^{\sigma, d} D_{L'}^{\chi,c} )}{\omega_\lambda^{\chi,c}(D_{L'}^{\chi,c})} = 0.  
		\end{align}
		This holds for all $L$, hence $\lambda_{\sigma,d} = 0$.
		Since the set of infinite volume ground states is closed in the weak$^*$ topology, 
		we apply equation~\eqref{eqn:gswslim} to the ground state decomposition~\eqref{eqn:gsdecomp4} of $\widehat{\omega}^{\chi,c}$,
		\begin{equation*}
		\widehat{\omega}^{\chi,c} = \lambda_0 \omega^0 +\wslim_{L\ra\infty} \bigg( \lambda_{\chi, e} \frac{\widehat{\omega}^{\chi,c}( \ \cdot \ D_L^{\chi, e})}{\widehat{\omega}^{\chi,c}(D_L^{\chi, e})} 
		+c_{\iota, c}  \frac{\widehat{\omega}^{\chi,c}( \ \cdot \ D_L^{\iota, c})}{\widehat{\omega}^{\chi,c}(D_L^{\iota, c})}
		+c_{\chi, c} \frac{\widehat{\omega}^{\chi,c}( \ \cdot \ D_L^{\chi, c})}{\widehat{\omega}^{\chi,c}(D_L^{\chi, c})}\bigg).
		\end{equation*}
		Therefore, $\widehat{\omega}^{\chi,c} \in \conv{K^0\cup K^{\chi,e} \cup K^{\iota,c} \cup K^{\chi,c}}$.
	\end{proof}
   
	\chapter[% 
	Stability of charges
	]{% 
		On the stability of charges in infinite quantum spin systems
	}%
	\label{ch:stability}

We consider a family of charge superselection sectors of infinite quantum spin systems
corresponding to almost localized endomorphisms.  
If the vacuum state is a pure state and satisfies certain locality conditions, we show how to recover the charge statistics. 
In particular, the superselection structure is that of a braided tensor category.
Further, we introduce a superselection criterion selecting excitations with energy below a threshold.
When the threshold energy falls in a gap of the spectrum of the vacuum state,
we prove stability of the entire superselection structure.
We apply our results to prove stability of anyons for Kitaev's quantum double.

The results in this chapter are joint work with Pieter Naaijkens and Bruno Nachtergaele \cite{ChaNN2} and currently being prepared for publication.

\section{Superselection sectors}

Let $\overline{\Lambda}_\alpha$ be the infinite cone in $\RR^\nu$
based at the origin pointing in the positive $x_1$ direction with angle  $0<\alpha< \pi$, 
\begin{equation*}
\overline{\Lambda}_\alpha \equiv  \left\{ (x_1, x_2, \ldots, x_\nu): x_1 \tan(\alpha) > \left( \sum_{i=2}^\nu x_i^2\right)^{1/2} \right\}.
\end{equation*}
Consider the corresponding cone in $\ZZ^\nu$ and  translations of the cone given by
\begin{equation*}
\Lambda_\alpha  \equiv \overline{\Lambda}_\alpha  \cap \ZZ^\nu \qquad \text{ and } \qquad
\Lambda_\alpha+ x \equiv \left(\overline{ \Lambda}_\alpha + x \right) \cap \ZZ^\nu\quad \text{ for } \quad x \in \ZZ^\nu,
\end{equation*}
respectively.
We will abuse notation to denote integer shifts of the cone as 
\begin{equation*}
\Lambda_\alpha+ n = \left(\overline{\Lambda}_\alpha + (n, 0, 0 \ldots, 0) \right) \cap \ZZ^\nu\quad \text{ for } \quad n \in \ZZ,
\end{equation*}
Denote the set of all cones in $\ZZ^\nu$ with angle less than $\alpha$ as

\begin{equation*}
\mathcal{C}_\alpha \equiv \{ \Lambda_{\beta} \subset \Gamma: 0\leq \beta<\alpha; \  \Lambda_\beta = R_\theta( \overline{ \Lambda}_\beta + x) \cap \ZZ^\nu \ \mbox{for some} \ \theta \in SO(n) \mbox{ and} \  x \in \ZZ^\nu \}
\end{equation*} 
where $R_\theta $ is the $ SO(\nu)$  action on $\RR^\nu$
and the set of all cones as $\mathcal{C}= \bigcup_{ 0\leq \alpha \leq 2\pi}  C_\alpha$, see Figure~\ref{fig:cone} for a sample cone.
For each cone $\Lambda_\alpha \in \mathcal{C}$ we can define translations by $ \Lambda_\alpha + x$ for all $x \in \RR^\nu$ 
and integer shifts by $\Lambda_\alpha + n \equiv \Lambda_{ \alpha } + n \hat{\lambda}$ where $\hat{\lambda} \in \RR^\nu$ 
is the unit vector pointing in the positive direction of the cone.

\begin{figure}
	\includegraphics[width=.5\textwidth]{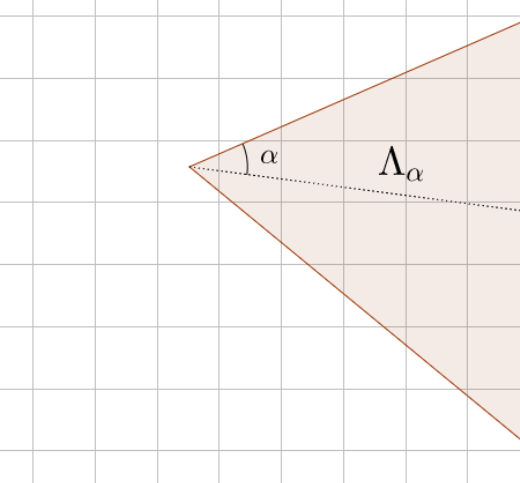}
	\caption{The convex cone $\Lambda_{ \alpha  }$ includes all points  in $\ZZ^\nu$ contained in the shaded region.}
	\label{fig:cone}
\end{figure}

\subsection{Vacuum state}
Let $\omega_0$ be a state of the quasi-local algebra $\calA$ and $(\pi_0, \Omega_0, \calH_0)$ its GNS representation.
Since $\calA$ is simple and $\pi_0$ is a faithful, we will typically abuse notation to  write $ \pi_0(A) =A $ for $ A \in \calA$.
The state $\omega_0$ will play the role of the `vacuum state' in that we use it the build elementary quasi-particle states.
In the quantum spin setting, $\omega_0$ will typically be a pure and translation-invariant ground state
for a dynamical system $(\calA, \tau_t = e^{it \delta}),$
that is,  $\omega_0 \circ T_x = \omega_0$ for all $x\in \ZZ^\nu$ where $T_x$ is the natural action of $Z^\nu$,
and $ \omega_0 (A^* \delta(A)) \geq 0$ for $A \in D(\delta)$.

Let $\mathcal{R}(\Lambda) \equiv \pi_0(\calA_\Lambda)''$ denote the von Neumann closure in $\calB(\calH_0)$ for any subset $\Lambda \subset \Gamma$.

\begin{defn}
	$\omega_0$ is said to satisfy \emph{Haag duality} for cones if   for each cone $\Lambda \in \mathcal{C}$ we have that
	\begin{equation*}
	\mathcal{R}(\Lambda) = \mathcal{R}(\Lambda^c)'.
	\end{equation*}
	
	$\omega_0$ is said to have an \emph{approximate split property} for cones
	if there is a relation $\ll$ on the set of cones $\mathcal{C}$ such that 
	if $\Lambda_1  \ll \Lambda_2$  then $ \Lambda_1$ is strictly contained in $\Lambda_2$ 
	and $\Lambda_2^c \ll \Lambda_1^c$, 
	and there is a type I factor $\mathcal{N}$ such that
	\begin{equation}\label{eqn:split}
	\mathcal{R}(\Lambda_1) \subset \mathcal{N} \subset \mathcal{R}(\Lambda_2).
	\end{equation}
\end{defn}

We will make the following assumptions on the state $\omega_0$:

\begin{assumption}\label{ass:hdsp}
	$\omega_0$ is a pure state and satisfies Haag duality and the approximate split property for cones. 
\end{assumption}

Let $\Lambda_1, \Lambda_2 \in \mathcal{C}$ be such that  $\Lambda_1 \ll \Lambda_2 $. 
By the approximate split property, there exists a type I factor $\mathcal{N}$ such that \eqref{eqn:split} holds
and $\calB(\calH_0) \cong \mathcal{N} \otimes \mathcal{N}'$. 
Taking commutants in \eqref{eqn:split} and applying Haag duality gives the relation
\begin{align*}
\mathcal{R}(\Lambda_2)' & \subset \mathcal{N}' \subset \mathcal{R}(\Lambda_1)'\\
& \quad \Downarrow \quad\quad\quad\quad\quad\quad\quad\quad\mbox{(Haag duality)} \\
\mathcal{R}(\Lambda_2^c) & \subset \mathcal{N}' \subset \mathcal{R}(\Lambda_1^c).
\end{align*}
Thus, $\mathcal{N}'$ plays the role of the type I factor in the approximate split property for the cones $ \Lambda_2^c \ll \Lambda_1^c$.

A state $\omega_0$ has the (non-approximate) split property \cite{DoplicherL,Matsui01}  if for all cones $\Lambda \in \mathcal{C}$
there exists a type I factor $\mathcal{N}$ such that 
\begin{equation*}
\mathcal{R}(\Lambda) \subset \mathcal{N} \subset \mathcal{R}(\Lambda^c)'.
\end{equation*}
Now suppose $\nu =1$ and consider the decomposition $\calA = \calA_L \otimes \calA_R$ where $\calA_L = \calA_{(-\infty,x]} $ and 
$\calA_R = \calA_{[x+1, \infty)}$.
Then, a pure state  $\omega_0$ has the split property if and only if $\omega_0$ is quasi-equivalent to a product state $\omega_L \otimes \omega_R$ 
where $\omega_L = \omega_0|_{\calA_L}$ and $\omega_R = \omega_0|_{\calA_R}$ for all $x \in \ZZ$.
\cite{Matsui01}.
Thus, the split property can be interpreted as a weak correlation in $\omega_0$  across the left and right regions.
In particular, if $\omega_0$ satisfies a uniform exponential decay of correlations then $\omega_0$ satisfies the split property \cite{Matsui13}.
If $\omega_0$ is a pure state then $\omega_L$ and $ \omega_R$ will be type I \cite{Matsui01} and thus, $\pi_0(\calA_L)'' \cong \pi_L(\calA_L)''$ 
will be a type I factor where $ \pi_L$ is a GNS-representation of $\omega_L$.
If $\nu \geq 2$ then we cannot generally expect  $\mathcal{R}(\Lambda)$ to be of type I \cite{Naaijkens11}. 
Thus, if $ \omega_0$ is a pure state then we cannot expect that $\omega_0$ will satisfy the (non-approximate) split property in higher dimensions.
However, the approximately split property can still be interpreted as  
weak correlation across the regions $\Lambda_1$ and $\Lambda_2$.

\subsection{Almost localized endomorphisms}

In the standard Doplicher-Haag-Roberts analysis in local quantum physics, 
representations that are equivalent to the vacuum representation outside a double cone are considered as charged sectors \cite{DHR1,DHR2}.
Here we consider representations obtained by application of some $*$-endomorphisms.  
Generally, any cyclic representation of the quasi-local algebra can be obtained, up to equivalence, by application of an asymptotically inner $*$-endomorphism \cite{Kishimoto}.

We define a notion of almost localized endomorphisms for a cone.
Let $ \mathcal{F}_\infty$ denote the family of continuous functions $f: \RR^+\times \RR^+ \ra \RR^+$ such that 
$f_\epsilon(n)$ is non-increasing in both variables and 
\begin{equation}\label{eqn:Oinf}
\lim_{n \ra \infty } n^{k} f_\epsilon(n)  = 0  \quad \mbox{ for all } \quad  k \in \NN.
\end{equation}
Generally, we say that $ f \in O(n^{-\infty})$ if $f$ satisfies \eqref{eqn:Oinf}.
Notice that $\mathcal{F}_\infty$ is closed under addition.
The family $\mathcal{F}_\infty$ will be used to measure locality of operators and endomorphisms of $\calA$.

\begin{defn}\label{defn:alend}
	A $*$-endomorphism $\rho$ is said to be \emph{almost localized} in a cone $\Lambda_\alpha \in \mathcal{C}$ if 
	there exists a function  $f \in \mathcal{F}_\infty$  such that 
	\begin{equation*}
	\sup_{A \in \calA(\Lambda^c_{\alpha+\epsilon}-n)} \frac{\| \rho(A) - A\|}{\|A\|} \leq f_\epsilon(n)  \quad \mbox{for all  \quad } 0< \epsilon< \pi - \alpha.
	\end{equation*}
	The function $f$ is called the \emph{decay function} for $\rho$ in $\Lambda_\alpha$.
\end{defn}

It follows that if $\rho$ is almost localized in $\Lambda_\alpha$ with decay function $f$ 
then $ \rho$  is almost localized in $ \Lambda_{ \alpha + \delta}$ for all $ \delta < \pi - \alpha $ with the same decay function.
On the other hand, $\rho$ is not guaranteed to be almost localized in cones with smaller opening angle.
In particular,  as $\epsilon \ra 0$ the function $f_\epsilon$ could diverge.

\emph{Example.}
Let  $U_n \in \calA_{ \Lambda_\alpha}$ be a sequence of unitary operators each supported on the cone $\Lambda_{ \alpha }$
and suppose the limit 
$\rho(A) \equiv \lim_{n\ra \infty} U_n^* A U_n$ exists.
Then, $\rho$ is a $*$-endomorphism with $\rho(B) = B$ for all $ B \in \calA_{ \Lambda_{\alpha}^c}$
and it follows that $\rho$ is almost localized in $\Lambda_\alpha$ for all decay functions in $\mathcal{F}_\infty$.
In this case, $\rho$ is said the be exactly localized in $\Lambda_{ \alpha }$.
The superselection structure for localized endomorphisms for the toric code model was first studied in \cite{Naaijkens11}.

In the definition of almost localized the class of functions defining the decay rate $\mathcal{F}_\infty$ could in principle be weakened, 
however it is crucial that composition of endomorphisms preserve the almost localized property.

\begin{prop}
	Let $\rho$ and $\sigma$ be almost localized in $\Lambda_\alpha$ with decay functions $f$ and $ g$, respectively. 
	Then, $\rho \circ \sigma$ is almost localized in $\Lambda_\alpha$ with decay function $ f+ g$.
\end{prop}

\begin{proof}
	This follows from the triangle inequality 
	\begin{align*}
	\sup_{A \in \calA(\Lambda^c_{\alpha+\epsilon}-n)}\frac{\| \rho\circ \sigma (A) - A \| }{\|A\|} 
	& \leq \sup_{A \in \calA(\Lambda^c_{\alpha+\epsilon}-n)} \frac{\|\rho\| \| ( \sigma(A) - A) \|}{\|A\|} + \frac{\| \rho(A) - A \|}{\|A\|}  \\
	&\leq f_\epsilon(n) + \| \rho\| g_\epsilon(n),
	\end{align*} 
	where $\|\rho\| = 1$ since $\calA$ is simple.
\end{proof}

Let $\rho$ be almost localized in an infinite cone $\Lambda_\alpha$ with decay function $f$.
Then, for any fixed $n \in \ZZ$ 
\[  \sup_{A \in \calA_{\Lambda_{\epsilon}^c-(n+n') } }\frac{\| \rho(A) - A\| }{\| A\| } \leq f_\epsilon(n+n'). \]
Thus, $\rho$ is almost localized in $\Lambda_\alpha + n$ with the decay function $f_\epsilon(n+n')$.
If $f$ is submultiplicative in the second variable then $f_\epsilon(n + n') \leq f_\epsilon(n) f_\epsilon(n')$.

Consider translations of the cone $\Lambda_\alpha+x$
and the corresponding translated endomorphism $\rho_x \equiv T_{-x} \circ \rho \circ T_{x}$.
Then,
\begin{align*}
\sup_{A \in \calA_{\Lambda_{\epsilon}^c +(x-n)}}\frac{\| \rho_x(A) - A\| }{\| A\| } 
& =  \sup_{A \in \calA_{\Lambda_{\epsilon}^c-n } }\frac{\|\rho_x(T_{x}(A)) - T_{x}(A)\| }{\| T_{x}(A)\| } \\
& = \sup_{A \in \calA_{\Lambda_{\epsilon}^c-n } }\frac{\|T_{x}( \rho(A) - A)\| }{\| T_{x}(A)\| } 
\qquad \leq f_\epsilon(n)
\end{align*}
Thus, the shifted endomorphism $\rho_x$ is almost localized in $\Lambda_\alpha+{x}$ with the same decay function.
In particular, if $x$ does not  the shifted endomorphisms converge pointwise to the identity,
\begin{equation*}
\lim_{\substack{ \abs{x} \ra \infty \\ d(\Lambda_{ \alpha  }+x, 0) \ra \infty}} \| \rho_x(A) - A \| = 0  \quad \text{ for all } \quad A \in \calA.
\end{equation*}

For a $*$-endomorphisms $ \rho$, we consider the cyclic representation of the form $(\pi_0\circ \rho, \Omega_0, \calH_0)$.
Let $ \rho \cong \rho'$ denote unitary equivalence of the corresponding representations $ \pi_0\circ \rho \cong \pi_0 \circ \rho'$.
Since $\pi_0$ is faithful we will often abuse notation to write $ \pi_0 \circ \rho$ as simply $ \rho$.

\begin{defn}\label{defn:transportable}
	Let $\rho$ be almost localized in $\Lambda_\alpha$ with a decay function $f$.
	We say that $\rho$ is \emph{transportable} with respect to $\omega_0$ if
	% $T_{-x} \circ \rho\circ T_x \cong \rho$ 	  for all $x \in \ZZ^\nu$,
	\begin{itemize}
		\item for each cone $\Lambda'_\beta \in \mathcal{C}$ there exists a $*$-endomorphism $\rho'_\beta$ almost localized in $\Lambda'_\beta$ 
		such that $\rho\cong \rho'_\beta$,
		\item further, if $\beta \geq \alpha$ then $\rho'_\beta$ can be chosen to have decay function $f$ in $\Lambda'_\beta$.
	\end{itemize} 
\end{defn}

When the context is clear, we simply say that $ \rho$ is transportable.
In translation invariant models, we often have that translations are covariant in charged sectors, that is, 
$ T_{-x} \circ \rho\circ T_x \cong \rho$ 	  for all $x \in \ZZ^\nu$.  
However, in our results we do not need to assume this.

Notice that almost locality of $\rho$ is a property on the level of the algebra of observables
and does not depend on the vacuum state, whereas transportability has clear dependence to the vacuum state.
In models, like the toric code model, transportability can typically be proven from the properties of the ground state such as
path independence of string-like operators and translation invariance.

Generally, any cyclic representation of the quasi-local algebra $\calA$ can be obtained, up to equivalence, by 
composing $\pi_0$ with an asymptotically inner endomorphism \cite{Kishimoto}, but not all are transportable.
Here we give an explicit construction of a sequence of unitaries implementing the asymptotically inner property for 
cyclic representations of the form $\pi_0 \circ \rho$ where $\rho$  is almost localized and transportable.
Similar constructions have been considered in models for electromagnetic charge \cite{BuchholzEM}.

\begin{lemma}\label{lem:asympinner}
	Let $\rho$ be almost localized and transportable $*$-endomorphism on $\calA$.
	Then, $\rho$ is asymptotically inner in $\calB(\calH_0)$, that is,
	there exists a sequence of unitary operators $U_n \in \calB(\calH_0)$ such that 
	\begin{equation*}
	\rho(A) = \lim_{n \ra \infty } U_n^* A U_n \quad \text{ for all } \quad A \in \calA.
	\end{equation*}
\end{lemma}

\begin{proof}
	By transportability, there exists a sequence of $*$-endomorphisms $ \rho_n$ almost localized in $\Lambda_{ \alpha  } + n$ with decay function $f$ and 	unitary operators $U_n \in (\rho, \rho_n)$ such that 
	\begin{equation*}
	U_n \rho(A) = \rho_n(A)U_n \quad \text{ for all } \quad A \in \calA.
	\end{equation*}
	
	Let $\epsilon >0$ be given.
	For  $A \in \calA_{loc}$ there exists $N>0$ be such that $ A \in \calA_{ \Lambda_{ \alpha  +\epsilon}^c +N}$. 
	It follows that 
	\begin{align*}
	\| \rho(A) - U_n^* A U_n\| &=  \| U_n^* ( \rho_n(A) - A ) U_n \| 
	= \| \rho_n(A) - A \| \\
	& \leq f(n - N) \| A\| \qquad  \ra 0 \quad \mbox{ as } \quad n \ra \infty.
	\end{align*}
	In the last inequality, transportability implies that each $ \rho_n$ can be assigned the same decay function $f$.
	Since $\calA_{loc}$ is dense in $\calA$, we have $\rho(A) = \lim_{n \ra \infty } U_n^* A U_n$ for all $A \in \calA$.
\end{proof}

\subsection{Locality structure for intertwiners}

Let $\rho$ and $ \sigma$ be $*$-endomorphisms.
Define the space of interwiners as 
\begin{equation*}
(\rho, \sigma)_{\pi_0}\equiv \{ T \in \calB(\calH) : T \pi_0( \rho (A) )= \pi_0(\sigma(A)) T, A \in \calA\}.
\end{equation*}
When the context is clear, we drop the subscript $\pi_0$ and write $(\rho, \sigma) = (\rho,\sigma)_{\pi_0}$.
Notice that if $R \in (\rho,\sigma)$ and $S \in (\sigma, \tau)$ then
\begin{equation*}
SR \rho(A) = S \sigma(A) R = \tau (A) SR,
\end{equation*} 
so that $ SR \in (\rho, \tau)$ intertwines $\rho$ and $\tau$.
We analyze the locality structure for intertwiners between almost localized endomorphisms.

In the context of local quantum physics, almost local observables were first introduced in \cite{ArakiH} and further studied 
in the context of quantum spin systems in \cite{Schmitz,BachmannDN}.
Here, we study almost local observables of $\overline{\pi_0(\calA)}^w = \calB(\calH_0)$  with respect to a cone region.

\begin{defn}\label{def:alocoperator}
	An operator $A \in \calB(\calH)$ is said to be \emph{almost localized} in a cone $\Lambda_{\alpha} \in \mathcal{C}$ if 
	there exists a function $f\in \mathcal{F}_\infty$ such that 
	\begin{equation}\label{eqn:alocobs}
	\sup_{B\in \mathcal{R}(\Lambda^c_{\alpha+\epsilon}-n)} \frac{ \| AB-BA \| }{\| B \| } \leq f_\epsilon(n) = O(n^{-\infty})
	\quad \mbox{ for all } \quad 0 < \epsilon < \pi - \alpha.
	\end{equation}
	The function $f$ is called the \emph{decay function} for $A$ in $\Lambda_{\alpha}$.
\end{defn}

The following lemma shows that the supremum in \eqref{eqn:alocobs} can be taken over the $C^*$-subalgebra of observables $\calA_{\Lambda^c_{\alpha+\epsilon}-n} \subset  \mathcal{R}(\Lambda^c_{\alpha+\epsilon}-n)$.
\begin{lemma}
	$A \in \calB(\calH)$ is almost localized in a cone $\Lambda_{\alpha}$ if and only if
	there exists a  function $f\in \mathcal{F}_\infty$ such that 
	\begin{equation*}
	\sup_{B\in \calA_{\Lambda^c_{\alpha+\epsilon}-n}} \frac{ \| AB-BA \| }{\| B \| } \leq f_\epsilon(n) = O(n^{-\infty})
	\quad \mbox{ for all } \quad 0 < \epsilon < \pi - \alpha.
	\end{equation*}
\end{lemma}

\begin{proof}
	The forward direction of the lemma  follows directly for the fact $\calA_{\Lambda^c_{\alpha+\epsilon}-n} \subset \mathcal{R}(\Lambda^c_{\alpha+\epsilon}-n)$.
	
	Suppose there exists a function $f\in \mathcal{F}_\infty$ such that 
	\begin{equation*}
	\sup_{B\in \calA_{\Lambda^c_{\alpha+\epsilon}-n}} \frac{ \| AB-BA \| }{\| B \| } \leq f_\epsilon(n) = O(n^{-\infty})
	\quad \mbox{ for all } \quad 0 < \epsilon < \pi - \alpha.
	\end{equation*}
	Let $ B \in \mathcal{R}(\Lambda^c_{\alpha+\epsilon}-n)$ be such that $\|B\| =1$.
	By Kaplansky's density theorem the unit ball of $\calA_{\Lambda^c_{\alpha+\epsilon}-n}$ is dense (in the weak operator topology) in the unit ball of $\mathcal{R}(\Lambda^c_{\alpha+\epsilon}-n)$.
	Hence there is a net $B_\lambda \in \calA_{\Lambda^c_{\alpha+\varepsilon}-n}$ such that $\|B_\lambda\| \leq 1$ and $ \wlim B_\lambda = B$.
	By weak operator continuity of $B \mapsto A B$ we have that $AB - BA = \wlim_\lambda A B_\lambda - B_\lambda A$.
	Let $\xi \in \calH$.
	Then by weak operator continuity and the Cauchy-Schwarz inequality we have
	\begin{align*}
	\| [A,B] \xi \|^2 &= \lim_\lambda \langle [A,B] \xi, [A, B_\lambda] \xi \rangle  \\
	&\leq \| [A,B] \xi\| \liminf_\lambda \| [A, B_\lambda] \xi \| \\ 
	&\leq \| [A,B] \xi \| \| \xi \| \liminf_{\lambda} \| [A, B_\lambda] \|.
	\end{align*}
	It follows that
	\begin{align*}
	\| A B - BA \| &\leq \liminf_\lambda \| AB_\lambda - B_\lambda A\| \\
	& \leq  f_\epsilon(n) \liminf_\lambda \| B_\lambda\| \\
	& = \| B \| f_\epsilon(n),
	\end{align*}
	completing the proof.
\end{proof}

As a consequence, we get the following locality property for intertwiners. 
\begin{cor}\label{cor:intertwineraloc}
	Let $\rho$ and $\sigma$ be almost localized in a cone $\Lambda_\alpha$
	with decay functions $f$ and $g$, respectively.
	If $R \in (\rho, \sigma)$ then $R$ is almost localized in $\Lambda_{\alpha}$ with decay function $\|R\|(f + g)$.
\end{cor}

\begin{proof}
	For all $B \in \calA$ we have that
	\begin{align*}
	\| R B - B R \| 
	& \leq  \| R B- R\rho(B)  \| + \| R\rho(B) - \sigma(B)R\| + \|  \sigma(B) R - B R\| \\
	&\leq  \| R\| ( \| \rho(B) - B \| + \|\sigma(B) -B \| ).
	\end{align*}
	Therefore,
	\begin{equation*} 
	\sup_{B\in \calA_{\Lambda^c_{\alpha+\epsilon}-n}} \frac{ \| RB-BR\| }{\| B \| }  \leq \|R\|( f_\epsilon(n) + g_\epsilon(n)).
	\end{equation*}
\end{proof}

The space of intertwiners $ (\rho,\sigma)$ naturally resides in the von Neumann algebra generated by $\calA$,
which by purity of the state is the entire $\calB(\calH_0)$. 
The almost local property for operators $A \in \calB(\calH_0)$ suggest that 
if $\omega_0$ satisfies certain locality conditions then $A$ may be well approximated in some cone algebra $\mathcal{R}(\Lambda)$.
Following \cite{Naaijkens11, BuchholzF}, we introduce a cone  $K$ with an arbitrary but small opening angle as a forbidden region
and define an \emph{auxiliary algebra}, 
\begin{equation*}
\calB_K\equiv  \overline{ \bigcup_{x\in \ZZ^\nu} \mathcal{R}((K+x)^c)}^{\| \cdot\|} =  \overline{ \bigcup_{n \in \NN } \mathcal{R}((K+n)^c)}^{\| \cdot\|},
\end{equation*}
where the second equality follows from the fact that for every $ x\in \ZZ^\nu$ there is an $n \in \NN$ such that $ K+n \subset K+x$
and $ (K+x)^c \subset (K+n)^c$.

\begin{lemma}\label{lem:alocauxalg}
	Suppose $\omega_0$ satisfies Assumption \ref{ass:hdsp} and
	let $\Lambda_\alpha$ be a cone such that   $\Lambda_{\alpha+\epsilon} \ll (K+x)^c$ for some $\epsilon>0$ and $ x\in \ZZ^\nu$.
	If $A$ is almost localized in $\Lambda_\alpha$  then $A \in \calB_{K}$.
\end{lemma}

\begin{proof}
	Let $\delta >0$ be given and $A \in \calB(\calH)$ be almost localized in $\Lambda_\alpha$.
	Then, by the almost localized condition, there exists an $N>0$ such that if $n>N$ then 
	\begin{equation}\label{eqn:alo}
	\|A B - B A\|  < \delta \|B\| \|A\| \quad \text{for all} \quad B \in \mathcal{R}(K+x+n).
	\end{equation}
	
	Let $x'$ be such that $ \Lambda_{\alpha+\epsilon} \ll (K+x+n)^c \ll (K+x'+n)^c$.
	By Haag duality and the approximate split property, there is a type I factor $\mathcal{N}$ such that 
	$\calB(\calH) \cong \mathcal{N} \otimes \mathcal{N}'$ and 
	\begin{align}
	\mathcal{R}((K+x+n)^c) &\subset \mathcal{N} \subset \mathcal{R}((K+x'+n)^c)\\
	\mathcal{R}((K+x'+n)) &\subset \mathcal{N}' \subset \mathcal{R}(K+x+n). \label{eqn:split2}
	\end{align}
	From \eqref{eqn:alo} and \eqref{eqn:split2} it follows that  
	\begin{equation*}
	\| A B - B A \| \leq \delta \|A\| \|B\| \quad \text{for all} \quad B \in \mathcal{N}'.
	\end{equation*}
	
	Applying Lemma 2.1 of \cite{NachSW}, there exists an operator $A' \in \mathcal{N}$ such that 
	\begin{equation*}
	\| A - A' \| \leq \delta \| A\|.
	\end{equation*}
	In other words, $A$ is arbitrarily well approximated  in norm by an operator $A'$ in the cone algebra $\mathcal{R}((K+y)^c)$ for some $y \in \ZZ^\nu$.
	Therefore, $A \in \calB_K$.
\end{proof}

\begin{cor}\label{cor:interwinerauxalg}
	Suppose $\omega_0$ satisfies Assumption \ref{ass:hdsp} and
	let $\Lambda_\alpha$ be a cone such that   $\Lambda_{\alpha+\epsilon} \ll (K+x)^c$ for some $\epsilon>0$ and $ x\in \ZZ^\nu$.
	For $\rho$ and $\sigma$  almost localized in a cone $\Lambda_\alpha$, if  $ R \in (\rho, \sigma)$  then $R \in \calB_K$.
\end{cor}

\begin{proof}
	By Corollary \ref{cor:intertwineraloc}, $R$ is almost localized in $\Lambda_{ \alpha  }$.
	The result is obtained by applying Lemma \ref{lem:alocauxalg}.
\end{proof}

\subsection{Superselection structure}

We assume throughout this section that the vacuum state $\omega_0$ satisfies 
Assumption \ref{ass:hdsp}.

Let $\Delta$ be a semi-group of $*$-endomorphisms on $\calA$
and $\mathcal{T}_\Delta$ denote the set of all interwiners $(\rho,\sigma)$ such that $\rho,\sigma \in \Delta$.
We form a $C^*$-category that we again call $\Delta$ whose objects are $\rho \in \Delta$ and arrows are $\mathcal{T}_\Delta$.
We introduce a \emph{superselection criterion} on $\Delta$ that requires objects be almost localized and transportable 
with respect to $\omega_0$.

\begin{defn}\label{def:superselection}(Superselection criterion)
	$\Delta$ is said to satisfy the almost localized and transportable superselection criterion for $\omega_0$ if
	\begin{enumerate}
		\item for all $\rho \in \Delta$ there exists a cone $\Lambda_{ \alpha  } \in \mathcal{C}$ such that $\rho$ is almost localized in $\Lambda_{ \alpha  }$, and
		\item  $\rho$ is transportable with respect to the state $\omega_0$.
	\end{enumerate}
\end{defn}

%\begin{cor}
%The set of intertwiners is contained in the auxiliary algebra, $\mathcal{T}_\Delta \subset \calB_K$.
%\end{cor}
%
%\begin{proof}
%This follows directly from the construction of $\Delta$ and Corollary \ref{cor:interwinerauxalg}.
%\end{proof}

The sequence $U_n$ of unitary operators implementing the asymptotically inner property as in Lemma \ref{lem:asympinner} is not unique.
Motivated by the following formal calculation,
\begin{align*}
\| R \rho(R')  -  \sigma(R') R  \| & = \lim_{m,n \ra \infty} \| R U_m^* R' U_m - V_n^* R' V_n R \| = \lim_{m,n \ra \infty}  \|V_n R U_m^* R' -  R' V_n RU_m^* \|
\end{align*}
\cite{BuchholzAA} introduced the notions of \emph{asymptopia} and {bi-asymptopia}, to be defined later.
The calculation above is formal in the sense that $\rho(R')$ may not be well defined. 
Although, the notion of asymptopia was originally intended to side step the existence of Haag duality in the vacuum state, 
our proofs will always assume Haag duality.
We consider maps $\mathcal{U}: \Delta \ra \mathcal{U}_\rho$ where each $\mathcal{U}_\rho$ is a family of unitary sequences
implementing the asymptotically inner property for $\rho$.

\begin{defn}
	A family of sequences of operators $\mathcal{S}\ni \{ A_n\}_{n=1}^\infty$ in $\calB(\calH)$  is called \emph{stable} if it is closed under taking subsequences.
\end{defn}

\begin{defn}[\cite{BuchholzAA}]\label{defn:asymptopia}
	An \emph{asymptopia} for $\Delta$ is a mapping $\mathcal{U} : \rho \mapsto \mathcal{U}_\rho$ 
	where each $\mathcal{U}_\rho$ is a stable family of sequences of unitary operators in $\calB(\calH_0)$
	such that for each $\{U_m\} \in \mathcal{U}_\rho$ we have
	\begin{equation*}
	\rho(A) = \lim_{n \ra \infty} U_m^* A U_m \quad \mbox{ for all } \quad A \in \calA,
	\end{equation*}
	and given $R \in (\rho, \rho')$ and $R' \in (\sigma,\sigma')$, and $ \{U_m\} \in \mathcal{U}_\rho$ and $ \{V_n\} \in \mathcal{U}_{\rho'}$,
	\begin{equation*}
	\lim_{m,n\ra \infty} \| [ V_n R U_m^*, R'] \| = 0.
	\end{equation*}
\end{defn}

Let $\mathcal{A}_\Delta$ be the $C^*$-subalgebra of $\calB(\calH_0)$ generated by $\calA$ and $\mathcal{T}_\Delta$.
\begin{thm}\label{thm:asymptopia}
	Let $\Delta$ be a $C^*$-category satisfying the almost localized and transportable superselection criterion \ref{def:superselection}.
	Suppose that there exists a cone $K \in \mathcal{C}$ such that for each $\rho \in \Delta$ almost localized in $\Lambda_{ \alpha }$
	there exist $\epsilon >0$ and $x \in \ZZ^\nu$ such that $\Lambda_{ \alpha  +\epsilon} \ll (K+x)^c$.
	Then, there exists an asymptopia $\mathcal{K}: \rho \mapsto \mathcal{K}_\rho$ for $\Delta$.
	Furthermore, each $\rho \in \Delta$ has a unique extension to a $*$-endomorphism $\rho_{\mathcal{K}}$ to $\calA_{ \Delta}$ such that
	for each $\{ U_m\} \in \mathcal{K}_\rho$ we have 
	\begin{equation*}
	\rho_{\mathcal{K}}(A) = \lim_{m \ra \infty} U_m^* A U_m \quad \mbox{ for all } \quad A \in \calA_\Delta,
	\end{equation*}
	and $(\rho, \sigma) = (\rho_{\mathcal{K}}, \sigma_{\mathcal{K}} )$ for all $\rho, \sigma \in \Delta$.
	$\Delta$ is a tensor $C^*$-category.
\end{thm}

\begin{proof}
	We construct an asymptopia $\mathcal{K}$ for $\Delta$ based on the forbidden region $K$.
	Let $\rho \in \Delta$ be given.
	Define $\mathcal{K}_\rho$ as the set of unitary sequences $\{ U_m\}$ with the following properties,
	\begin{enumerate}
		\item $\rho(A) = \lim_{m \ra \infty } U_m^* A U_m$ for all $A \in \calA$,
		\item there is an increasing sequence $k_m \in \NN$ and a sequence of $*$-endomorphisms $\rho_m$ 
		such that $\rho_m$ is almost localized in $K + k_m$ with decay function $g$, chosen uniformly in $m$.
		\item $U_m$ intertwines $\rho$ and $\rho_{m}$, that is, $U_m \in (\rho, \rho_m)$.
	\end{enumerate} 
	A similar argument as the proof in Lemma \ref{lem:asympinner} shows that $\mathcal{K}_\rho$ is non-empty.
	The sequence $\rho_m$ will converge pointwise to the identity,
	\begin{equation}\label{eqn:pointwiseidentity}
	\lim_{ m\ra \infty} \| \rho_m(A) - A \| = \lim_{m \ra \infty } \| \rho(A) -U_m^* A U_m \| = 0.
	\end{equation}
	By construction, the family $\mathcal{K}_\rho$ is stable, that is, any subsequence of $\{ U_m\}$ will also belong to $\mathcal{K}_\rho$.
	
	Let $\epsilon >0$ and $\rho, \rho', \sigma, \sigma' \in \Delta$ be given.
	Consider  $R \in (\rho, \rho')$ and $S \in (\sigma, \sigma')$, and $\{ U_m\} \in \mathcal{K}_\rho$ and $ \{V_{m'}\} \in \mathcal{K}_{\rho'}$.
	It follows that  	$V_{m'}R U_m^* \in ( \rho_m , \rho'_{m'})$ is an intertwiner, and thus
	by Corollary \ref{cor:intertwineraloc},  $V_{m'}R U_m^* $ is almost localized in $K + \min\{ k_m,k_m'\}$ with decay function $2\|R\| f$.
	By Corollary \ref{cor:interwinerauxalg}, $S \in \mathcal{B}_K$. 
	Therefore,  there exists $M>0$ such that 
	$\| S - S_{M} \| < \epsilon \| S\|$ for some $S_{M} \in \mathcal{R}((K + M)^c)$.
	It follows that 
	\begin{align*}
	\| [ V_{m'} R U_m^*, S] \| & \leq   \| [ V_{m'} R U_m^*, S_M] \| + \| R\| \| S- S_M\|\\
	&\leq  2 \|R \| \| S\|  f( \min\{k_m,k_m'\} - M) + \epsilon \| R \| \|S\|.
	\end{align*}
	Taking the limit as $m,m' \ra \infty$ and since $\epsilon >0$ was arbitrary we have that 
	\begin{equation*}
	\lim_{m,m' \ra \infty}  \| [ V_{m'} R U_m^*, S] \| = 0.
	\end{equation*}
	Therefore, $\mathcal{K}$ is an asymptopia for $\Delta.$
	Applying Theorem 4 of \cite{BuchholzAA} we recover that $\rho$ extends uniquely to a $*$-endomorphism $\rho_{\mathcal{K}}$ on  $\calA_{\Delta}$ 
	and satisfies the properties as listed in the theorem.
	
	The tensor product for objects is composition of maps, 
	\begin{equation*} 
	\rho \otimes \sigma \equiv \rho \circ\sigma,
	\end{equation*}
	and the tensor product for arrows is defined for $R \in (\rho, \rho')$ and $S \in (\sigma, \sigma')$ as 
	\begin{equation*} 
	R \times S \equiv R \rho_{\mathcal{K}}(S)  .
	\end{equation*}
	A calculation shows that for all $ A \in \calA$,
	\begin{align*}
	(R \times S )\rho \circ \sigma (A) & = R \rho_{\mathcal{K}}(S) \rho ( \sigma(A))  
	= R \rho_{\mathcal{K}}(S \sigma(A) ) 
	= \rho'(\sigma' (A)) R \rho_{\mathcal{K}}(S) \\
	&= \rho' \circ \sigma (A) (R \times S),
	\end{align*}
	and thus $R \times S \in ( \rho\circ \sigma, \rho' \circ \sigma')$.
\end{proof}

If every object $\rho\in \Delta$ can be almost localized in a fixed convex cone $\Lambda_\alpha$ with $\alpha < \pi/2$
then the forbidden region can be chosen as $ \Lambda_{ \alpha  }^c -1$.

To define a braiding on $\Delta$, we must to consider two objects simultaneously.  
Intuitively, if $\rho$ and $\sigma$ are almost localized in $\Lambda_{ \alpha  }$ 
a braiding could be performed by first transporting 
$\rho$ to $\rho'$ almost localized in the cone $\Lambda_{ \alpha  }^\mathcal{U}$ 
and $\sigma$ to $\sigma'$ almost localized in the cone $\Lambda_{ \beta}^\mathcal{V}$, see Figure \ref{fig:cones}.
By the almost localized property, it would follows that $ \rho'\circ \sigma' \approx \sigma' \circ \rho'$. 
Transporting the $\sigma'$ and $\rho'$ back would complete the braid.
To make this precise, we again follow the notations introduced in \cite{BuchholzAA}.

\begin{defn}[\cite{BuchholzAA}]\label{defn:biasymptopia}
	Suppose  $\Delta$ is a tensor $C^*$-category of endomorphisms on $\calA$.
	A \emph{bi-asymptopia} for $\Delta$ is a pair of mappings  
	$\mathcal{U} : \rho \mapsto \mathcal{U}_\rho$ and $ \mathcal{V}: \rho \mapsto \mathcal{V}_\rho$ 
	where $\mathcal{U}_\rho$ and $\mathcal{V}_\rho$ are stable families of unitary sequences
	such that for each $ \{ U_m\} \in \mathcal{U}_\rho$ and $\{ V_n\} \in \mathcal{V}_\rho$ we have 
	\begin{equation*}
	\rho(A) = \lim_{m \ra \infty} U_m^* A U_m = \lim_{n \ra\infty} V_n^* A V_n \quad \mbox{ for all } \quad A \in \calA,
	\end{equation*}
	and if given $ \rho, \rho' \in \Delta$ we can find $\{U_m\} \in \mathcal{U}_{\rho}$, $ \{U'_m\} \in \mathcal{U}_{\rho'}$, $\{V_n\} \in \mathcal{V}_{\rho}$, and $ \{V'_n\} \in \mathcal{V}_{\rho'}$ such that $\{U_m \times U'_m\} \in \mathcal{U}_{\rho \circ \rho'}$ and $ \{V_n \times V'_n\} \in \mathcal{V}_{\rho\circ \rho'}$, and
	if given intertwiners  $R \in (\rho,\rho')$ and $ S \in ( \sigma, \sigma')$ then
	\begin{equation*}
	\lim_{m,m',n,n' \ra \infty} \| U'_{m'} R U_{m}^* \times V'_{n'} S V_{n}^* - V'_{n'} S V_{n}^* \times U'_{m'} R U_{m}^* \| =  0
	\end{equation*}
	for all $U_m \in \mathcal{U}_\rho$, $U_{m'} \in \mathcal{U}_{\rho'}$, $ V_m \in \mathcal{V}_\sigma$, and $ V_{m'} \in \mathcal{V}_{\sigma'}$.
\end{defn}

\begin{figure}
	\includegraphics[scale=.7]{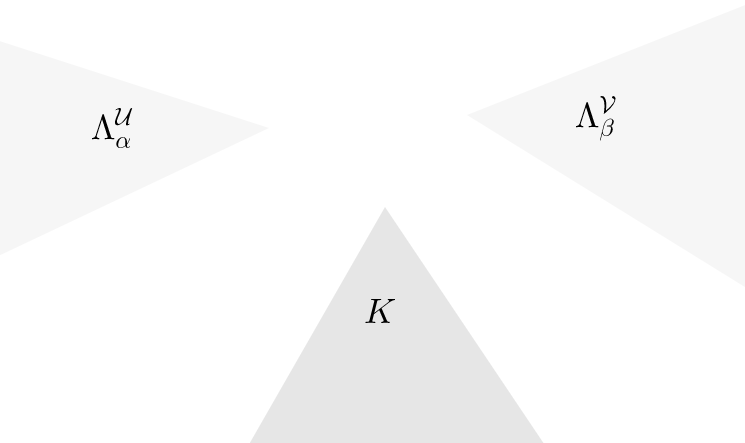}
	\caption{A sample configuration of cones $\Lambda_{ \alpha}^{\mathcal{U}}$, 
		$\Lambda_\beta^{\mathcal{V}}$ and $K$ satisfying the conditions set for the bi-asymptopia $\{\mathcal{U}, \mathcal{V}\}$ of $\Delta.$}
	\label{fig:cones}
\end{figure}

\begin{thm}\label{thm:biasymptopia}
	Suppose  $\Delta$  satisfies the assumptions of Theorem \ref{thm:asymptopia}.
	Then, there exists a bi-asymptopia $\{ \mathcal{U}, \mathcal{V}\}$ for $\Delta$
	and given $\rho, \sigma \in \Delta$ the following limit exists
	\begin{equation}\label{eqn:braid}
	\epsilon_{\rho,\sigma} \equiv \lim_{m, n \ra \infty } ( V_{n}^* \times U_{m}^*) (  U_{m} \times  V_n ),
	\end{equation}
	is independent of choice of $ U_m \in \mathcal{U}_\rho$ and $ V_n \in \mathcal{V}_\sigma$, and $ \epsilon_{\rho,\sigma} \in ( \rho \circ \sigma, \sigma \circ \rho).$
	
	Furthermore, if $ R \in (\rho, \rho')$ and $ S \in (\sigma, \sigma')$ then 
	\begin{equation}\label{eqn:braid1}
	\epsilon_{\rho',\sigma'} R\times S = S \times R \  \epsilon_{\rho, \sigma}
	\end{equation}
	and if $\tau \in \Delta$  then 
	\begin{align}
	\epsilon_{ \rho \circ \sigma, \tau}  &= ( \epsilon_{ \rho, \tau} \times 1_\sigma )( 1_\rho \times \epsilon_{\sigma, \tau}) \label{eqn:braid2}\\
	\epsilon_{ \rho, \sigma \circ \tau} &= ( 1_\sigma \times \epsilon_{ \rho, \tau}) ( \epsilon_{ \rho, \sigma} \times 1_\tau) \label{eqn:braid3}.
	\end{align}
	That is, $\Delta$ is a braided tensor $C^*$-category.
\end{thm}

\begin{proof}
	Let $\Lambda_\alpha^{\mathcal{U}}$ and $\Lambda_\beta^{\mathcal{V}}$ be cones in $\mathcal{C}$ 
	such that for some $\epsilon >0$ and $x \in \ZZ^\nu$ we have
	$ \Lambda_{\alpha+\epsilon}^{ \mathcal{U}}\ll  (K+x)^c$, 
	$ \Lambda_{\beta+\epsilon}^{ \mathcal{V}}\ll  (K+x)^c$
	and $\Lambda_{ \alpha + \epsilon}^{\mathcal{U}} \ll (\Lambda_{\beta}^{ \mathcal{V}})^c$.
	For instance, see Figure \ref{fig:cones} for a sample configuration on $\ZZ^2$
	
	Define the family of unitary sequences $\mathcal{U}_\rho$  and $\mathcal{V}_\rho$ 
	in the same way $\mathcal{K}_\rho$ was defined in Theorem \ref{thm:asymptopia}, 
	except in the construction replace the role of the cone $K$ for the cones $  \Lambda_{\alpha}^{ \mathcal{U}}$ and 
	$ \Lambda_{\beta}^{ \mathcal{V}}$, respectively.
	If $ \{ U_m\} \in \mathcal{U}_\rho $ then  
	$\rho(A) = \lim_{m \ra \infty} U_m^* A U_m$ for all $ A \in \calA$
	and  there exists an increasing sequence $k_m \in \NN$ and a sequence of $*$-endomorphisms $\rho_m$  such that 
	$\rho_m$ is almost localized in $\Lambda^{\mathcal{U}}_\alpha+k_m$  with decay function $f$, chosen independent of $m$, 
	and $U_m \in (\rho, \rho_m)$.
	By construction $\mathcal{K}_\rho$ is stable, that is, any subsequence of $\{ U_m\}$ will also belong to $\mathcal{K}_\rho$.
	Lemma \ref{lem:asympinner}, shows that $\mathcal{K}_\rho$ is non-empty
	and the sequence $ \rho_m$ will converge pointwise to the identity, see \eqref{eqn:pointwiseidentity}.
	We define $ \mathcal{V}_\rho$ similarly where  $ \Lambda_\beta^\mathcal{V}$ will play the analogous role for $\Lambda_\alpha^\mathcal{U}.$
	
	Let $\rho,\rho' \in \Delta$ and $\{ U_m\} \in \mathcal{U}_\rho$ and $ \{U_m'\} \in \mathcal{U}_{\rho'}$.
	For all $A\in \calA$ we have
	\begin{align*}
	\lim_{m \ra \infty} (U_m \rho( U_m') )^* A U_m \rho( U_m') & = 	\lim_{m \ra \infty} U_m^* \rho_m( U_m'^*) A  \rho_m( U_m') U_m = \rho \circ \rho' (A)
	\end{align*}
	where for the second equality we used \eqref{eqn:pointwiseidentity}.
	Therefore, $ \{ U_m \times U_m'\} \in \mathcal{U}_{\rho \circ \rho'}$.
	
	Let $\epsilon>0$ be given.
	For $\rho, \rho', \sigma, \sigma' \in \Delta$ and $ \{U_m\}\in \mathcal{U}_\rho$, $\{U_{m'}'\} \in \mathcal{U}_{\rho'}$, $  \{ V_n\} \in \mathcal{V}_\sigma$ and $\{ V_{n'}'\} \in \mathcal{V}_{\sigma'}$ 
	we have that 
	$ U'_{m'} R U_{m}^* \in ( \rho_m, \rho'_{m'} )$ and $ V'_{n'} S V_{n}^* \in (\sigma_n, \sigma_{n'}')$.
	By Corollary \ref{cor:intertwineraloc} and the construction of $\mathcal{U}$ and $ \mathcal{V}$,
	the intertwiner $U'_{m'} R U_{m}^*$ is almost localized in $\Lambda_{ \alpha} + \min\{k_m,k_{m'}\}$ with decay function $2 \| R\| f$
	and $ V'_{n'} S V_{n}^* $ is almost localized in $\Lambda'_{\beta} + \min\{k_n, k_{n'} \}$ with decay function $2 \| S\| g$.
	Thus, there exists $M, N>0$ such that for all $m,m',n,n'$ sufficiently large we have that 
	$\| U'_{m'} R U_{m}^* - R_M\| < \epsilon \|R\|$ and $\| V'_{n'} S V_{n}^* - S_N\| < \epsilon \|S\|$ 
	for some $ R_M \in \mathcal{R}( (\Lambda_{ \alpha }^{\mathcal{U}} + M)^c$ and $S_N \in \mathcal{R}( (\Lambda_{ \beta }^{\mathcal{V}} + N)^c)$.
	It follows that
	\begin{align*}
	\| U'_{m'} R U_{m}^* \times V'_{n'} S V_{n}^* - V'_{n'} S V_{n}^* \times U'_{m'} R U_{m}^* \| 
	&  = \| U'_{m'} R U_{m}^* \rho_m( V'_{n'} S V_{n}^* )- V'_{n'} S V_{n}^* \sigma_n(  U'_{m'} R U_{m}^*) \|  \\
	&  \leq  3 \| R\| \| V'_{n'} S V_{n}^*  -  S_N\| 
	+ 3\| S \| \|  U'_{m'} R U_{m}^* - R_M \| \\
	&  \qquad +\| R\| \| \rho_m( S_N ) - S_N\| 
	+ \| S \| \| \sigma_n( R_M) -  R_M\|  + \|R_M S_N - S_NR_M \|\\
	&  < 6 \epsilon \| R\| \|S\| + 2 \|R\| \|S\|(  f(N + m) + g(M +n)) . 
	\end{align*}
	Taking the limit as $m,m',n,n' \ra \infty$ and since $\epsilon>0$ was arbitrary, 
	gives that $\mathcal{U}$ and $ \mathcal{V}$ form a bi-asymptopia for $\Delta.$
	The result follows from Theorem 8 of \cite{BuchholzAA}.
	
	Here we include a calculation showing that $\epsilon_{ \rho,\sigma}$ intertwines $ \rho \otimes \sigma$ with $\sigma \otimes \rho$.
	A crucial point is that  $\rho_m $ and $\sigma_n$, being almost localized in far removed disjoint cones, commute asymptotically:
	\begin{align*}
	\epsilon_{ \rho, \sigma} \rho \otimes \sigma (A) & = \lim_{m,n \ra \infty} ( V_{n}^* \times U_{m}^*) (  U_{m} \times  V_n ) \rho (\sigma(A))
	& = \lim_{m,n \ra \infty} V_n^* \sigma_n(U_m^*) U_m \rho( V_n) \rho(\sigma(A))\\
	& = \lim_{m,n \ra \infty} V_n^* \sigma_n(U_m^*) U_m \rho( \sigma_n(A)V_n)
	& = \lim_{m,n \ra \infty} V_n^* \sigma_n(U_m^*)  \rho_m( \sigma_n(A)V_n) U_m\\
	& = \lim_{m,n \ra \infty} V_n^* \sigma_n(U_m^*)  \sigma_n( \rho_m(A))\rho_m(V_n) U_m
	& = \lim_{m,n \ra \infty} V_n^* \sigma_n(U_m^* \rho_m(A))\rho_m(V_n) U_m\\
	& = \lim_{m,n \ra \infty} V_n^* \sigma_n( \rho(A) U_m^*)\rho_m(V_n) U_m 
	& = \lim_{m,n \ra \infty}  \sigma(\rho(A))V_n^* \sigma_n(U_m^*) U_m \rho( V_n)\\
	& = \lim_{m,n \ra \infty} \sigma (\rho(A)) ( V_{n}^* \times U_{m}^*) (  U_{m} \times  V_n )\\
	& = \sigma\otimes \rho(A) \epsilon_{ \rho,\sigma}
	\end{align*}
	Similar arguments hold to show properties \eqref{eqn:braid1}, \eqref{eqn:braid2}, and \eqref{eqn:braid3}.
\end{proof}

\section{Stability of the superselection structure}
Let $\omega_0$ and $ \omega_1$ be gapped ground states for a quantum spin system with quasi-local algebra $\calA$.
Then, it is commonly said that $\omega_0$ and $\omega_1$ are in the same \emph{gapped ground state phase}
if there exists a continuous family of Hamiltonians $H(s)$ with finite range interactions such that 
$H(s)$ has a non-vanishing spectral gap above the ground state uniform in $s$ and 
$ \omega_0$ and $ \omega_1$ are ground states of $H(0)$ and $H(1)$ \cite{ChenGW,BachmannMNS}.
A current challenge in mathematical physics is the classification of gapped ground state
phases \cite{BachmannMNS,BachmannO,Ogata1,Ogata2}.
One approach to classifying a phase is to construct a complete set of invariants.
By definition, an invariant is a quantity that is constant within a phase.
Consequently, if an invariant is computed for two systems and is found to take different values,
the systems must be in different phases. 
The construction of invariants can be expected to rely on the existence of a spectral gap. 

The superselection structure of quasi-particle excitations is expected to be an invariant in a gapped ground state phase \cite{KitaevQD, BravyiHM, Haah}.
In particular, the stability of anyons is expected to play a crucial role in the classification of two dimensional topological phases \cite{KitaevHC}.
In the previous section, we showed that starting from a vacuum state satisfying Assumption \ref{ass:hdsp},
the statistics for a family of almost localized and transportable endomorphisms $\Delta$ is described by a braided tensor $C^*$-category.
In this section, we show that a the superselection structure of $\Delta$ is stable against any deformation of $\Delta$ corresponding to a quasi-local dynamics.
In the next section, we will apply our results to study certain gapped phases where there exists a quasi-local dynamics connecting the relevant states in each phases.

%Let  $\Phi$ have a finite $F$-norm and $A \in \calA_{ \{x\}}$.
%For an increasing and exhausting sequence $\Lambda_n$ in $\Gamma$ such that $ x \in \Lambda_1$ we have 
%that for $m<n$ 
%\begin{align}
%\| [H_{\Lambda_n} - H_{\Lambda_m}, A] \| &\leq 2\|A\| \sum_{y\in \Lambda_n\setminus\Lambda_m} \sum_{\substack{X \subset \Gamma \\ x,y\in X}} \|\Phi(X)\| \\
%&\leq 2 \|A\| \| \Phi\|_F \sum_{y \in \Lambda_n\setminus\Lambda_m} F(d(x,y)),
%\end{align}
%where the second inequality comes from the assumption that $\Phi$ has a finite $F$-norm.
%The right hand side can be made arbitrarily small due to the uniform integrability condition \eqref{eqn:unifint}.
%By a similar argument,  the norm limit exists for all $A\in \calA_{loc}$
%\begin{equation}
%\delta \equiv\lim_{n \ra \infty} [ H_{\Lambda_n}, A] \quad \text{ for all } \quad A\in\calA_{loc}.
%\end{equation}
%Therefore, $\delta$ is a norm densely defined derivation and is norm-closable 
%with $\calA_{loc}$ as a core (\cite{BratteliR2}, Proposition 6.2.3).
%It can be checked that
%\begin{align}
%\frac{d}{dt} \tau_t(A)  = i \tau_t(\delta(A) ) \quad \text{ for all} \quad A\in\calA_{loc}.
%\end{align}
%and $\delta$ is the generator of the dynamics $\tau_t = e^{it \delta}$.

\subsection{Lieb-Robinson bound for cones}

The main technical tool in our proof will be the Lieb-Robinson bounds, discussed in Section \ref{sec:dynamics}.
In this section, we freely use the language introduced in Section \ref{sec:dynamics}.

Recall that in the definition of almost localized endomorphisms, we consider observables supported on two cone regions.
However, if $X, Y \subset \Gamma$ are infinite regions 
then the Lieb-Robinson bound \eqref{eqn:LRbound}  may not be better than the trivial bound $2\|A\| \|B\|$.
For $X\in \mathcal{C}$  and  $Y\subset X^c$ two cone regions,
we show that the Lieb-Robinson bound recovers a good approximation when the distance between the cones is large and the difference of opening angles is strictly positive.
Results of this type were first discussed and proved in the thesis of Schmitz \cite{Schmitz}. 
For completeness, and because access to \cite{Schmitz} is not readily available, we present the results here.

Let $F$ be an $\mathcal{F}$-function, as defined in \ref{def:Ffunc}.
\begin{assumption}\label{asp1}
	Let $g:\RR^{\geq 0} \ra \RR^{\geq 0}$ be uniformly continuous, non-decreasing and sub-additive.
	We assume there exists  $t_0$ such that for all $t > t_0$, $b>0$ and $k \in \NN$ 
	\begin{equation*}
	\int_{t}^\infty r^k e^{-b g(r)} dr \leq K(b,k)t^{l(k)} e^{-b g(t)}.
	\end{equation*}
	for some positive function $ K(b,k)>0$ and affine function $l(k)$.
\end{assumption}
By uniform continuity, if $g$ satisfies Assumption \ref{asp1} then $\lim_{r\ra \infty} r^k e^{ - b g(r)}  = 0$ for all $k\in \NN$. 
From the following 	inequalities 
\begin{align*}
\int_{t}^\infty r^k e^{-b r} dr &\leq \frac{k+1}{b} t^{k} e^{- b t} \quad \mbox{ for } t>0\\
\int_{t}^\infty r^k e^{-b \frac{r}{\ln^2 r}} dr &\leq \frac{2k+3}{b} t^{2k+2} e^{- b \frac{r}{\ln^2 r}} \quad \mbox{ for } t>(e/2)^2.
\end{align*}
we see that the functions
\begin{align*}
g(r) = r \qquad\mbox{and}\qquad g(r) = \frac{r}{\ln^2r}
\end{align*}
satisfy Assumption \ref{asp1}.

\begin{lemma}( Satz II.8, \cite{Schmitz})\label{lem:doubImpInt}
	Suppose $g$ satisfies Assumption \ref{asp1} and let $F$ be an $\mathcal{F}$-function.
	Define the following sets as $X = \Lambda_\alpha \in \mathcal{C}$ and $ Y_{\epsilon,n} = \left( \Lambda_{\alpha+ \epsilon} - n\right)^c$.
	Then, there exists an affine function $\tilde{l}$ such that 
	for all $0\leq \alpha < \pi $, $ 0< \epsilon< \pi - \alpha$ and $b>0$ 
	\begin{equation}\label{eqn:LRdoubsum}
	\sum_{x\in X}\sum_{ y\in Y_{\epsilon,n}} F_{bg} (d(x,y)) \leq C_\epsilon d(X,Y_{\epsilon,n})^{\tilde{l}(\nu)} e^{ - b g(d(X,Y_{\epsilon,n})\sin\epsilon )}
	\end{equation}
	where 
	\begin{equation*}
	n \sin(\alpha + \epsilon) \leq d(X,Y_{\epsilon,n}) \leq n \sin(\alpha + \epsilon)+2
	\end{equation*}
	and  $C_\epsilon$ is non-increasing in $\epsilon$ and only depends on $\nu$, $b$ and $\alpha$.
\end{lemma}

\begin{figure}
	\includegraphics[width=.5\textwidth]{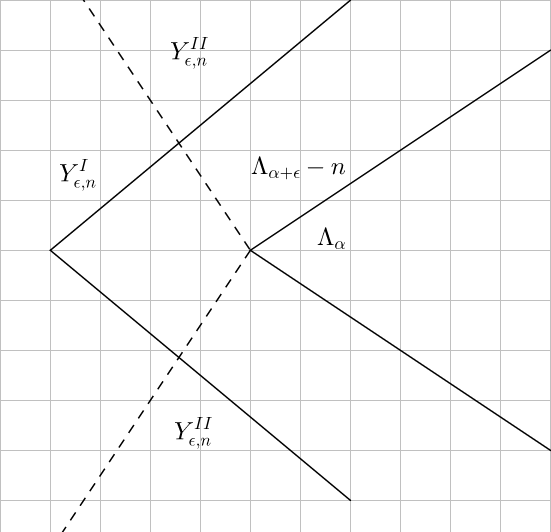}
	\caption{ The dashed lines are perpendicular to the surface of $\Lambda_{ \alpha  }$. 
		The region $\Lambda_{ \alpha  +\epsilon}^c - n$ is composed of two disjoint regions $Y_{\epsilon,n}^I$ and $Y_{\epsilon,n}^{II}$.}
	\label{fig:LRcones}
\end{figure}

\begin{proof}
	Without loss of generality, suppose $\Lambda_{ \alpha  }$ is a cone based at the origin.
	Suppose that  $0\leq \alpha < \pi/2$, that is, $\Lambda_{ \alpha  }$ is a convex cone.
	Let $\Lambda^*_\alpha = - \Lambda_{\pi/2	 - \alpha}$ be the polar cone for $\Lambda_\alpha$.
	Let $Y^I_{\epsilon,n} = Y_{\epsilon,n} \cap \Lambda^*_\alpha$ and $Y^{II}_{\epsilon,n} = Y_{\epsilon,n} \setminus Y_{\epsilon,n}^I$, see Figure \ref{fig:LRcones}.
	It follows that
	\begin{equation*}
	\sum_{x\in X}\sum_{ y\in Y_{\epsilon, n}} F_{bg}(d(x,y)) =   \sum_{y \in Y_{\epsilon,n}^I} \sum_{x\in X}F_{bg} (d(x,y)) + \sum_{y \in Y_{\epsilon,n}^{II}}  \sum_{x\in X}F_{bg} (d(x,y)).
	\end{equation*}
	
	We proceed by bounding the first sum.
	By geometry, if $y \in C^*_\alpha$ and $ x \in C_\alpha$ then $ d(y,0) \leq d(x,y)$.
	It follows that
	\begin{align*}
	\sum_{y \in Y^I_{\epsilon,n}} \sum_{x\in X}  F_{bg} (d(x,y)) 
	& \leq  \sum_{ y \in Y^I_{\epsilon,n}} \sum_{x\in B_{y}^c(\abs{y} ) }  F_{bg}( d(x,y)))\\
	& = \sum_{ y \in Y^I_{\epsilon,n}} \sum_{x\in B_{0}^c(\abs{y}) ) }  F_{bg}( \abs{x}) .
	\end{align*}
	Bounding the sums by integration and using the explicit form $F_{bg}(r) = e^{-bg(r)} F(r)$ along with Assumption \ref{asp1} gives,
	\begin{align*}
	\sum_{ y \in Y^I_{\epsilon,n}} \sum_{x\in B_{0}^c(\abs{y}) ) }  F_{bg}( \abs{x}) 
	& \leq V  \sum_{ y \in Y^I_{\epsilon,n}} F(\abs{y}) \int_{\abs{y}}^\infty  d r \  r^{\nu-1} e^{- b g(r)} \\
	& \leq  V K_{\nu - 1} \sum_{ y \in B_0^c( d(X,Y_{\epsilon,n}))} F(\abs{y}) \abs{y}^{l(\nu -1)} e^{ - b g(\abs{y}) }  \\
	& \leq V^2 K_{\nu-1} F(d(X,Y_{\epsilon,n})) \int_{d(X,Y_{\epsilon,n})}^\infty dr  r^{l(\nu - 1) +\nu - 1} e^{-b g(r)}  \\
	&\leq C_1  F(d(X,Y_{\epsilon,n}))   d(X,Y_{\epsilon,n})^{\tilde{l}(\nu)} e^{-b g(d(X,Y_{\epsilon,n}))} 
	\end{align*}
	where the constant $V$ is proportional to the volume of the unit sphere in $\RR^\nu$,
	$C_1 = V^2 K_{\nu-1}K_{l(\nu-1)+nu-1}$, and $\tilde{l}(\nu) = l(l(\nu - 1) + \nu-1)$.
	One can compute that $ n \sin( \alpha + \epsilon) \leq  d(X,Y) \leq  n \sin( \alpha + \epsilon)+2 $.  
	Thus,
	\begin{equation}\label{eqn:boundI}
	\sum_{y \in Y^I_{\epsilon,n}} \sum_{x\in X}  F_{bg} (d(x,y))  \leq C_1  F(d(X,Y_{\epsilon,n}))   d(X,Y_{\epsilon,n})^{\tilde{l}(\nu)} e^{-b g(d(X,Y_{\epsilon,n}))} 
	\end{equation}
	
	Next, we bound the second term in the sum.
	For each $y \in Y_{\epsilon,n}^{II}$, the inclusion $X \subset B_y^c(d(y,X))$ holds.
	It follows that
	\begin{align*}
	\sum_{y \in Y_{\epsilon,n}^{II}} \sum_{x \in X} F_{bg}( d(x,y)) & \leq \sum_{y \in Y_{\epsilon,n}^{II}} \sum_{x \in B^c_y(d(y,X) )} F_{bg}( d(x,y))\\
	& =  \sum_{y \in Y_{\epsilon,n}^{II}} \sum_{x \in B^c_0(d(y,X) )} F_{bg}( \abs{x}).
	\end{align*}
	
	In spherical coordinates, label $y = (r, \phi_1, \phi_2 \ldots, \phi_{\nu -1}) \in \RR^\nu$
	where $\phi_1$ labels the angle from the positive $x_1$-axis, $\phi_i \in [0,2 \pi)$ for $ 1 \leq i \leq \nu -2$ and $ \phi_{\nu-1} \in [0,\pi).$
	For convenience we will denote $\phi \equiv \phi_1$.
	If $y = (r, \phi, \ldots, \phi_{\nu -1}) \in Y_{\epsilon,n}^{II}$ then $\alpha+\epsilon < \phi \leq \pi/2 + \alpha$ 
	and for a fixed $\phi$ the values of the radius $r$ range from $ r_{\phi} < r < \infty$, 
	where 
	\[  r_{\phi} \equiv \frac{ n \sin(\alpha + \epsilon)}{\sin(\phi - \alpha - \epsilon)}.\] 
	For $y \in Y_{II}$ a simple calculation gives that $d(y, X) = r \sin(\phi - \alpha) $.
	Bounding the sum by integration gives,
	\begin{align*}
	& \sum_{y \in Y_{\epsilon,n}^{II}} \sum_{x \in B^c_0(d(y,X) )} F_{bg}( d(x,0)) \\
	& \qquad \qquad\leq V \sum_{y \in Y_{\epsilon,n}^{II}} F(d(y,X)) \int_{d(y,X)}^\infty d r  \  r^{\nu -1} e^{-b g(r)} \\
	& \qquad \qquad\leq  V K_{\nu-1}\sum_{y \in Y_{\epsilon,n}^{II}} F(d(y,X)) r^{l(\nu-1)} e^{- b g(d(y,X))} \\
	& \qquad \qquad\leq  V^2 K_{\nu-1} F(d(X,Y_{\epsilon,n})) \int_{\alpha+ \epsilon}^{\pi/2 + \alpha} d\phi \sin \phi
	\int_{r_{\phi}}^\infty  dr \ r^{l(\nu -1) +\nu -1} e^{-b g( r \sin(\phi - \alpha))} \\
	&\qquad \qquad\leq C_1 F(d(X,Y_{\epsilon,n})) \int_{\alpha+ \epsilon}^{\pi/2 + \alpha} d\phi \sin \phi \sin(\phi - \alpha)^{l(\nu-1) - \nu}
	r_\phi^{l(l(\nu -1) +\nu -1)} e^{-b g( r_\phi \sin(\phi - \alpha))}.
	\end{align*}
	For all $ \phi \in (\alpha+ \epsilon, \pi/2 + \alpha)$, the bound $ d(X,Y_{\epsilon,n})\geq  r_{\phi} \geq (n-1)\sin(\alpha + \epsilon) + \frac{ \sin(\alpha+\epsilon)}{\sin(\phi  - \alpha - \epsilon )}$ holds. 
	It follows from the properties of $g$ that
	\begin{equation*}
	e^{-b g(r_{\phi}\sin(\phi - \alpha))} \leq e^{ -b g( (n-1) \sin(\alpha + \epsilon) \sin(\phi - \alpha))}e^{- b g\left(\frac{ \sin(\alpha+\epsilon) \sin(\phi-\alpha)}{\sin(\phi  - \alpha - \epsilon )}\right) }.
	\end{equation*}
	Substituting back into the bound for $Y_{\epsilon,n}^{II}$ and simplifying gives
	\begin{align}\label{eqn:boundII}
	& \sum_{y \in Y_{\epsilon,n}^{II}} \sum_{x \in X} F_{bg}( d(x,y)) \\
	& \qquad \qquad \leq 	C_1 F(d(X,Y_{\epsilon,n})) d(X,Y_{\epsilon,n})^{\tilde{l}(\nu)} e^{ -b g(d(X,Y_{\epsilon,n}) \sin\epsilon)}
	\int_{0}^{\pi/2 - \epsilon} d\phi  \sin(\phi)^{-\nu} e^{- a \frac{ \sin(\alpha+\epsilon) \sin\epsilon}{\sin\phi} }
	\end{align}
	The $\phi$-integral is bounded for fixed $ 0< \epsilon < \pi/2$, is non-increasing in $\epsilon$, and only depends on $\nu$, $a$ and  $\alpha$.
	Combining the bounds \eqref{eqn:boundI} and \eqref{eqn:boundII} gives the result.
	
	The double sum in \eqref{eqn:LRdoubsum} is  symmetric in the interchange of $x$ and $y$
	and the opening angle for $Y_{\epsilon,n}$ is $ \beta = \pi - (\alpha+ \epsilon)$.
	Since $ 0\leq  \alpha \leq \pi/2 $ and $ 0<\epsilon< \pi/2 -\alpha $ this implies that $  \pi/2 < \beta < \pi$.
	Thus, exchanging the roles of $ X$ and $ Y_{\epsilon, n}$  gives the result.
\end{proof}

\begin{rem}
	Lemma \ref{lem:doubImpInt} can be extended to hold for more general graphs embedded in $\RR^\nu$
	which satisfy certain regularity conditions.
	For instance, 
	if $\sum_{x \in B_0^c(r)} F_{bg}(d(x,0)) \leq C \int_{r}^\infty dr r^{\nu-1} F_{bg}(r)$,
	we expect a similar bound to hold.
\end{rem}

\begin{rem}
	The double sum $\sum_{x\in X} \sum_{y\in Y_{\epsilon,n}} F_{bg} (d(x,y))$ is indeed divergent as $\epsilon \ra 0$ for all $n$ and $b>0$. 
	By comparing the sets $Y_{\epsilon, n}$ and $ Y_{0, n+2}$ one achieves the lower bound
	\begin{equation}\label{eqn:LRlowerbound}
	\left\lfloor\frac{3 \tan \alpha \tan (\alpha+ \epsilon)}{\tan(\alpha+\epsilon) - \tan\alpha} \right\rfloor \sum_{j= 0 }^\infty F_{bg}( j+ n+3)
	\leq \sum_{x\in X} \sum_{y\in Y_{\epsilon,n}} F_{bg} (d(x,y)),
	\end{equation}
	where the term $\left\lfloor\frac{3 \tan \alpha \tan (\alpha+ \epsilon)}{\tan(\alpha+\epsilon) - \tan\alpha} \right\rfloor   = O\left(\frac{1}{\epsilon}\right)$ is a lower bound for the number of lattice points in $ Y_{\epsilon,n}\setminus Y_{0,n+2}$.	
	This agrees with the upper bound found in Theorem \ref{thm:LRcone} as the $\phi$-integral in \eqref{eqn:boundII} diverges as $\epsilon \ra 0$.
	The lower bound in \eqref{eqn:LRlowerbound} and upper bound in Lemma \ref{lem:doubImpInt} are not expected to be tight.
\end{rem}

\begin{cor}\label{cor:fsatpolydecay}
	Suppose $g$ satisfies Assumption \ref{asp1}.  
	Let $X$ and $ Y_{\epsilon,n}$ be defined as in Lemma \ref{lem:doubImpInt}.
	Then, for all $0\leq \alpha < \pi/2$, $ 0< \epsilon< \pi/2 - \alpha$, $b>0$ and $k \in \NN$ we have that
	\begin{equation*}
	\lim_{n \ra \infty } \bigg[ n^{k} \sum_{x\in X} \sum_{y\in Y_{\epsilon,n}} F_{bg} (d(x,y)) \bigg] \ra 0.
	\end{equation*}
\end{cor}
\begin{proof}
	
	Comparing the bound in Lemma \ref{lem:doubImpInt} with the limit $\lim_{n \ra \infty} n^{k} e^{-b g(n)} = 0$
	for all $b>0$ and $k \in \NN$ gives the result.
\end{proof}

The main result in this section is the following Lieb-Robinson bound for cones.
Let $\Phi$ be an interaction map, that is, $\Phi: \mathcal{P}_0(\Gamma) \ra \calA_{loc}$ such that
$\Phi(X) \in \calA_X $ and $ \Phi(X)^* = \Phi(X)$ .
\begin{thm}\label{thm:LRcone}
	Suppose $g$ satisfies Assumption \ref{asp1} and 
	$\Phi$ satisfies a finite $F_{bg}$-norm for some $b>0$.  
	If $X$ and $ Y_{\epsilon,n}$ are  defined as in Lemma \ref{lem:doubImpInt}
	then there exists an affine function $\tilde{l}$ and $v_{bg}>0$ such that
	for all $0\leq \alpha < \pi$, $ 0< \epsilon< \pi - \alpha$,
	and observables $A\in \calA_X$ and $B\in \calA_{Y_{\epsilon,n}}$, 
	\begin{equation*}
	\| [ \tau_t(A),B] \|\leq 2 \|A \| \| B\| C_\epsilon d(X,Y_{\epsilon,n})^{\tilde{l}(\nu)} e^{v_{bg} \abs{t} - b g( d(X,Y)\sin\epsilon)},
	\end{equation*}
	where 
	\begin{equation*}
	n \sin(\alpha + \epsilon) \leq d(X,Y_{\epsilon,n}) \leq n \sin(\alpha + \epsilon)+2
	\end{equation*}
	and $C_\epsilon$ is non-increasing in $\epsilon$ and only depends on $\nu$, $b$ and $\alpha$.
\end{thm}

\begin{proof}
	Substituting the bound found in Lemma \ref{lem:doubImpInt} into the Lieb-Robinson bound \eqref{eqn:LRbound} gives the result.
\end{proof}

The quasi-locality of the dynamics $\tau_t$ can be interpreted by measuring the growth of the support of a time-evolved observable.
Let $\Lambda \in \mathcal{P}_0(\Gamma)$ be a finite  subset and take $X\subset \Lambda$.
For any observable $A \in \calA_\Lambda$ consider the conditional expectation
\begin{equation*}
\langle A \rangle_{X^c}  \equiv \int_{\mathcal{U}(X^c\cap \Lambda)} U^* A U \mu(dU),
\end{equation*}
where $\mu$ is the normalized Haar measure on the family of unitary operators $\mathcal{U}(X^c\cap\Lambda) \subset \calA_\Lambda$, that is,
$\int_{\mathcal{U}(X^c\cap\Lambda)}  \mu(dU)= 1$.
If $ A \in \calA_X$ then $A$ commutes with $\mathcal{U}(X^c)$ so that $\langle A\rangle_{X^c} = A$.
In particular, $\langle \langle A \rangle_{X^c} \rangle_{X^c}=\langle A \rangle_{X^c} $.

Let $ n>0$ and denote
$B_t(X,n) = \{ y \in \Gamma :  d(y,X) \leq v_{bg} \abs{t} + n \}.$
Suppose $ A \in \calA_X$.
Applying the Lieb-Robinson bound \eqref{eqn:LRbound}, it follows that 
\begin{align}
\| \tau_t(A) - \langle \tau_t(A) \rangle_{B^c_t(X,n)} \ \| 
&= \Big\| \int_{\mathcal{U}(B_t^c(X,n))} [ \tau_t(A), U] \  \mu(dU) \Big\| \\
&\leq  \sup_{U \in\mathcal{U}(B_t^c(X,n)) } \| [ \tau_t(A), U] \| \\
& \leq K_b \| A \| \abs{X} (1 - e^{ - v_b \abs{t}}) e^{- b n }. \label{eqn:fvQLdyn}
\end{align}

Now let $X\subset \ZZ^\nu$ be a potentially infinite subset, e.g., $X = \Lambda_\alpha$ an infinite cone.
Denote the following subsets of $\ZZ^\nu$ by 
\[ \Lambda_L \equiv [-L,L]^\nu \cap \ZZ^\nu, \quad  X_L  \equiv X \cap \Lambda_L, \quad \text{ and } \quad X_L^c \equiv \Lambda_L \setminus X_L. \]
\begin{defn}
	We define $ \langle \ \cdot \  \rangle_{X^c} : \calA \ra \calA_{X} \cong \calA_X \otimes I \subset \calA$ by
	\begin{equation}\label{eqn:conePT}
	\langle A \rangle_{X^c} \equiv \lim_{L \ra \infty} \langle A \rangle_{X_L^c}
	= \lim_{L\ra \infty} \int_{\mathcal{U}(X_L^c)} U^* A U \mu(dU) \quad \text{ for all } \quad A \in \calA,
	\end{equation}
	where the limit is in the norm sense.
\end{defn}

\begin{lemma}\label{lem:coneProj}
	The operator $\langle A \rangle_{X^c}$ is well defined; the limit \eqref{eqn:conePT} exists and is unique, and $\langle A \rangle_{X^c} \in \calA_{X}$.
\end{lemma}

\begin{proof}
	Let $\epsilon>0$ be given.  
	By density of $\calA_{loc}$ in $\calA$ there exists $L>0$ and an operator $A_L \in \calA_{\Lambda_L}$ such that 
	\[ \| A - A_L \| < \epsilon. \]
	For any $n,m > L$ we have that 
	\begin{align*}
	\| \langle A \rangle_{X_m^c} - \langle A \rangle_{X_n^c} \| 
	& \leq \| \langle A \rangle_{X_m^c} - \langle A_L \rangle_{X_m^c} \| 
	+  \| \langle A \rangle_{X_n^c} - \langle A_L \rangle_{X_n^c} \| 
	+ \| \langle A_L \rangle_{X_m^c} - \langle A_L \rangle_{X_n^c}\|\\
	& = \left\| \int_{\mathcal{U}(X_m^c)} U^* (A - A_L) U \mu (dU) \right\|
	+ \left\| \int_{\mathcal{U}(X_n^c)} U^* (A - A_L) U \mu (dU) \right\|\\
	& \leq 2 \epsilon,
	\end{align*}
	where we use that $\langle A_L \rangle_{X_m^c} = \langle A_L \rangle_{X_n^c} =  A_L$ for the last term. 
	Thus, the sequence $ \langle A \rangle_{X_L^c}$ is Cauchy in $\calA_X$ and its limit is defined  to be the operator $\langle A \rangle_{X^c} \in \calA_X$.
\end{proof}

\begin{cor}\label{cor:quasiloc}
	Suppose $g$ satisfies Assumption \ref{asp1} and $\Phi$ has a finite $F_{bg}$-norm for $b>0$.  
	Let  $X$ and $ Y_{\epsilon,n}$ as from Lemma \ref{lem:doubImpInt}.
	Denote $ Y_{\epsilon,n}(t) =  \Lambda_{\alpha+ \epsilon}^c(\lceil \frac{v_a \abs{t}}{a} +  n \rceil)$.
	Then, there exists an affine function $\tilde{l}$ such that 
	for all $0\leq \alpha < \pi$, $ 0< \epsilon< \pi - \alpha$,
	and $ A \in \calA_{X}$,
	\begin{equation*}
	\| \langle \tau_t(A) \rangle_{Y_{\epsilon,n }(t)} - \tau_t(A) \| \leq C_\epsilon \| A\|  d(X,Y_{\epsilon,n}(t))^{\tilde{l}(\nu)} e^{- b g( d(X,Y_{\epsilon,n}(t)) \sin\epsilon )},
	\end{equation*}
	where 
	\[ \Big\lceil v_{bg} \abs{t} +  n \Big\rceil \sin(\alpha + \epsilon)
	\leq d(X,Y_{\epsilon,n}(t)) 
	\leq \Big\lceil v_{bg} \abs{t}+  n \Big\rceil  \sin(\alpha + \epsilon) + 2, \]
	and $C_\epsilon$ is non-increasing in $\epsilon$ and only depends on $\nu$, $b$ and $\alpha$.
	In particular, 
	\begin{equation}\label{eqn:coneql}
	\lim_{n\ra \infty}  n^{k} \| \langle \tau_t(A) \rangle_{Y_{\epsilon,n}(t)} - \tau_t(A) \| = 0 \quad \text{ for all } \quad k\in \NN.
	\end{equation}
\end{cor}

\begin{proof}
	The proof is similar to the argument given for the finite volume case \eqref{eqn:fvQLdyn} where we take the limit as in Lemma \ref{lem:coneProj}
	and use the improved bound from Theorem \ref{thm:LRcone} to avoid the dependence on the size of the support of $A$.
	Equation \eqref{eqn:coneql} comes from a  similar argument to the proof of Corollary \ref{cor:fsatpolydecay}.
\end{proof}

\subsection{Stability}

Let $\tau_t$ be the quasi-local dynamics corresponding the an interaction $\Phi$ with a finite $F_{bg}$-norm for some $b>0$.
We will consider endomorphisms of the form $\tau_t^{-1}\circ \rho \circ \tau_t$, and in a later section, 
argue why these correspond to charge generators in a perturbed system.
The main result of this section is that the evolution by a quasi-local dynamics will preserve the defining property of almost-localized endomorphism. 
The main tool in the proof will be the Lieb-Robinson bound for cones that was established in the previous section.

Suppose $g$ satisfies Assumption \ref{asp1}.
Then, there exist an $n_0 >$ such that for constants $r, b, a >0$ and $C_\epsilon$ non-increasing in $\epsilon$, the function
\begin{equation*}
h_\epsilon(n) \equiv 
\Bigg\{ \begin{array}{ll}
C_\epsilon n_0^r e^{ -b g( a n_0 ) \sin\epsilon} &\mbox{ if } n \leq n_0\\
C_\epsilon n^r e^{ -b g( a n ) \sin\epsilon}  & \mbox{ if } n > n_0
\end{array}  
\end{equation*}
is a member of the class of decay functions $\mathcal{F}_\infty$. 
By Corollary \ref{cor:quasiloc}, there are constants $r,b,a>0$  such that 
\begin{equation}\label{eqn:quasiloc}
\| \langle \tau_t(A) \rangle_{Y_{\epsilon,n }(t)} - \tau_t(A) \| \leq h_\epsilon(n) \|A\|.
\end{equation}

\begin{lemma}\label{lem:stabaloc}
	Suppose $g$ satisfies Assumption \ref{asp1} and $\Phi$ satisfies a finite  $F_{bg}$-norm for some $b>0$.
	If $ \rho$  is an almost-localized endomorphism in $\Lambda_\alpha$
	with decay function $f$ 
	then for all $ t \in \RR$, $ \tau_t^{-1} \circ \rho \circ \tau_t$ is an
	almost-localized endomorphism in $\Lambda_{\alpha}$
	with decay function $f_{\epsilon/2}(n/2) + 2 h_{\epsilon/2}(v_{bg} \abs{t}+ n/2)$.
\end{lemma}

\begin{proof}
	Let $\epsilon >0$ be given.
	By Corollary \ref{cor:quasiloc} and \eqref{eqn:quasiloc} we have that there exists constants $C, r >0$ such that  if $n$ is even then
	\begin{align*}
	\sup_{A \in \calA(\Lambda^c_{\alpha+\epsilon}-n)} \frac{  \| \tau_t^{-1} \circ \rho \circ \tau_t (A) - A \| }{\|A\|}
	&\leq  \sup_{A \in \calA(\Lambda^c_{\alpha+\epsilon}-n)}  \frac{ \| \rho( \langle \tau_t(A) \rangle) -  \langle \tau_t(A) \rangle \|}{\|A\|}
	+ 2 \frac{\|  \langle \tau_t(A) \rangle - \tau_t(A)\| }{\|A\|}\\
	& \leq f_{\epsilon/2}(n/2) +2 h_{\epsilon/2}(v_{bg} \abs{t}+ n/2).
	\end{align*}
	where $\langle \ \cdot \  \rangle = \langle \ \cdot \ \rangle_{ \Lambda_{ \alpha + \epsilon/2} + n/2 }$
\end{proof}

\begin{thm}\label{thm:stabsectorstructure}
	Let  $\Delta$ is a semi-group of almost localized and transportable endomorphisms satisfying the superselection criterion \ref{def:superselection} for a vacuum state $\omega_0$ satisfying Assumption \ref{ass:hdsp}.
	Suppose $g$ satisfies Assumption \ref{asp1} and $\Phi$ satisfies a finite  $F_{bg}$-norm for some $b>0$.
	Then, for all $ t\in[0,1]$ the semi-group $\tau_t^{-1} \circ \Delta\circ \tau_t $ satisfies the superselection criterion
	for the vacuum state $\omega_0 \circ \tau_t$.
	Furthermore, for all $\rho, \sigma \in \Delta$ we have that
	\begin{equation} \label{eqn:stabintertwiner}
	( \rho, \sigma)_{\pi} = (\tau_t^{-1}\circ \rho \circ \tau_t, \tau_t^{-1} \circ \sigma \circ \tau_t)_{\pi\circ \tau_t}.
	\end{equation}
	
	If in addition, $\Delta$ satisfies the assumptions of Theorem \ref{thm:biasymptopia}, 
	then  $\tau_t^{-1} \circ \Delta \circ \tau_t$ with vacuum state $\omega_0 \circ \tau_t$ is a braided tensor $C^*$-category
	and is braided equivalent to $\Delta$.
\end{thm}

\begin{proof}
	Let $\rho\in \Delta$ be almost localized in a cone $\Lambda_{ \alpha }$. 
	Then, by Lemma \ref{lem:stabaloc}, $\tau_t\circ \rho \circ \tau_t$ is almost localized in $\Lambda_{ \alpha }$.
	Thus, the first part of the superselection criterion is satisfied.
	
	Let $T \in (\rho,\sigma)_\pi$.
	It follows that $T\in (\rho \circ \tau_t, \sigma\circ \tau_t)_\pi$ and 
	\begin{align*}
	T \pi\circ \alpha \circ  \alpha^{-1}  \rho  \sigma \circ \tau_t (A)  = T  \pi \circ  \sigma \rho \tau_t (A) 
	=  \pi \circ \sigma\circ \tau_t(A) T
	= \pi\circ \alpha \circ  \alpha^{-1}  \sigma  \sigma \circ \tau_t (A) T.
	\end{align*}
	Thus, $T \in (\tau_t^{-1}\circ \rho \circ \tau_t, \tau_t^{-1}\sigma\circ \tau_t)_{\pi \circ \tau_t}$.
	A similar argument shows that, $(\rho,\sigma)_\pi = (\tau_t^{-1}\circ \rho \circ \tau_t, \tau_t^{-1}\sigma\circ \tau_t)_{\pi \circ \tau_t}$.
	
	Let $\rho_t \in \tau_t^{-1} \circ \Delta\circ \tau_t$ and write $ \rho_t = \tau_t^{-1} \circ \rho \circ \tau_t$ for some $ \rho \in \Delta$.
	Let $\Lambda_{\beta}' \in \mathcal{C}$ be a cone. 
	Since $ \rho$ is transportable there is a $ \rho'$ almost localized in $\Lambda_{ \beta}'$ and a unitary $U$ such that $U  \pi\circ \rho(A) = \pi\circ\rho'(A) U$ for all $A\in\calA$.
	Let $\rho_t' \equiv \tau_t^{-1} \circ \rho' \circ \tau_t$. 
	By Lemma \ref{lem:stabaloc} $\rho_t'$ is almost localized in $\Lambda_{ \beta}'$
	and by \eqref{eqn:stabintertwiner} we have that $( \rho, \rho')_\pi = (\rho_t, \rho_t')_{\pi\circ \tau_t}$.
	Therefore, $\sigma$ is transportable for $\omega\circ \tau_t$.	
	
	If in addition, $\Delta$ satisfies the assumptions of Theorem \ref{thm:biasymptopia} then $\Delta$ is a braided tensor $C^*$-category.
	Thus, $\tau_t^{-1} \circ \Delta \tau_t$ is a braided tensor $C^*$-category with braiding defined by $\epsilon_{\tau_t^{-1} \circ \rho \circ \tau_t , \tau_t^{-1}\circ \sigma \circ \tau_t} \equiv \epsilon_{ \rho, \sigma}$.
	Define the functor $F(\rho) = \tau^{-1}_t \circ \rho \tau_t$.  It follows that 
	\begin{align*}
	F (\rho \otimes \sigma ) &= \tau_t^{-1} \circ \rho \circ \sigma \circ \tau_t 
	= \tau_t^{-1} \circ \rho \circ \tau_t \circ \tau_t^{-1} \circ \sigma \circ \tau_t\\
	&= F(\rho) \otimes F(\sigma).
	\end{align*}
	Thus, $F$ gives a braided equivalence of tensor $C^*$-categories.
\end{proof}

%\section{Application}
%
%Let $\omega_0$ and $ \omega_1$ be gapped ground states for a quantum spin system with quasi-local algebra $\calA$.
%Then, it is commonly said that $\omega_0$ and $\omega_1$ are in the same \emph{gapped ground state phase}
%if there exists a continuous family of Hamiltonians $H(s)$ with finite range interactions such that 
%$H(s)$ has a non-vanishing spectral gap above the ground state uniform in $s$ and 
%$ \omega_0$ and $ \omega_1$ are ground states of $H(0)$ and $H(1)$ \cite{BachmannMNS}.
%
%A current challenge in mathematical physics is the classification of gapped ground state
%phases \cite{BachmannMNS,BachmannO,Ogata1,Ogata2,Ogata3}.
%One approach to classifying a phase is to construct a complete set of invariants.
%By definition, an invariant is a quantity that is constant within a phase.
%Consequently, if an invariant is computed for two systems and is found to take different values,
%the systems must be in different phases. 
%The construction of invariants can be expected to rely on the existence of a spectral gap. 
%The superselection structure of quasi-particle excitations is expected to be an invariant in a gapped ground state phase.
%
%In this section, we consider Kitaev's the quantum double model under local perturbations.
%We will use explicit construction of quasi-particle states as low-energy states in a gapped part of the spectrum \cite{ChaNN,FiedlerN},
% the stability of the spectral gap \cite{BravyiHM} and the spectral flow dynamics \cite{BachmannMNS} to show the stability of anyons under local perturbations.

\section{Stability of anyons in Kitaev's abelian quantum double models}

We consider a family of quantum double models for abelian groups as defined by Kitaev \cite{KitaevQD},
see Section \ref{sec:qdoub}.
The statistics of the quasi-particle excitations of the model are encoded in the modular tensor category $\operatorname{Rep}(\mathcal{D}(G))$
and thus, are anyons.
There is a spectral gap above the ground state that is stable under local perturbations \cite{AlickiFH,BravyiHM}.
In \cite{KitaevQD}, the ground state space of the model was proposed to serve as a fault-tolerant quantum code,
where quantum computation is implemented by the braiding, fusion and measurement of anyon excitations. 
A premise of this approach to fault-tolerant quantum computation 
is the stability of the spectral gap and anyon structure against local perturbations. 
In this section, we combine the superselection criterion and stability results of the previous section with the techniques of spectral flow \cite{BachmannMNS} to prove the stability of anyons in the abelian quantum double models.
Our proof relies heavily on the complete classification of low energy states in terms of single excitation states for the infinite system \cite{ChaNN}.
Although the planar quantum double models cannot store quantum information at finite temperature \cite{AlickiFH2}, 
we believe thorough analysis of the exactly solvable topological models may give insights on further possibilities for topological quantum computation.

The elementary excitations have an energy bounded by four as shown in \eqref{eqn:ribenergy}. 
We consider the set of all states in the infinite systems that have the same energy threshold as follows.
Let $\mathcal{S}_L^{qd}$ be the set of mixtures of eigenstates of $H_L^{per}$  with energy in $[0, 4]$
and $\mathcal{S}^{qd}$ be the set of all weak$^*$ limit points of the sets $S^{qd}_L$.
The following classification theorem follows directly from Lemma \ref{lem:gshambound} and Theorem \ref{thm:eqstates}.

\begin{thm}(\cite{ChaNN})\label{thm:aloctoloc}
	The state $\omega_s^{\chi,c} \in \mathcal{S}^{qd}$ for all $(\chi,c) \in \widehat{G}\times G$ and sites $s$.
	Furthermore, if $\omega \in \mathcal{S}^{qd}$ is a pure state then $\omega \cong \omega_s^{\chi,c}$ for some $(\chi,c) \in \widehat{G} \times G$.
\end{thm}

The energy threshold for $\mathcal{S}^{qd}$ was set to guarantee that each single excitation state is represented as a state in the  $\mathcal{S}^{qd}$.
In principle, this energy threshold could be increased without changing the sector theory.

\begin{defn}
	Define the $C^*$-category of endomorphisms $\Delta^{qd}$ with vacuum $\omega_0$ 
	whose objects are $*$-endomorphism $\rho \in \Delta^{qd}$  satisfying
	\begin{enumerate}
		\item $ \omega_0 \circ \rho \cong \omega$ for some $ \omega \in \mathcal{S}^{qd}$,
		\item $\rho$ is almost localized and transportable (see~\ref{def:superselection}) and further satisfies the assumptions of Theorem \ref{thm:biasymptopia},
	\end{enumerate}
	and arrows are the intertwiner space $(\rho,\sigma)$ for each $ \rho,\sigma \in \Delta^{qd}$.
\end{defn}

\begin{thm}\label{thm:stabqd}
	The category $\Delta^{qd}$ is a braided tensor $C^*$-category
	and is braided tensor equivalent to the category of finite dimensional representations of the quantum double of $G$, $\operatorname{Rep}(\mathcal{D}(G))$.
\end{thm}

\begin{proof}
	From Proposition \ref{prop:qdffgs}, the vacuum state $\omega_0$ is a pure state satisfying Assumption \ref{ass:hdsp}.
	By construction, we may apply Theorem \ref{thm:biasymptopia} directly to $\Delta^{qd}$ to achieve a braided tensor $C^*$-category
	structure.
	
	We now show that $\Delta^{qd}$ is identical to the category constructed in \cite{FiedlerN, Naaijkens11}, that is, 
	transportable endomorphisms exactly localized in a fixed cone $\Lambda$ and  satisfying \eqref{eqn:conecrit} for $\pi = \pi_0 \circ \rho$.
	Let $ \rho \in \Delta^{qd}$ be such that $\pi_0 \circ \rho$ is irreducible. 
	Then, by Theorem \ref{thm:aloctoloc} we have that $\pi_0 \circ \rho \cong \pi_s^{\chi,c}$ for some site $s$ and $(\chi,c) \in \widehat{G}\times G$.
	By Proposition \ref{prop:singleexc}, $\pi_s^{\chi,c}$ satisfies the criterion \eqref{eqn:conecrit} from which it follows that 
	$\pi_0 \upharpoonright \calA_{\Lambda^c} \cong \pi_s^{\chi,c}  \upharpoonright \calA_{ \Lambda^c}\cong  \pi_0 \circ \rho \upharpoonright \calA_{ \Lambda^c}.$
	For each site $s$ of $\Lambda$ and $(\chi,c) \in \widehat{G} \times G$, consider the representation $\pi_s^{\chi,c} = \pi_0 \circ \rho^{\chi,c}_s$.
	By path independence \ref{eqn:ribbonpathind}, the state $ \omega_0 \circ \rho^{\chi,c}_s$ can be constructed for an infinite ribbon supported entirely in $\Lambda$.
	It follows that $ \rho^{\chi,c}_s \in \Delta^{qd}$.
	Therefore, the irreducible objects of $\Delta^{qd}$ and the irreducible representations satisfying \eqref{eqn:conecrit} are exactly equal.
	
	Let $\rho^{(\chi,c)}, \sigma^{\psi,d} \in \Delta^{qd}$ be irreducible objects localized in a cone $\Lambda$ 
	and for shorthand write $\rho = \rho^{(\chi,c)}$ and $ \sigma=\sigma^{\psi,d} $.
	The bi-asymptopia as defined in Theorem \ref{thm:biasymptopia}
	can be chosen such that $\Lambda^{\mathcal{U}}_{\alpha + \epsilon} \ll \Lambda$
	and  $ \Lambda^{\mathcal{V}}_{\beta + \epsilon} \ll \Lambda$.
	Recall the braiding $\epsilon_{\rho,\sigma} \in (\rho\circ \sigma, \sigma\circ \rho) $
	is defined by \eqref{eqn:braid} where for some $*$-endomorphisms $\rho_m$ and $\sigma_n$ exactly localized in cones $ \Lambda^{\mathcal{U}}_{\alpha + \epsilon}+ m$ and $\Lambda^{\mathcal{V}}_{\beta + \epsilon} \ll \Lambda + n$, respectively, we have $U_m \in (\rho, \rho_m)$ and $V_n \in (\rho, \rho_n)$.
	By an argument similar to Lemma 3.2 of \cite{NaaijkensKL}, the intertwiners 
	$U_m \in \mathcal{R}(\Lambda) \otimes \mathcal{R}(\Lambda^{\mathcal{U}}_{\alpha+\epsilon}+m)$.
	It follows that $\sigma_n(U_m) = U_m$.
	Therefore, for any fixed $k$ an extension of Lemma 4.8 of \cite{Naaijkens11} gives
	\begin{align*}
	\epsilon_{\rho,\sigma}  &= \lim_{m,n\ra \infty} V_n^* \sigma_n(U_m^*) U_m \rho(V_n) 
	= \lim_{n \ra \infty} V_n^* \rho(V_n) \\
	& = V_k^* \rho(V_k).
	\end{align*}
	Therefore, by Theorem 6.3 of \cite{Naaijkens11} and its generalization to finite abelian groups in \cite{FiedlerN}, we have a braided tensor equivalence of categories $\Delta^{qd} \ra \operatorname{Rep}(\mathcal{D}(G))$.
\end{proof}

Consider perturbations of the quantum double models with periodic boundary conditions of the form \ref{eqn:perturbationlinear}.
Let $I_0 = [0,4]$ and $ I_1=[5,\infty)$.
By Proposition \ref{prop:singleexc}, $\operatorname{spec} H_L^{per} \subset I_0 \cup I_1$.
Combining the Local Topological Quantum Order property \ref{cor:qdLTQO} and the stability of frustration-free Hamiltonians Theorem \ref{thm:stablegap} \cite{BravyiHM}, for some $\epsilon>0$  we have that for  all $ 0\leq s <\epsilon$ 
there are intervals $I_0(s)$ and $ I_1(S)$ with endpoint depending continuously on $s$ 
such that $I_0(0) = I_0$ and $ I_1(0) = I_1$, 
there is a $\gamma>0$ such that $ d( I_0(s), I_1(s) ) >\gamma$ and 
the spectrum of $H_L(s)$ splits into two disjoint sets $spec(H_L(s)) = \Sigma_L^0(s) \cup \Sigma_L^1(s)$ with 
$ \Sigma_L^0(s) \subset I_0(s)$ and $\Sigma_L^1(s) \subset I_L^1(s)$ for all $L$.
By Theorem \ref{thm:infvoldyn}, for each $s$ the dynamics defined by $H_L(s)$ will satisfy a Lieb-Robinson bound.

Consider the set of elementary excitations for the perturbed quantum double model.
Naturally, these correspond to states with energy supported in the interval $I_0(s)$.
Let $\mathcal{S}_L^{qd}(s)$ be the set of mixtures of eigenstates with energy in $I_0(s)$
and $\mathcal{S}^{qd}(s)$ be the set of all weak$^*$ limit points of the sets $S^{qd}_L(s)$.
Recall that by Theorem \ref{thm:autoeq}, the spectral flow dynamics $\alpha_s$ satisfies $\mathcal{S}^{qd}(s) = \mathcal{S}^{qd} \circ \alpha_s$.

Define the $C^*$-category $\Delta^{qd}(s)$ of $*$-endomorphisms for vacuum state $\omega_0 \circ \alpha_s$ 
to have objects $\rho$ satisfying
\begin{enumerate}
	\item $\omega_0 \circ \alpha_s\circ   \rho  \cong \omega$ for some $ \omega \in \mathcal{S}(s)$, 
	\item $\rho$ is almost localized and transportable with respect to $\omega_0 \circ \alpha_s$ and further satisfies assumptions of Theorem \ref{thm:biasymptopia}.
\end{enumerate}

\begin{thm}\label{thm:stabqds}
	The category $\Delta^{qd}(s)$ is a braided tensor $C^*$-category
	and is braided tensor equivalent to the category of finite dimensional representations of the quantum double of $G$, $\operatorname{Rep}(\mathcal{D}(G))$.
\end{thm}

\begin{proof}
	By construction, $\Delta^{qd}(0) = \Delta^{qd}$ and thus, is braided equivalent to $\operatorname{Rep}(\mathcal{D}(G))$.
	
	Now consider the $C^*$-category $\alpha_{s}^{-1} \circ \Delta^{qd} \circ  \alpha_s$ with vacuum state $\omega_0 \circ \alpha_s$.
	Theorems \ref{thm:stabsectorstructure} and \ref{thm:autoeq}, give $\alpha_{s}^{-1} \circ \Delta^{qd} \circ  \alpha_s \subset \Delta^{qd}(s)$.
	We claim that $\Delta^{qd}(s) = \alpha_{s}^{-1} \circ \Delta^{qd}\circ  \alpha_s$.
	For each irreducible $ \rho \in \Delta^{qd}(s)$ there is a pure state $\omega(s) \in \mathcal{S}(s)$ 
	such that $ \omega_0 \circ \alpha_s \circ \rho \cong \omega(s)$.
	By Theorem \ref{thm:autoeq}, $\omega(s) = \omega \circ \alpha_s$ for $ \omega \in \mathcal{S}$.
	The state $\omega$ must be a pure state since $\alpha_s$ is an automorphism.
	By purity and Theorem \ref{prop:singleexc}, $\omega \cong \omega^{\chi,c}_x$ for some single excitation state.
	It follows that 
	\begin{align*}
	\omega_0 \circ \alpha_s \circ \rho & \cong \omega(s)  = \omega \circ \alpha_s \cong \omega^{\chi,c}_x \circ \alpha_s \\
	& = \omega_0 \circ \alpha_s \circ \alpha_s^{-1} \circ \rho^{\chi,c}_x \circ \alpha_s,
	\end{align*}
	leading to the equivalence $\rho \cong \alpha_s^{-1} \circ \rho^{\chi,c}_x \circ \alpha_s$
	from which the claim follows.
	
	By Theorem \ref{thm:stabsectorstructure}, $\alpha_{s}^{-1} \circ \Delta^{qd}\circ  \alpha_s $ is braided equivalent to $\Delta^{qd}(0)$.
	The result follows from Theorem \ref{thm:stabqd}.
\end{proof}

   \backmatter
   
\newcommand{\etalchar}[1]{$^{#1}$}
\providecommand{\bysame}{\leavevmode\hbox to3em{\hrulefill}\thinspace}
\providecommand{\MR}{\relax\ifhmode\unskip\space\fi MR }
\providecommand{\MRhref}[2]{%
	\href{http://www.ams.org/mathscinet-getitem?mr=#1}{#2}
}
\providecommand{\href}[2]{#2}

\end{document}